\documentclass[10pt]{article}

\usepackage{amsfonts,amsmath,amssymb,amsthm}
\usepackage[pdftex]{graphicx}
\usepackage[hmargin=1in,vmargin=1in]{geometry}
\usepackage{natbib,setspace,enumerate,mathtools}
\usepackage{multirow}
\usepackage{booktabs}
\allowdisplaybreaks
\usepackage[toc,title,titletoc,header]{appendix}
\usepackage{adjustbox}
\usepackage{etoolbox} 
\usepackage{threeparttable}
\usepackage[colorlinks,citecolor=blue]{hyperref}
\def\qed{\rule{2mm}{2mm}}
\def\independent{\perp \!\!\! \perp}

\newtheorem{theorem}{Theorem}[section]
\newtheorem{lemma}{Lemma}[section]

\newtheorem{proposition}{Proposition}[section]
\theoremstyle{definition}
\newtheorem{example}{Example}[section]
\newtheorem{remark}{Remark}[section]
\newtheorem{assumption}{Assumption}[section]

\AtEndEnvironment{remark}{~\qed}
\AtEndEnvironment{example}{~\qed}

\DeclareMathOperator*{\var}{Var}
\DeclareMathOperator*{\cov}{Cov}
\DeclareMathOperator*{\diag}{diag}

\newcommand{\mycomment}[1]{}

\begin{document}

\author{
Yuehao Bai \\
Department of Economics\\
University of Southern California \\
\url{yuehao.bai@usc.edu}
\and
Jizhou Liu \\
Booth School of Business\\
University of Chicago\\
\url{jliu32@chicagobooth.edu}
\and
Max Tabord-Meehan\\
Department of Economics\\
University of Chicago \\
\url{maxtm@uchicago.edu}
}

\bigskip

\title{Inference for Matched Tuples and Fully Blocked Factorial Designs \thanks{We thank the editor and anonymous referees, as well as seminar participants at Columbia University, Duke University, Indiana University, Michigan State University, Penn State University, UCLA, University of Pennsylvania, University of Pittsburgh, University of Southern California, UW Milwaukee, and Yale University for helpful comments. We thank Jiehan Xu for excellent research assistance. The third author acknowledges support from NSF grant SES-2149408.}}

\maketitle

\vspace{-0.3in}

\begin{spacing}{1.2}
\begin{abstract}
This paper studies inference in randomized controlled trials with multiple treatments, where treatment status is determined according to a ``matched tuples” design. If there are $|\mathcal{D}|$ possible treatments, then by a matched tuples design, we mean an experimental design where units are sampled i.i.d.\ from the population of interest, grouped into ``homogeneous” blocks of size $|\mathcal{D}|$, and finally, within each block, exactly one individual is randomly assigned to each of the $|\mathcal{D}|$ treatments. We first study estimation and inference for matched tuples designs in the general setting where the parameter of interest is a vector of linear contrasts over the collection of average potential outcomes for each treatment. Parameters of this form include standard average treatment effects used to compare one treatment relative to another, but also include parameters which may be of interest in the analysis of factorial designs. We first establish conditions under which a sample analogue estimator is asymptotically normal and construct a consistent estimator of its corresponding asymptotic variance. Combining these results establishes the asymptotic exactness of tests based on these estimators. In contrast, we show that, for two common testing procedures based on $t$-tests constructed from linear regressions, one test is generally conservative while the other is generally invalid. We go on to apply our results to study the asymptotic properties of what we call ``fully-blocked" $2^K$ factorial designs, which are simply matched tuples designs applied to a full factorial experiment. Leveraging our previous results, we establish that our estimator achieves a lower asymptotic variance under the fully-blocked design than that under any stratified factorial design which stratifies the experimental sample into a finite number of ``large" strata. A simulation study and empirical application illustrate the practical relevance of our results. 
\end{abstract}
\end{spacing}

\noindent KEYWORDS: Randomized controlled trials, matched tuples, matched pairs, multiple treatments, factorial designs

\noindent JEL classification codes: C12, C14
\hypersetup{pageanchor=false}
\thispagestyle{empty} 
\newpage
\hypersetup{pageanchor=true}
\setcounter{page}{1}

\section{Introduction}

This paper studies inference in randomized controlled trials with multiple treatments, where treatment status is determined according to a ``matched tuples” design. If there are $|\mathcal{D}|$ possible treatments, then by a matched tuples design, we mean an experimental design where units are sampled i.i.d.\ from the population of interest, grouped into ``homogeneous” blocks of size $|\mathcal{D}|$, and finally, within each block, exactly one individual is randomly assigned to each of the $|\mathcal{D}|$ treatments. As such, matched tuples designs generalize the concept of matched pairs designs to settings with more than two treatments. Matched tuples designs are commonly used in the social sciences: see \cite{bold2018experimental}, \cite{brown2020inducing}, \cite{McKenzie2013}, and \cite{McKenzie2014} for examples in economics, and are often motivated using the simulation evidence presented in \cite{bruhn2009pursuit}. However, we are not aware of any formal results which establish valid asymptotically exact methods of inference for matched tuples designs. Accordingly, in this paper we establish general results about estimation and inference for matched tuples designs, and then apply these results to study the asymptotic properties of what we call ``fully-blocked” $2^K$ factorial designs.

We first study estimation and inference for matched tuples designs in the general setting where the parameter of interest is a vector of linear contrasts over the collection of average outcomes for each treatment. Parameters of this form include standard average treatment effects (ATEs) used to compare one treatment relative to another, but as we explain below also include more complicated parameters which may be of interest, for instance, in the analysis of factorial designs. We first establish conditions under which a sample analogue estimator is asymptotically normal and construct a consistent estimator of its corresponding asymptotic variance. Combining these results establishes the asymptotic validity of tests based on these estimators. We then consider the asymptotic properties of two commonly recommended inference procedures. The first is based on a linear regression with block fixed effects. Importantly, we find the $t$-test based on such a regression is in general not valid for testing the null hypothesis that a pairwise ATE is equal to a prespecified value. The second is based on a linear regression with cluster-robust standard errors, where clusters are defined at the block level. Here we find that the corresponding $t$-test is generally valid but conservative, and that this conservativeness increases in the number of treatments.

Next, we apply our results to study the asymptotic properties of ``fully-blocked” $2^K$ factorial designs. Factorial designs are classical experimental designs \citep[see][for a textbook treatment]{wu2011experiments} which are increasingly being used in the social sciences \citep[see for instance][]{alatas2012targeting, besedevs2012age, dellavigna2016voting, kaur2015self, karlan2014agricultural}. In a $2^K$ factorial design, each treatment is a combination of multiple ``factors,” where each factor can take two distinct values, or ``levels.” As a consequence, a full $2^K$ factorial design can be thought of as a randomized experiment with $2^K$ distinct treatments (importantly however, the analysis of factorial designs typically considers factorial effects as the parameters of interest: see Section \ref{sec:factorial} for a definition). A fully-blocked factorial design is then simply a matched tuples design with blocks of size $2^K$. Leveraging our previous results, we establish that our estimator achieves a lower asymptotic variance under the fully-blocked design than under any stratified factorial design which stratifies the experimental sample into a finite number of ``large” strata (such designs include complete randomization as a special case). We also consider settings where only one factor may be of primary interest, and establish that even in such cases it is more efficient to perform a fully-blocked design than to perform a matched pairs design which exclusively focuses on the primary factor of interest. 

In a simulation study, we find that although our inference results are asymptotically exact, our proposed tests may be conservative in finite samples when the experiment features many treatments or many blocking variables. Accordingly, we also study the behavior of a matched tuples design with ``replicates,” where we form blocks of size \emph{two} times the number of treatments, and each treatment is assigned exactly \emph{twice} at random within each block. Although we find that such a design results in an estimator with slightly larger mean-squared error, the rejection probabilities of our proposed tests become much closer to the nominal level, which may result in improved power. Further discussion is provided in Section \ref{sec:replicate} below.

Although the analysis of matched tuples designs has to our knowledge not received much attention, there are large literatures on both the analysis of matched pairs designs and the analysis of factorial designs. Recent papers which have analyzed the properties of matched pairs designs include \cite{athey2017econometrics}, \cite{bai2021inference}, \cite{bai2022optimality}, \cite{chaisemartin2022at}, \cite{cytrynbaum2021designing}, \cite{imai2009essential}, \cite{jiang2020bootstrap}, \cite{fogarty2018regression}, and \cite{van2012adaptive}. Our analysis builds directly on the framework developed in \cite{bai2021inference}, and our Theorems \ref{thm:main_delta} and \ref{thm:V_const} nest some of their results when specialized to the setting of a binary treatment. \cite{cytrynbaum2021designing} considers a generalization of matched pairs designs, a special case of which he refers to as a matched tuples design. However, his design groups units into homogeneous blocks in order to assign a binary treatment with unequal treatment fractions. In contrast, we consider a setting where units are grouped into homogeneous blocks in order to assign multiple treatments.   

Recent papers which have analyzed factorial designs include \cite{rubin2016}, \cite{dasgupta2015causal}, \cite{li2020rerandomization}, \cite{muralidharan2019factorial}, \cite{pashley2019causal}, and \cite{liu2022randomization}. Our setup and notation for $2^K$ factorial designs mirrors the framework introduced in \cite{dasgupta2015causal}, although our setup differs in that we consider a ``super-population” framework where potential outcomes are modeled as random, whereas they maintain a finite population framework where potential outcomes are modeled as fixed.\footnote{The finite population ``design-based" perspective may be particularly attractive in settings where the experimental sample is not explicitly drawn from a larger population. In Appendix \ref{sec:fin_pop} we provide some preliminary simulation evidence that our proposed estimators may be relevant in such a setting as well, however, given the simulation evidence in \cite{chaisemartin2022at} and our currently incomplete understanding of the design-based properties of our estimators, we do not make any general claims in this paper.} Borrowing the framework from \cite{dasgupta2015causal}, \cite{rubin2016} and \cite{li2020rerandomization} propose re-randomization designs for factorial experiments which are shown to have favorable efficiency properties relative to a completely randomized design. Although we do not provide formal results comparing our fully-blocked design to these re-randomization designs, our simulation evidence suggests that, at least in the inferential framework considered in our paper, the fully-blocked design can improve efficiency relative to these re-randomization designs. Also closely related to our paper is \cite{liu2022randomization}, who extend the results in \cite{dasgupta2015causal} to general stratified randomized designs. Their results specifically exclude the setting where each treatment is assigned exactly once per block, which is the primary setting that we consider in this paper.

The rest of the paper is organized as follows. In Section \ref{sec:setup} we describe our setup and notation. Section \ref{sec:main} presents the main results. In Section \ref{sec:sims}, we examine the finite sample behavior of various experimental designs via simulation in the context of $2^K$ factorial experiments. Finally, in Section \ref{sec:application} we illustrate our proposed inference methods in an empirical application based on the experiment conducted in \cite{McKenzie2014}. We conclude with recommendations for empirical practice in Section \ref{sec:rec}.

\section{Setup and Notation} \label{sec:setup}
Let $Y_i \in \mathbf R$ denote the observed outcome of interest for the $i$th unit. Let $D_i \in \mathcal D$ denote treatment status for the $i$th unit, where $\mathcal{D}$ denotes a finite set of values of the treatment. We assume $\mathcal D = \{1, \ldots, |\mathcal D|\}$. Generally, we use $D_i = 1$ to indicate the $i$th unit is untreated, but such a restriction is not necessary for our results. Let $X_i$ denote the observed baseline covariates for the $i$th unit, and denote its dimension by $\mathrm{dim}(X_i)$. For $d \in \mathcal{D}$, let $Y_i(d)$ denote the potential outcome for the $i$th unit if its treatment status were $d$. The observed outcome and potential outcomes are related to treatment status by the expression
\begin{equation}\label{eq:PO}
Y_i = \sum_{d \in \mathcal{D}} Y_i(d)I\{D_i = d\}~.
\end{equation}
We suppose our sample consists of $J_n := (|\mathcal{D}|)n$ i.i.d.\ units. For any random variable indexed by $i$, for example $D_i$, we denote by $D^{(n)}$ the random vector $(D_1, D_2, \ldots, D_{J_n})$. Let $P_n$ denote the distribution of the observed data $Z^{(n)}$ where $Z_i = (Y_i, D_i, X_i)$, and $Q_n$ denote the distribution of $W^{(n)}$, where $W_i = (Y_i(1), Y_i(2), \ldots, Y_i(|\mathcal{D}|), X_i)$. We assume that $W^{(n)}$ consists of $J_n$ i.i.d observations, so that $Q_n = Q^{J_n}$, where $Q$ is the marginal distribution of $W_i$. Given $Q_n$, $P_n$ is then determined by (\ref{eq:PO}) and the mechanism for determining treatment assignment. We thus state our assumptions in terms of assumptions on $Q$ and the treatment assignment mechanism. 

Our object of interest will generically be defined as a vector of linear contrasts over the collection of expected potential outcomes across treatments. Formally, let 
\[\Gamma(Q) := (\Gamma_1(Q), \ldots, \Gamma_{|\mathcal{D}|}(Q))'~,\]
where $\Gamma_d(Q) := E_Q[Y_i(d)]$ for $d \in \mathcal{D}$. Let $\nu$ be a real-valued $m\times |\mathcal{D}|$ matrix. We define
\[\Delta_\nu(Q) := \nu\Gamma(Q) \in \mathbf R^m~,\]
as our generic parameter of interest. For example, in the special case where $\mathcal{D} = \{1, 2\}$ and $\nu = (-1, 1)$, $\Delta_\nu(Q) = E_Q[Y_i(2) - Y_i(1)]$ corresponds to the familiar average treatment effect for a binary treatment. Further examples of $\Delta_\nu(Q)$ are provided in Examples \ref{ex:matched-triples} and \ref{ex:2-factor} below.

We now describe our assumptions on $Q$. Our first assumption imposes restrictions on the (conditional) moments of the potential outcomes:

\begin{assumption} \label{as:Q}
The distribution $Q$ is such that
\begin{enumerate}[(a)]
\item $0 < E[\mathrm{Var}[Y_i(d) | X_i]]$ for $d \in \mathcal D$.
\item $E[Y_i^2(d)] < \infty$ for $d \in \mathcal D$.
\item $E[Y_i(d) | X_i = x]$, $E[Y_i^2(d) | X_i = x]$, and $\var[Y_i(d) | X_i]$ are Lipschitz for $d \in \mathcal D$.
\end{enumerate}
\end{assumption}
Assumption \ref{as:Q}(a) is a mild restriction imposed to rule out degenerate situations and Assumption \ref{as:Q}(b) is another mild restriction that permits the application of suitable laws of large numbers and central limit theorems. Assumption \ref{as:Q}(c), on the other hand, is a smoothness requirement that ensures that units that are ``close'' in terms of their baseline covariates are also ``close'' in terms of their potential outcomes. Assumption \ref{as:Q}(c) is a key assumption for establishing the asymptotic exactness of our proposed tests, since it allows us to argue that certain intermediate quantities in the derivations of our variance estimators vanish asymptotically (see for instance the proof of Lemma \ref{lem:rho-dd'}). Similar smoothness requirements are also imposed in \cite{bai2021inference}.

Next, we specify our assumptions on the mechanism determining treatment status. In words, we consider treatment assignments which first stratify the experimental sample into $n$ blocks of size $|\mathcal{D}|$ using the observed baseline covariates $X^{(n)}$, and then assign one unit to each treatment uniformly at random within each block. We call such a design a \emph{matched tuples} design. Formally, let
\[ \lambda_j = \lambda_j(X^{(n)}) \subseteq \{1, \ldots, J_n\},~ 1 \leq j \leq n \]
denote $n$ sets each consisting of $|\mathcal D|$ elements that form a partition of $\{1, \ldots, J_n\}$.

We assume treatment is assigned as follows:
\begin{assumption} \label{as:D}
Treatments are assigned so that $\{Y^{(n)}(d): d \in \mathcal D\} \independent D^{(n)} | X^{(n)}$ and, conditional on $X^{(n)}$,
\[ \{(D_i: i \in \lambda_j): 1 \leq j \leq n\} \]
are i.i.d.\ and each uniformly distributed over all permutations of $(1, 2, \ldots, |\mathcal{D}|)$.
\end{assumption}

We further require that the units in each block be ``close'' in terms of their baseline covariates in the following sense:
\begin{assumption} \label{as:close}
The blocks satisfy
\[ \frac{1}{n} \sum_{1 \leq j \leq n} \max_{i, k \in \lambda_j} ||X_i - X_k||^2 \stackrel{P}{\to} 0~.\]
\end{assumption}

We will also sometimes require that the distances between units in adjacent blocks be ``close" in terms of their baseline covariates:
\begin{assumption} \label{as:close-4}
The blocks satisfy
\[ \frac{1}{n} \sum_{1 \leq j \leq \lfloor n / 2 \rfloor} \max_{i \in \lambda_{2j - 1}, k \in \lambda_{2j}} ||X_i - X_k||^2 \stackrel{P}{\to} 0~. \]
\end{assumption}

We provide three examples of blocking algorithms which satisfy Assumptions \ref{as:close}--\ref{as:close-4}:
\begin{enumerate}
    \item Univariate covariate: When $\mathrm{dim}(X_i) = 1$, we can order units from smallest to largest according to $X_i$ and then block adjacent units into blocks of size $|\mathcal D|$. It then follows from Theorem 4.1 of \cite{bai2021inference} that Assumptions \ref{as:close}--\ref{as:close-4} are satisfied as long as $E[X_i^2] < \infty$.
    \item Pre-stratification: Suppose we have a covariate vector $\tilde X_i = (\tilde X_{1i}, \tilde X_{2i})$, where $\mathrm{dim}(\tilde X_{2i}) = 1$. Let $S$ be a function that maps from the support of $\tilde X_{1i}$ to a discrete set $\mathcal S = \{1, \ldots, |\mathcal S|\}$. Define $S_{1i} = S(\tilde X_{1i})$. For all units with the same value of $S_i$, order the units from smallest to largest according to $\tilde X_{2i}$ and then block adjacent units into blocks of size $|\mathcal D|$.\footnote{If the number of units in a stratum is not divisible by $|\mathcal{D}|$, we could simply assign the remaining units at random or drop them from the experiment.} It follows from Theorem 4.1 of \cite{bai2021inference} that the resulting blocks satisfy Assumptions \ref{as:close}--\ref{as:close-4} with $X_i = (S_{1i}, \tilde X_{2i})$ as long as $E[\tilde X_{2i}^2] < \infty$. As an example, suppose $\tilde X_1 = (\text{gender, education level})$ and $\tilde X_2 = \text{income}$. In this case, the blocks could be formed by first stratifying according to gender and education level and then blocking on income. A similar blocking procedure is used in the experiment conducted by \cite{McKenzie2014} which we revisit in our empirical application in Section \ref{sec:application}.
    \item Recursive pairing: When $\mathrm{dim}(X_i) > 1$ and $|\mathcal D| = 2^K$ for some $K$, we could form blocks by repeatedly implementing the ``pairs of pairs'' algorithm in Section 4 of \cite{bai2021inference} to successively larger groups of size $2^k$ for $k = 0, 1, \ldots, K$. To do this, units would first be matched into pairs (using for instance the non-bipartite matching algorithm from the \texttt{R} package \texttt{nbpMatching}). Next, these matched pairs would themselves be matched into ``pairs of pairs" using the average value of the covariates in each pair, in order to generate groups of size four. Continuing in this fashion, we would match pairs of groups until obtaining groups of size $2^K$. This is the algorithm we employ in our simulation designs. Such an algorithm could again be shown to satisfy Assumptions \ref{as:close}--\ref{as:close-4}. 
\end{enumerate}


\section{Main Results} \label{sec:main}
\subsection{Inference for Matched Tuples Designs}\label{sec:main_tuple}
In this section, we study estimation and inference for a general $m$-dimensional parameter $\Delta_\nu(Q)$ under a matched tuples design. For a pre-specified $\ell \times 1$ column vector $\Delta_0$ and $\ell \times m$ matrix $\Psi$ of rank $\ell$, the testing problem of interest is 
\begin{equation}\label{eq:nu_test}
H_0: \Psi\Delta_{\nu}(Q) = \Delta_0 \text{ versus } H_1:\Psi\Delta_\nu(Q) \ne \Delta_0
\end{equation}
at level $\alpha \in (0, 1)$.
First we describe our estimator of $\Delta_\nu(Q)$. For $d \in \mathcal D$, define
\[ \hat \Gamma_n(d) := \frac{1}{n} \sum_{1 \leq i \leq J_n} I \{D_i = d\} Y_i~, \]
and let $\hat \Gamma_n = (\hat \Gamma_{n}(1), \ldots, \hat \Gamma_{n}(|\mathcal{D}|))'$. In words, $\hat\Gamma_n(d)$ is simply the sample mean of the observations with treatment status $D_i = d$, and $\hat \Gamma_n$ is the vector of sample means across all treatments $d \in \mathcal{D}$. With $\hat\Gamma_n$ in hand, our estimator of $\Delta_\nu(Q)$ is then given by
\[\hat{\Delta}_{\nu,n} := \nu\hat\Gamma_n~.\]
 In what follows, it will be useful to define $\Gamma_d(X_i) := E[Y_i(d)|X_i]$. Our first result derives the limiting distribution of $\hat{\Delta}_{\nu, n}$ under our maintained assumptions.
\begin{theorem}\label{thm:main_delta}
Suppose $Q$ satisfies Assumption \ref{as:Q} and the treatment assignment mechanism satisfies Assumptions \ref{as:D}--\ref{as:close}. Then,
\[ \sqrt{n}(\hat\Delta_{\nu,n}- \Delta_\nu(Q)) \stackrel{d}{\to} N(0, \mathbb V_\nu)~, \]
where $\mathbb V_\nu := \nu\mathbb V \nu'$, with 
\begin{equation} \label{eq:V}
\mathbb V := \mathbb V_1 + \mathbb V_2~,
\end{equation}
\[\mathbb V_1 := \mathrm{diag}(E[\mathrm{Var}[Y_i(d) | X_i]]: d \in \mathcal D)~,\]
\[\mathbb V_2 := \left[\frac{1}{|\mathcal D|} \mathrm{Cov}[\Gamma_d(X_i), \Gamma_{d'}(X_i)]\right]_{d,d'\in \mathcal{D}}~.\]
\end{theorem}


To construct our test, we next define a consistent estimator for the asymptotic variance matrix $\mathbb V_\nu$. To begin, note by the law of total variance that
\[ E[\var[Y_i(d) | X_i]] = \var[Y_i(d)] - E[E[Y_i(d) | X_i]^2] + E[Y_i(d)]^2~. \]
Therefore, in order to estimate $\mathbb V_1$ consistently, it suffices to provide consistent estimators for $E[E[Y_i(d) | X_i]^2]$, $E[Y_i(d)]$, and $\var[Y_i(d)]$. A similar remark applies to $\mathbb V_2$. In light of this, define
\begin{align*}\label{eq:rho}
\hat \rho_n(d, d) &:= \frac{2}{n} \sum_{1 \leq j \leq \lfloor n / 2 \rfloor} \Big ( \sum_{i \in \lambda_{2j - 1}} Y_i I \{D_i = d\} \Big ) \Big ( \sum_{i \in \lambda_{2j}} Y_i I \{D_i = d\} \Big ) \\
\hat \rho_n(d, d') &:= \frac{1}{n} \sum_{1 \leq j \leq n} \Big ( \sum_{i \in \lambda_j} Y_i I \{D_i = d\} \Big ) \Big ( \sum_{i \in \lambda_j} Y_i I \{D_i = d'\} \Big ) \text{ if } d \neq d' \\
\hat \sigma_n^2(d) &:= \frac{1}{n} \sum_{1 \leq i \leq J_n} (Y_i - \hat \Gamma_n(d))^2 I \{D_i = d\}~.
\end{align*}
To understand the construction, note that in order to estimate $E[E[Y_i(d) | X_i]^2]$ consistently, we would ideally average over the products of the outcomes of two units with similar values of $X_i$ and both with treatment status $d$. By construction, however, only one unit in each block has treatment status $d$. To overcome this problem, note that Assumption \ref{as:close-4} ensures that in the limit units in adjacent blocks also have similar values of $X_i$. Therefore, to construct our estimator of $E[E[Y_i(d)|X_i]^2]$, denoted by $\hat{\rho}_n(d,d)$, we average over the product of the outcomes of the units with treatment status $d$ in two adjacent blocks. $\hat \rho_n(d, d)$ is analogous to the ``pairs of pairs" variance estimator in \cite{bai2021inference}. A similar construction has also been used in \cite{abadie2008estimation} in a related setting. On the other hand, for $d \neq d'$, we have distinct units with treatment status $d$ and $d'$ within each block, and therefore our estimator of $E[E[Y_i(d) | X_i] E[Y_i(d') | X_i]]$, denoted $\hat{\rho}_n(d,d')$, can be estimated using units within the same block.

Our estimator for $\mathbb V_\nu$ is then given by $\hat{\mathbb V}_{\nu, n} := \nu\hat{\mathbb V}_n\nu'$, where 
\begin{align*}
\hat{\mathbb{V}}_n &:= \hat{\mathbb{V}}_{1,n} + \hat{\mathbb{V}}_{2,n}\\
\hat{\mathbb{V}}_{1,n} &:= \mathrm{diag}\left(\hat{\mathbb{V}}_{1,n}(d):d \in \mathcal{D}\right)\\
\hat{\mathbb{V}}_{2,n} &:= \left[\hat{\mathbb{V}}_{2,n}(d,d')\right]_{d, d' \in \mathcal{D}}~,
\end{align*}
with
\begin{align*}
\hat{\mathbb{V}}_{1,n}(d) &:= \hat\sigma^2_n(d) - (\hat{\rho}_n(d,d) - \hat{\Gamma}_n^2(d))\\
\hat{\mathbb{V}}_{2,n}(d, d') &:= \frac{1}{|\mathcal{D}|}(\hat{\rho}_n(d, d') - \hat{\Gamma}_n(d)\hat{\Gamma}_n(d'))~.
\end{align*}

Given this estimator, our test is given by
\[\phi_n^{\nu}(Z^{(n)}) = I\{T_n^{\nu}(Z^{(n)}) > c_{1 - \alpha}\}~,\]
where 
\[T_n^{\nu}(Z^{(n)}) = n(\Psi\hat\Delta_{\nu,n}- \Psi \Delta_0)'(\Psi\hat{\mathbb V}_{\nu, n}\Psi')^{-1}(\Psi\hat\Delta_{\nu,n} - \Psi \Delta_0)~,\]
and $c_{1 - \alpha}$ is the $1 - \alpha$ quantile of the $\chi^2_\ell$ distribution. Our next result establishes the consistency of $\hat{\mathbb{V}}_n$ for $\mathbb{V}$ and the asymptotic validity of the above test.
\begin{theorem}\label{thm:V_const}
Suppose $Q$ satisfies Assumption \ref{as:Q} and the treatment assignment mechanism satisfies Assumptions \ref{as:D}--\ref{as:close-4}. Then,
\[\hat{\mathbb{V}}_n \stackrel{P}{\to} \mathbb{V}~.\]
Therefore, for the problem of testing (\ref{eq:nu_test}) at level $\alpha \in (0,1)$, $\phi_n^{\nu}(Z^{(n)})$ satisfies
\[\lim_{n \rightarrow \infty}E[\phi_n^{\nu}(Z^{(n)})] = \alpha~,\]
under the null hypothesis.
\end{theorem}



\begin{example}{(Inference for Matched Triples)}\label{ex:matched-triples}
Consider the setting where $\mathcal{D} = \{1, 2, 3\}$, where we consider $d = 1$ as a control arm and $d = 2, 3$ as treatment sub-arms. See, for example, \cite{bold2018experimental} and \cite{brown2020inducing}. Suppose our parameter of interest is the vector of average treatment effects for the treatments $d = 2, 3$ versus control $d = 1$. In this case, the parameter of interest is given by $\Delta_\nu(Q)$, where
\[
\nu = 
\begin{pmatrix}
-1 & 1 & 0 \\
-1 & 0 & 1
\end{pmatrix}
~.\]
It follows from Theorem \ref{thm:main_delta} that
\[\sqrt n(\hat\Delta_{\nu,n}- \Delta_\nu(Q)) \stackrel{d}{\to} N(0, \mathbb{V}_\nu)~,\]
where
\[ \mathbb V_\nu = \begin{pmatrix}
\sigma_{\nu,1,1}^2 & \sigma_{\nu,1,2}^2  \\
\sigma_{\nu,1,2}^2 & \sigma_{\nu,2,2}^2
\end{pmatrix}~,\]
and
\begin{align*}
\sigma_{\nu,1,1}^2 &= E[\var[Y_i(1) | X_i]] + E[\var[Y_i(2) | X_i]] + \frac{1}{3} E\left[ \left( (\Gamma_1(X_i) - \Gamma_1) - (\Gamma_2(X_i) - \Gamma_2) \right)^2 \right] \\
\sigma_{\nu,2,2}^2 &= E[\var[Y_i(1) | X_i]] + E[\var[Y_i(3) | X_i]] + \frac{1}{3} E\left[ \left( (\Gamma_1(X_i) - \Gamma_1) - (\Gamma_3(X_i) - \Gamma_3) \right)^2 \right] \\
\sigma_{\nu,1,2}^2 &= E[\var[Y_i(1) | X_i]] + \frac{1}{3} E\left[\left( (\Gamma_1(X_i) - \Gamma_1) - (\Gamma_2(X_i) - \Gamma_2) \right)\left( (\Gamma_1(X_i) - \Gamma_1) - (\Gamma_3(X_i) - \Gamma_3) \right) \right] ~,
\end{align*}
where we recall $\Gamma_d(X_i) = E[Y_i(d)|X_i]$. These variance formulas imply the following two observations: first, by decomposing $\sigma^2_{\nu,1,1}$ using the law of total variance, we can show that the commonly-used two-sample $t$-test is conservative when testing the null hypothesis on the contrast of any two treatment levels in a matched tuples design. A similar observation was made in the special case of a matched-pair design in \cite{bai2021inference}. Second, the adjusted $t$-test developed in \cite{bai2021inference} is also conservative for testing such hypotheses. Specifically, \cite{bai2021inference} study inference for $E[Y(2) - Y(1)]$ in a matched-pair design when $|\mathcal D| = 2$ and the sample size is $2n$. In a matched triples experiment with $|\mathcal D| = 3$ and sample size $3n$, researchers may be tempted to apply the variance estimator from Theorem 3.3 in \cite{bai2021inference} to the subsample with $D_i \in \{1, 2\}$. However, it can be shown in our framework that the limit of the variance estimator from \cite{bai2021inference} is given by replacing $\frac{1}{3}$ in the last term of $\sigma_{\nu, 1, 1}^2$ with $\frac{1}{2}$. Therefore, the test which studentizes using the variance estimator from \cite{bai2021inference} would be asymptotically conservative in our setting. 
\end{example}

Next, we study the properties of two commonly recommended inference procedures in the analysis of matched tuple designs. The first procedure is a $t$-test obtained from a linear regression of outcomes on treatment indicators and block fixed effects. Specifically, we consider a $t$-test obtained from the following regression:
\begin{equation} \label{eq:sfe}
Y_i = \sum_{d \in \mathcal D \backslash \{1\}} \beta(d) I \{D_i = d\} + \sum_{1 \leq j \leq n} \delta_j I \{i \in \lambda_j\} + \epsilon_i~,
\end{equation}
which we interpret as the projection of $Y$ on the indicators for treatment status and block fixed effects. Let $\hat \beta_n(d)$, $d \in \mathcal D \backslash \{1\}$ and $\hat \delta_{j, n}$, $1 \leq j \leq n$ denote the OLS estimators of $\beta(d)$, $d \in \mathcal D \backslash \{1\}$ and $\delta_j$, $1 \leq j \leq n$. It is common in practice to use $\hat \beta_n(d)$ as an estimator for the pairwise average treatment effect between treatment $d$ and treatment $1$. See, for instance, \cite{McKenzie2013} and \cite{McKenzie2014}. Furthermore, researchers often conduct inference on the pairwise ATEs using the heteroskedasticity-robust variance estimator obtained from \eqref{eq:sfe}. Formally, for $d \in \mathcal D \backslash \{1\}$ and $\Delta_0 \in \mathbf R$, consider the problem of testing
\begin{equation} \label{eq:H0-sfe}
E_Q[Y_i(d)] - E_Q[Y_i(1)] = \Delta_0 \text{ versus } H_1: E_Q[Y_i(d)] - E_Q[Y_i(1)] \neq \Delta_0
\end{equation}
at level $\alpha \in (0, 1)$. Let $\kappa_j\cdot \hat {\mathbb V}_n^{\rm sfe}(d, 1)$ denote the ``HC$j$" heteroskedasticity-robust variance estimator of $\hat \beta_n(d)$ from the linear regression in \eqref{eq:sfe}, where $\kappa_j$ for $j \in \{0, 1\}$ corresponds to one of two common degrees of freedom corrections \citep[see][]{mackinnon1985some}:
\[ \kappa_j  = \begin{cases} 
      1 & \text{if $j = 0$} \\
      \frac{\mathcal{|D|}n}{|\mathcal{D}|n - (|\mathcal{D}| - 1 + n)} & \text{if $j = 1$}~. 
   \end{cases}
\] The test is then defined as
\begin{equation} \label{eq:test-sfe}
\phi_n^{\rm sfe}(Z^{(n)}) = I \{|T_n^{\rm sfe}(Z^{(n)})| > z_{1 - \frac{\alpha}{2}}\}~,
\end{equation}
where $z_{1 - \frac{\alpha}{2}}$ is the $(1 - \frac{\alpha}{2})$-th quantile of the standard normal distribution and
\begin{equation} \label{eq:stat-sfe}
T_n^{\rm sfe}(Z^{(n)}) = \frac{\hat \beta_n(d) - \Delta_0}{\sqrt{\kappa_j \cdot \hat {\mathbb V}_n^{\rm sfe}(d, 1)}}~.
\end{equation}
The following theorem shows that the OLS estimator $\hat \beta_n(d)$ is numerically equivalent to the standard difference-in-means estimator. However, it shows that the $t$-test defined in (\ref{eq:test-sfe}) is not generally valid for testing the null hypothesis defined in (\ref{eq:H0-sfe}). 
\begin{theorem} \label{thm:sfe}
Suppose $Q$ satisfies Assumption \ref{as:Q} and the treatment assignment mechanism satisfies Assumptions \ref{as:D}--\ref{as:close-4}. Then,
\[ \hat \beta_n(d) = \hat \Gamma_n(d) - \hat \Gamma_n(1) \text{ for } d \in \mathcal D \backslash \{1\}~. \]
Moreover, 
\begin{itemize}
    \item  Using estimator $\mathrm{HC}0$, the limiting rejection probability of the test defined in \eqref{eq:test-sfe} could be strictly larger than $\alpha$. 
    \item Using estimator $\mathrm{HC}1$, the limiting rejection probability of the test defined in \eqref{eq:test-sfe} could be strictly larger than $\alpha$ for $|\mathcal{D}| > 2$.
\end{itemize}
\end{theorem}

\cite{bai2021inference} remark that the test defined in \eqref{eq:test-sfe} is conservative in the context of a matched-pair design when using $\mathrm{HC}1$. Theorem \ref{thm:sfe} shows that, when considering a matched tuples design with more than two treatments, this is no longer necessarily the case. 

\begin{remark}\label{rem:HC_conservative}
An inspection of the proof of Theorem \ref{thm:sfe} reveals that the probability limit of $n\cdot\kappa_1\hat{\mathbb{V}}^{\rm sfe}_n(d,1)$ is given by
\begin{align*}
    &\frac{|\mathcal{D}|}{|\mathcal{D}| - 1}\Big(\var \left [ \Gamma_1(X_i) - \frac{1}{|\mathcal D|} \sum_{d' \in \mathcal D} \Gamma_{d'}(X_i) \right ] + \left ( 1 - \frac{1}{|\mathcal D|} \right )^2 E[\var[Y_i(1) | X_i]] + \frac{1}{|\mathcal D|^2} \sum_{d' \in \mathcal D \backslash \{1\}} E[\var[Y_i(d') | X_i]] \\
    & \hspace{3em} + \var \left [ \Gamma_d(X_i) - \frac{1}{|\mathcal D|} \sum_{d' \in \mathcal D} \Gamma_{d'}(X_i) \right ] + \left ( 1 - \frac{1}{|\mathcal D|} \right )^2 E[\var[Y_i(d) | X_i]] + \frac{1}{|\mathcal D|^2} \sum_{d' \in \mathcal D \backslash \{d\}} E[\var[Y_i(d') | X_i]]\Big)~,
\end{align*}
whereas the true asymptotic variance of $\hat{\beta}_n(d)$ is given by
\[ E\left[\var[Y_i(d) | X_i]] + E[\var[Y_i(1) | X_i] \right] + \frac{1}{|\mathcal D|} E\left[ \left( (\Gamma_d(X_i) - \Gamma_d) - (\Gamma_1(X_i) - \Gamma_1) \right)^2 \right]~. \]
From these expressions, we can conclude that when $|\mathcal{D}|$ is large it is likely that $\kappa_1\hat{\mathbb{V}}^{\rm sfe}_n(d,1)$ is conservative. However, as shown in the proof of Theorem \ref{thm:sfe}, this cannot be guaranteed for finite $|\mathcal{D}| > 2$ in general.
\end{remark}

The second procedure is a block-cluster robust $t$-test which modifies a recent proposal in \cite{chaisemartin2022at} to the setting with multiple treatments. Specifically, we consider a cluster-robust $t$-test constructed from a regression of outcomes on a constant and treatment indicators:
\[ Y_i = \gamma(1) + \sum_{d \in \mathcal D \backslash \{1\}} \gamma(d) I \{D_i = d\} + \epsilon_i~, \]
where clusters are defined at the level of \emph{blocks} of units $\{\lambda_j\}_{1 \le j \le \mathcal{D}}$.
Let $\hat{\gamma}_n(d)$, $d \in \mathcal{D} \backslash \{1\}$ denote the OLS estimator of $\gamma(d)$, it then follows immediately that $\hat{\gamma}_n(d) = \hat{\Gamma}_n(d) - \hat{\Gamma}_n(1)$. We then consider the problem of testing \eqref{eq:H0-sfe} at level $\alpha \in (0, 1)$ using a test defined by
\[\phi^{\rm bcve}_n(Z^{(n)}) = I\{|T^{\rm bcve}_n(Z^{(n)})| > z_{1 - \frac{\alpha}{2}}\}~,\]
where $z_{1 - \frac{\alpha}{2}}$ is the $(1 - \frac{\alpha}{2})$-th quantile of the standard normal distribution and 
\begin{equation} \label{eq:stat-bcve}
T_n^{\rm bcve}(Z^{(n)}) = \frac{\hat \gamma_n(d) - \Delta_0}{\sqrt{\hat {\mathbb V}_n^{\rm bcve}(d)}}~,
\end{equation}
with $\hat {\mathbb V}_n^{\rm bcve}(d)$ denoting the $d$-th diagonal element of the block-cluster variance estimator defined as: 
\begin{equation}\label{eq:BCVE}
    \hat{\mathbb{V}}^{\rm bcve}_n =  \left(\sum_{1\leq j \leq n} \sum_{i \in \lambda_j} C_i C_i'\right)^{-1} \left(\sum_{1\leq j \leq n} \left(\sum_{i\in\lambda_j} \hat \epsilon_i C_i \right) \left(\sum_{i\in\lambda_j} \hat \epsilon_i C_i \right)^{\prime} \right)\left(\sum_{1\leq j \leq n} \sum_{i \in \lambda_j} C_i C_i'\right)^{-1}~,
\end{equation}
where $C_i = (1, I \{D_i = 2\},\dots, I \{D_i = |\mathcal{D}|\})'$ and $\hat \epsilon_i = \sum_{d \in \mathcal{D}\backslash \{1\}} (Y_i - \hat \gamma_n(d)) I\{D_i = d\} + Y_i I\{D_i=1\}- \hat\gamma_n(1)$.

The following theorem shows that the $t$-test defined in \eqref{eq:stat-bcve} is generally conservative for testing the null hypothesis defined in \eqref{eq:H0-sfe}. 

\begin{theorem}\label{thm:bcve}
Consider the block-cluster variance estimator $\hat{\mathbb V}_n^{\rm bcve}$ as defined in \eqref{eq:BCVE} in the Appendix. Then the $d$-th diagonal element of this estimator is equal to
\[n\cdot\hat{\mathbb {V}}_n^{\rm bcve}(d) = \frac{1}{n} \sum_{1 \leq j \leq n} \left ( \sum_{i \in \lambda_j} Y_i I \{D_i = d\} - \sum_{i \in \lambda_j} Y_i I \{D_i = 1\} \right )^2 - (\hat \Gamma_n(d) - \hat \Gamma_n(1))^2 ~. \]
Moreover, under Assumptions \ref{as:Q}--\ref{as:close}, 
\[n\cdot\hat{\mathbb {V}}_n^{\rm bcve}(d) \xrightarrow{p} E[\var[Y_i(d) | X_i]] + E[\var[Y_i(1) | X_i]] + E\left[ \left( (\Gamma_d(X_i) - \Gamma_d) - (\Gamma_1(X_i) - \Gamma_1) \right)^2 \right]~.\]
It thus follows that the test defined in \eqref{eq:stat-bcve} is conservative for testing the null hypothesis defined in \eqref{eq:H0-sfe} unless
\begin{equation}\label{eq:bcve-exact}
E\left[ \left( (\Gamma_d(X_i) - \Gamma_d) - (\Gamma_1(X_i) - \Gamma_1) \right)^2 \right] = 0~.
\end{equation}
\end{theorem}

\begin{remark}\label{rem:bcve}
An inspection of the proof of Theorem \ref{thm:bcve} reveals that, unless \eqref{eq:bcve-exact} holds, the difference between the probability limit of $n\cdot\hat{\mathbb V}^{\rm bcve}_n(d)$ and the asymptotic variance of $\hat{\Gamma}_n(d) - \hat{\Gamma}_n(1)$ is equal to
\[ \left ( 1 - \frac{1}{|\mathcal{D}|} \right ) E\left[ \left( (\Gamma_d(X_i) - \Gamma_d) - (\Gamma_1(X_i) - \Gamma_1) \right)^2 \right]~. \]
It thus follows that the test defined in \eqref{eq:stat-bcve} in fact becomes more conservative for testing $\eqref{eq:H0-sfe}$ as the number of treatments $|\mathcal{D}|$ increases.
\end{remark}

\subsection{Inference for ``Replicate'' Designs} \label{sec:replicate}
Our analysis so far has focused on the setting where $J_n = |\mathcal{D}| n$ units are blocked into $n$ blocks of size $|\mathcal{D}|$, and each treatment $d \in \mathcal{D}$ is assigned exactly once in each block. In this section, we consider a modification of this design where units are grouped into blocks of size $2|\mathcal{D}|$ and each treatment status $d \in \mathcal{D}$ is assigned exactly \emph{twice} in each block. Formally, for the remainder of this section suppose $n$ is even, and let
\[ \tilde \lambda_j = \tilde \lambda_j(X^{(n)}) \subseteq \{1, \ldots, J_n\},~ 1 \leq j \leq n / 2\]
denote $n/2$ sets each consisting of $2 |\mathcal D|$ elements that form a partition of $\{1, \ldots, J_n\}$.

We assume treatment is assigned as follows:
\begin{assumption} \label{as:D-replicate}
Treatments are assigned so that $\{Y^{(n)}(d): d \in \mathcal D\} \independent D^{(n)} | X^{(n)}$ and, conditional on $X^{(n)}$,
\[ \{(D_i: i \in \tilde \lambda_j): 1 \leq j \leq n / 2\} \]
are i.i.d.\ and each uniformly distributed over all permutations of $(1, 1, 2, 2, \ldots, |\mathcal{D}|, |\mathcal{D}|)$.
\end{assumption}

We further require that the units in each block be ``close'' in terms of their baseline covariates in the following sense:
\begin{assumption} \label{as:close-replicate}
The blocks satisfy
\[ \frac{1}{n} \sum_{1 \leq j \leq n / 2} \max_{i, k \in \tilde \lambda_j} ||X_i - X_k||^2 \stackrel{P}{\to} 0~.\]
\end{assumption}

We first establish that the limiting distribution of $\hat{\Delta}_{\nu,n}$ for such a ``replicate'' design is the same as that for the matched tuples design considered in Theorem \ref{thm:main_delta}.

\begin{theorem} \label{thm:replicate-delta}
Suppose $Q$ satisfies Assumption \ref{as:Q} and the treatment assignment mechanism satisfies Assumptions \ref{as:D-replicate}--\ref{as:close-replicate}. Then,
\[ \sqrt n(\hat \Delta_{\nu, n} - \Delta_\nu(Q)) \stackrel{d}{\to} N(0, \mathbb V_\nu)~, \]
with $\mathbb{V}_{\nu}$ as defined in Theorem \ref{thm:main_delta}.
\end{theorem}

Although the limiting distribution of $\hat{\Delta}_{\nu,n}$ for the standard matched tuples and replicate designs are identical, variance estimation in the replicate design is often understood to be conceptually simpler, because each treatment status is assigned \emph{twice} in each block \citep[see for instance the discussion of variance estimation in][in the context of matched pair designs]{athey2017econometrics}. Indeed, in this case an alternative variance estimator can be constructed which is identical to the estimator proposed in Section \ref{sec:main_tuple} except that we replace $\hat \rho_n(d,d)$ by
\begin{equation*}
    \tilde \rho_n(d, d) = \frac{2}{n} \sum_{1 \leq j \leq \lfloor n/2 \rfloor} \Big ( \prod_{i \in \lambda_{j}} Y_i I \{D_i = d\} \Big ) ~,
\end{equation*}
which no longer requires averaging over the product of outcomes of units in adjacent blocks. The following theorem establishes the consistency of $\tilde \rho_n(d, d)$, where importantly we note that Assumption \ref{as:close-4}, which maintains that adjacent blocks be ``close", is no longer required. It is then straightforward to show the consistency of the corresponding variance estimator for $\hat \Delta_{\nu, n}$ constructed by replacing $\hat{\rho}_n(d,d)$ in $\hat{\mathbb V}_n$ with $\tilde \rho_n(d, d)$.

\begin{theorem} \label{thm:replicate-rho}
Suppose $Q$ satisfies Assumption \ref{as:Q} and the treatment assignment mechanism satisfies Assumptions \ref{as:D-replicate}--\ref{as:close-replicate}. Then,
\begin{equation} \label{eq:replicate-consistent}
    \tilde \rho_n(d, d) \stackrel{P}{\to} E[E[Y_i(d) | X_i]^2]~.
\end{equation}
\end{theorem}

We remark that Theorems \ref{thm:main_delta}--\ref{thm:V_const} and Theorems \ref{thm:replicate-delta}--\ref{thm:replicate-rho}, yielding identical conclusions, do not allow us to effectively compare the properties of the standard matched tuples design and matched tuples with replicates. In order to compare these designs, we evaluate their finite sample properties via simulation in Section \ref{sec:sims}. There, we find that the mean squared error of $\hat{\Delta}_{\nu,n}$ under the replicate design is typically larger than under the standard non-replicate design. However, we also find that the rejection probabilities of our proposed tests under the replicate design are much closer to the nominal level relative to the non-replicate design, which can sometimes exhibit rejection probabilities strictly smaller than the nominal level when matching on multiple covariates. As a result, the replicate design is sometimes able to achieve better power relative to the non-replicate design. We emphasize, however, that our current asymptotic framework is not precise enough to capture these differences. One possible conjecture is that since replicate designs could be thought of as convex combinations of matched tuples designs \citep[see Lemma 2 in][] {bai2022optimality}, it is as if we are averaging over multiple matched tuples designs when we estimate the limiting variance. However, we leave a detailed theoretical comparison of these two designs to future work.

\subsection{Asymptotic Properties of Fully-Blocked $2^K$ Factorial Designs}\label{sec:factorial}
In this section we apply the results derived in Sections \ref{sec:main_tuple}--\ref{sec:replicate} to study the asymptotic properties of what we call ``fully-blocked" $2^K$ factorial designs. Section \ref{sec:factorial_setup} introduces $2^K$ factorial experiments. Section \ref{sec:fact_properties} introduces the fully-blocked factorial design and compares the efficiency properties of fully-blocked factorial designs to some alternative designs.

\subsubsection{Setup and Notation for $2^K$ factorial designs}\label{sec:factorial_setup}
In this section we describe the setup of a $2^K$ factorial experiment, the resulting parameters of interest, and their corresponding estimators \citep[see][for a textbook treatment] {wu2011experiments}. A $2^K$ factorial design assigns treatments which are combinations of multiple ``factors," where each factor can take two distinct values, or ``levels.'' For instance, \cite{karlan2014agricultural} study the effect of capital constraints and uninsured risk on the investment decisions of farmers in Ghana. In their setting, each treatment consists of two factors: whether or not a household receives a cash grant, and whether or not a household receives an insurance grant. Our setup and notation mirror the framework introduced in \cite{dasgupta2015causal} and \cite{li2020rerandomization}.  Given $K$ factors each with two treatment levels $\{-1, +1\}$, our set of treatments $\mathcal{D}$ now consists of all possible $2^K$ factor combinations. For a factor combination $d \in \mathcal{D}$, define $\iota_k(d) \in \{-1, +1\}$ to be the level of factor $k$ under treatment $d$. The vector $\iota(d) := (\iota_1(d), \iota_2(d), \ldots, \iota_K(d))$ then describes the levels of all $K$ factors associated with factor combination $d$. This notation allows us to define \emph{factorial effects} as parameters of the form $\Delta_{\nu}(Q)$ for appropriately constructed contrast vectors $\nu$. For instance, consider the contrast vector defined as
\[\nu_k := \left(\iota_k(1), \iota_k(2), \ldots, \iota_k(|\mathcal{D}|)\right)~.\]
Then, the parameter $\Delta_{\nu_k}(Q)$ obtained from this contrast can be written as
\begin{align*}
\Delta_{\nu_k}(Q) = \sum_{d \in \mathcal{D}}I\{\iota_k(d) = +1\}\Gamma_d(Q) - \sum_{d \in \mathcal{D}}I\{\iota_k(d) = -1\}\Gamma_d(Q)~.
\end{align*}
We define the \emph{main effect} of factor $k$ as $2^{-(K-1)}\Delta_{\nu_k}(Q)$. In words, the main effect of factor $k$ measures the average difference between the outcomes of factor combinations under which the $k$th factorial effect is $1$ versus the outcomes of factor combinations under which the $k$th factorial effect is $-1$. The re-scaling $2^{-(K-1)}$ is introduced because there are $2^{K-1}$ possible values for all the factor combinations when fixing the $k$th factor. We call $\nu_k$ the \emph{generating vector} for the main effect of factor $k$.

We can subsequently build on the generating vectors of the main effects in order to define the \emph{interaction effects} between various factors. The interaction effect between a given set of factors is defined using the contrast obtained from taking the element-wise product of the generating vectors for the relevant factors. For instance, the two-factor interaction between factors $k$ and $k'$ is defined as $2^{-(K-1)}\Delta_{\nu_{k,k'}}(Q)$, where $\nu_{k,k'} := \nu_k \odot \nu_{k'}$ and $\odot$ denotes element-wise multiplication. Similarly, the three-factor interaction $2^{-(K-1)}\Delta_{\nu_{k,k',k''}}(Q)$ is defined using the contrast vector $\nu_{k,k',k''} := \nu_k \odot \nu_{k'} \odot \nu_{k''}$. We illustrate these definitions in the special case of a $2^2$ factorial design in Example \ref{ex:2-factor} below. For simplicity, in what follows, we omit the re-scaling by $2^{-(K-1)}$ in our discussions and results.

\begin{example}\label{ex:2-factor}
Here we illustrate the concept of main and interaction effects in the case of a $2^2$ factorial design. Table \ref{table:2-factor} depicts the 4 factor combinations and their corresponding factor levels. 
\begin{table}[htbp]\label{table:2-factor}
\centering
\begin{tabular}{cccc}
\toprule
Factor Combination & Factor 1 & Factor 2 & Factor 1/2 Interaction \\
\midrule
1     & -1    & -1    & +1 \\
2     & -1    & +1     & -1 \\
3     & +1     & -1    & -1 \\
4     & +1     & +1     & +1 \\
\bottomrule
    \end{tabular}%
  \label{tab:addlabel}%
  \caption{Example of a $2^2$ factorial design}
\end{table}%

From the column labeled Factor 1 we observe that the generating vector for the main effect of factor one, $\nu_1$, is given by
\[\nu_1 = \left(-1, -1, +1, +1\right)~,\]
so that the main effect of factor one is given by (up to re-scaling)
\[\Delta_{\nu_1}(Q) = E_Q[Y_i(+1, +1) + Y_i(+1, -1)] - E_Q[Y_i(-1, +1) + Y_i(-1, -1)]~,\]
where here we have indexed potential outcomes explicitly by their factor levels. Similarly, the column labeled Factor 2 corresponds to the generating vector for the main effect of factor two, $\nu_2$. To define the interaction effect between factors one and two, we construct the relevant contrast by taking the element-wise product of $\nu_1$ and $\nu_2$:
\[\nu_{1,2} = \nu_1 \odot \nu_2 = \left(+1, -1, -1, +1\right)~,\]
this produces the column labeled Factor 1/2 Interaction. Accordingly, the interaction effect between factors one and two is given by (up to re-scaling)
\[\Delta_{\nu_{1,2}}(Q) = E_Q[Y_i(+1,+1) - Y_i(-1,+1)] - E_Q[Y_i(+1, -1) - Y_i(-1, -1)]~.\]
In words, $\Delta_{\nu_{1,2}}(Q)$ measures the difference in the the average difference in potential outcomes over factor one when factor two is set to $1$ versus the average difference in potential outcomes over factor one when factor two is set to $-1$.
\end{example}

Given the above setup, we estimate the factorial effect given by $\Delta_{\nu}(Q)$ using the estimator $\hat{\Delta}_{\nu,n}$ defined in Section \ref{sec:main_tuple}. \cite{wu2011experiments} and \cite{dasgupta2015causal} explain that $\hat{\Delta}_{\nu,n}$ is a standard estimator in this context. For instance, the estimator of the main effect of factor $k$, $2^{-(K-1)}\hat \Delta_{\nu_k,n}$, is in fact the difference-in-means estimator over the $k$-th factor:
\begin{align*}
2^{-(K-1)}\hat\Delta_{\nu_k,n} &= \frac{1}{2^{K-1}}\sum_{d \in \mathcal{D}}I\{\iota_k(d) = +1\}\hat{\Gamma}_n(d) - \frac{1}{2^{K-1}}\sum_{d \in \mathcal{D}}I\{\iota_k(d) = -1\}\hat{\Gamma}_n(d) \\
&= \frac{1}{n2^{K-1}}\sum_{1 \leq i \leq J_n}\sum_{d \in \mathcal{D}}I\{\iota_k(d) = +1\}  I \{D_i = d\} Y_i - \frac{1}{n2^{K-1}}\sum_{1 \leq i \leq J_n}\sum_{d \in \mathcal{D}}I\{\iota_k(d) = -1\}  I \{D_i = d\} Y_i\\
&= \frac{1}{n2^{K-1}}\sum_{1 \leq i \leq J_n}I\{\iota_k(D_i) = +1\}  Y_i - \frac{1}{n2^{K-1}}\sum_{1 \leq i \leq J_n}I\{\iota_k(D_i) = -1\}  Y_i ~.
\end{align*}

\subsubsection{Efficiency Properties of Fully-Blocked Factorial Designs}\label{sec:fact_properties}
In this section, we compare the asymptotic variance of the estimator $\hat\Delta_{\nu,n}$ under what we call a ``fully-blocked" factorial design relative to some alternative designs. A fully-blocked factorial design first blocks the experimental sample into $n$ blocks of size $2^K$ based on the observable characteristics $X^{(n)}$, and then assigns each of the $2^K$ factor combinations exactly once in each block. Formally, a fully-blocked factorial design is simply a matched tuples design as defined in Section \ref{sec:setup}, where $\mathcal{D}$ consists of the set of all possible factor combinations. 

Our first result compares the fully-blocked factorial design to completely randomized and stratified factorial designs. Given a $2^K$ factorial experiment and a sample of size $J_n = n2^K$, a completely randomized factorial design simply assigns $n$ individuals to each of the $2^K$ factor combinations at random. A stratified factorial design first partitions the covariate space into a finite number of groups, or ``strata", and then performs a completely randomized factorial design within each stratum. Formally, let $h: \mathrm{supp}(X) \to \{1, \ldots, S\}$ be a function which maps covariate values into a set of discrete strata labels. Then, a stratified factorial design performs a completely randomized factorial design within each stratum produced by $h(\cdot)$. Note that a completely randomized design is a special case of the stratified factorial design where the co-domain of $h(\cdot)$ is a singleton. See \cite{rubin2016} and \cite{li2020rerandomization} for further discussion of these designs. Theorem \ref{thm:block_factorial} shows that the asymptotic variance of $\hat{\Delta}_{\nu,n}$ is weakly smaller under a fully-blocked factorial design than that under \emph{any} stratified factorial design as defined above, as long as the potential outcomes satisfy the smoothness assumptions described in Assumption \ref{as:Q}(c).

\begin{theorem}\label{thm:block_factorial}
Suppose Assumptions \ref{as:Q}(a)-(b) hold and let $h: \mathrm{supp}(X) \to \{1, \ldots, S\}$ be any measurable function which maps covariate values into a set of discrete strata labels. Let $\Delta_{\nu}(Q)$ be a factorial effect for some $1\times2^K$ contrast vector $\nu$. Then under a stratified factorial design with strata defined by $h(\cdot)$,
\[\sqrt{n}(\hat\Delta_{\nu,n}- \Delta_\nu(Q)) \stackrel{d}{\to} N(0, \mathbb \sigma^2_{h,\nu})~,\]
where $\sigma^2_{h,\nu} = \nu\mathbb{V}_{h}\nu'$, with
\begin{align*}
\mathbb V_h &:= \mathbb V_{h,1} + \mathbb V_{h,2} \\
\mathbb V_{h,1} &:= \diag(E[\var[Y_i(d) | h(X_i)]]: d \in \mathcal D) \\
\mathbb V_{h,2} &:= \left[\frac{1}{|\mathcal D|} \cov[E[Y_i(d) | h(X_i)], E[Y_i(d') | h(X_i)]]\right]_{d,d'\in \mathcal{D}}~.
\end{align*}
Moreover, 
\[\sigma^2_{\nu} \le \sigma^2_{h,\nu}~,\]
where $\sigma^2_\nu = \mathbb V_\nu$ (as defined in Theorem \ref{thm:main_delta}) is the asymptotic variance of $\hat{\Delta}_{\nu,n}$ (under Assumptions \ref{as:Q}--\ref{as:close}) for a fully-blocked factorial design.
\end{theorem}

\begin{remark}\label{rem:re-randomization}
\cite{rubin2016} and \cite{li2020rerandomization} propose re-randomization designs in the context of factorial experiments which are also shown to have favorable properties relative to complete and stratified factorial designs. In Section \ref{sec:sims-mse}, we compare the mean-squared error of the fully-blocked design to a re-randomized design via Monte Carlo simulation.
\end{remark}

Our next result considers settings where only a subset of the factors are of primary interest to the researcher. For instance, \cite{besedevs2012age} use a factorial design to study how the number of options in an agent's choice set affects their ability to make optimal decisions. Here the primary factor of interest is the number of options (four or thirteen), but the design also features other secondary factors. In such a case we might imagine that a matched pairs design which focuses on the factor of primary interest and assigns the other factors by i.i.d.\ coin flips may be more efficient for estimating the primary factorial effect than the fully-blocked design which treats all the factors symmetrically. In particular, we consider a setting where we are interested in the average main effect on the $k$th factor, $\Delta_{\nu_k}(Q)$, and compare the performance of the fully-blocked design to a design which performs matched pairs over the $k$th factor while assigning the other factors to individuals at random using i.i.d.\ Bernoulli(1/2) assignment. We call such a design the ``factor $k$ specific" matched pairs design. Formally, let
\[ \zeta_j = \zeta_j(X^{(n)}) \subset \{1, \dots, 2^K n\},~ 1 \leq j \leq 2^{K - 1} n \]
denote a partition of the set of indices such that each $\zeta_j$ contains two units. The ``factor $k$ specific" matched pairs design satisfies the following assumption:

\begin{assumption} \label{ass:kspecific}
   Treatment status is assigned so that $\{Y^{(n)}(d): d \in \mathcal D\} \independent D^{(n)} | X^{(n)}$ and, conditional on $X^{(n)}$,
\[ \{(\iota_k(D_i): i \in \zeta_j): 1 \leq j \leq 2^{K - 1} n\} \]
are i.i.d.\ and each uniformly distributed over $\{(-1, +1), (+1, -1)\}$. Furthermore, independently of $X^{(n)}$ and independently across $1 \leq j \leq K, j \neq k$, $\iota_j(D_i)$ is i.i.d.\ across $1 \leq i \leq 2^K n$ and $P \{\iota_j(D_i) = -1\} = P \{\iota_j(D_i) = +1\} = \frac{1}{2}$. 
\end{assumption}

Theorem \ref{thm:matched-pair} shows that the asymptotic variance of $\hat{\Delta}_{\nu_1,n}$ is weakly smaller under a fully-blocked design than that under the factor specific matched pairs design.

\begin{theorem}\label{thm:matched-pair}
Suppose Assumptions \ref{as:Q}--\ref{as:close} hold and the treatment assignment mechanism satisfies Assumption \ref{ass:kspecific}. Then,
\[ \sqrt n(\hat \Delta_{\nu_k,n} - \Delta_{\nu_k}(Q)) \stackrel{d}{\to} N(0, \mathbb{V}_{\nu_k} + \xi_1 + \xi_0 )~, \]
where $\mathbb{V}_{\nu_k}$ is defined in Theorem \ref{thm:main_delta}, and
\begin{align*}
\xi_1 &=  \sum_{d \in \mathcal{D}:\iota_k(d) = +1} E \left[\left(\Gamma_d(X_i) - \frac{1}{2^{K-1}}\sum_{d' \in\mathcal{D}:\iota_k(d') = +1} \Gamma_{d'}(X_i)\right)^2\right] \\
\xi_0 &=  \sum_{d \in \mathcal{D}:\iota_k(d) = -1}E\left[\left(\Gamma_d(X_i) - \frac{1}{2^{K-1}}\sum_{d' \in\mathcal{D}:\iota_k(d') = -1} \Gamma_{d'}(X_i)\right)^2\right].
\end{align*}
\end{theorem}

\begin{remark}
In this section we have presented results for ``full" factorial designs, which assign individuals to every possible combination of factors. This is in contrast to ``fractional" factorial designs, which assign only a subset of the possible factor combinations \citep[see for example][]{wu2011experiments,pashley2019causal}. We leave possible extensions of our procedure to the fractional case for future work.
\end{remark}


\section{Simulations}\label{sec:sims}
In this section we examine the finite sample performance of the estimator $\hat\Delta_{\nu,n}$ and the test $\phi^\nu_n(Z^{(n)})$ in the context of a $2^K$ factorial experiment, under various alternative experimental designs. In Sections \ref{sec:sims-mse} and \ref{sec:sims-inference} the data generating processes are as specified below (in Section \ref{sec:sims-multcovs} we study an alternative design with multiple covariates and factors). For $d = (d^{(1)}, d^{(2)}) \in \{-1, 1\}^2$ and $1\leq i \leq 4 n$, the potential outcomes are generated according to the equation:
\begin{equation*}
Y_i(d) = \mu_d + \mu_d(X_i) + \sigma_d(X_i) \epsilon_{i}~.
\end{equation*}
In each of the specifications, $((X_i, \epsilon_{i}): 1\leq i \leq 4 n)$ are i.i.d; for $1 \leq i \leq 4n$, $X_i$ and $\epsilon_{i}$ are independent.
\begin{enumerate}[{\bf Model} 1:]
	\item $\mu_{1, a}(X_i) = \mu_{-1, a}(X_i) = \gamma X_i$ for $a \in \{-1, 1\}$, where $\gamma = 1$. $\mu_{1,1} = 2\mu_{1,-1} = 4\mu_{-1,1} = 2\tau$ for a parameter $\tau \in \{0, 0.2\}$, $\mu_{-1, -1} = 0$, $\epsilon_{i} \sim N(0, 1)$ and $X_i \sim N(0, 1)$ for all $d \in \{-1, 1\}^2$ and $\sigma_d(X_i) = 1$. 
	\item As in Model 1, but $\mu_d(X_i) = X_i + (X_i^2 - 1)/3$.
	\item As in Model 1, but $\mu_d(X_i) = \gamma_d X_i + (X_i^2 - 1)/3$. $\gamma_{1,1}=2$, $\gamma_{-1,1}=1$, $\gamma_{1,-1}=1/2$ and $\gamma_{-1,-1}=-1$.
	\item As in Model 3, but $\mu_d(X_i) = \sin(\gamma_d X_i)$.
	\item As in Model 3, $\mu_d(X_i) = \sin(\gamma_d X_i) + \gamma_d X_i/10 + (X_i^2 - 1)/3$.
	\item As in Model 3, but $\sigma_d(X_i) = (1 + d^{(1)} + d^{(2)})X_i^2$.
\end{enumerate}
 We consider five parameters of interest as listed in Table \ref{table:estimands}. $\Delta_{\nu_1}(Q)$ and $\Delta_{\nu_2}(Q)$ correspond to the main factorial effects for the two factors. $\Delta_{\nu_{1,2}}(Q)$ corresponds to the interaction effect between the two factors, as discussed in Example \ref{ex:2-factor}. $\Delta_{\nu_1^1}(Q)$ and $\Delta_{\nu_{-1}^1}(Q)$ denote the average effect of one factor, keeping the value of the other factor fixed at $1$ or $-1$. All simulations are performed with a sample of size $4n = 1000$. 
\begin{table}[ht!]
\centering
\setlength{\tabcolsep}{8pt}
\renewcommand{\arraystretch}{1.5}
\begin{tabular}{cc}
\toprule
Parameter of interest & Formula \\ \midrule
$\frac{1}{2}\Delta_{\nu_1}(Q)$      &   $ \frac{1}{2} E[Y_i(1, 1) - Y_i(-1, 1)] + \frac{1}{2} E[Y_i(1, -1) - Y_i(-1, -1)] $      \\
$\frac{1}{2}\Delta_{\nu_2}(Q)$   & $ \frac{1}{2} E[Y_i(1, 1) - Y_i(1, -1)] + \frac{1}{2} E[Y_i(-1, 1) - Y_i(-1, -1)]$  \\
$\frac{1}{2}\Delta_{\nu_{1,2}}(Q)$ & $\frac{1}{2} E[Y_i(1, 1) - Y_i(-1, 1)] - \frac{1}{2} E[Y_i(1, -1) - Y_i(-1, -1)]$ \\
$\Delta_{\nu_1^1}(Q)$      &   $E[Y_i(1, 1) - Y_i(-1, 1)]$      \\
$\Delta_{\nu_{-1}^1}(Q)$      &        $E[Y_i(1, -1) - Y_i(-1, -1)]$ \\
\bottomrule
\end{tabular}
\caption{Parameters of interest}
\label{table:estimands}
\end{table}
\subsection{MSE Properties of the Matched Tuples Design}\label{sec:sims-mse}
In this section, we study the mean-squared-error performance of $\hat\Delta_{\nu,n}$ across several experimental designs. We analyze and compare the MSE for all five parameters of interest for the following seven experimental designs:
\begin{enumerate}
	\item \textbf{(B-B)} $(D_i^{(1)}, D_i^{(2)})$ are i.i.d.\ across $1 \leq i \leq 4n$ and the two entries are independently distributed as $2A - 1$, where $A$ follows Bernoulli$(1/2)$.
	\item \textbf{(C)} $(D_i^{(1)}, D_i^{(2)})$ are jointly drawn from a completely randomized design. We uniformly at random divide the experimental sample of size $4n$ into four groups of size $n$ and assign a different $d \in \{-1, 1\}^2$ for each group.
	\item \textbf{(MP-B)} A matched-pair design for $D^{(1)}$, where units are ordered and paired according to $X_i$. For each pair, uniformly at random assign $D_i^{(1)} = 1$ to one of the units. Independently, $(D_i^{(2)}: 1 \leq i \leq 4n)$ are i.i.d.\ with the distribution of $2A - 1$, where $A \sim \mathrm{Bernoulli}(1/2)$.
	\item \textbf{(MT)} Matched tuples design where units are ordered according to $X_i$.
	\item \textbf{(Large-2)} A stratified design, where the experimental sample is divided into two strata using the median of $X_i$ as the cutoff. In each stratum, treatment is assigned as in $\textbf{C}$. 
	\item \textbf{(Large-4)} As in \textbf{(Large-2)}, but with four strata. 
	\item \textbf{(RE)} A re-randomization design using a Mahalanobis balance function. As outlined in \cite{rubin2016}, we select the main-effect threshold criterion to be the $100(0.01^{1/K})$ percentile of a $\chi^2_{p}$ distribution with $p=\mathrm{dim}(X_i)$, and select the interaction-effect threshold criterion to be $100(0.01^{1/L})$, where $L$ is the number of interaction effects.
\end{enumerate}

Table \ref{table:mse} displays the ratio of the MSE of each design relative to the MSE of \textbf{MT}, computed across 4,000 Monte Carlo replications. In each of the designs, we set treatment effects to zero by setting $\tau=0$. As expected from Theorems \ref{thm:block_factorial} and \ref{thm:matched-pair}, \textbf{MT} outperforms \textbf{B-B}, \textbf{C}, \textbf{MP-B}, \textbf{Large-2}, and \textbf{Large-4} in every model specification. We also find that \textbf{MT} compares favorably to \textbf{RE}, with \textbf{RE} slightly outperforming \textbf{MT} in some cases, but with \textbf{MT} outperforming in general. Although we do not have formal results comparing the matched tuples design to re-randomization, we note that re-randomization redraws treatments until the distances between certain features of the covariate distribution across treatment statuses are below certain pre-specified thresholds. In contrast, the matched tuples design attempts to \emph{minimize} these distances by blocking units finely based on the covariates. See also Remark 3 of \cite{bai2022optimality} for a related observation in the binary treatment setting.

\begin{table}[ht!]
\centering
\setlength{\tabcolsep}{5pt}
\begin{adjustbox}{max width=0.75\linewidth,center}
\begin{tabular}{lllllllll}
\toprule
Model              & Parameter & \textbf{B-B} & \textbf{C} & \textbf{MP-B} & \textbf{MT} & \textbf{Large-2} & \textbf{Large-4} & \textbf{RE}    \\
\midrule
\multirow{5}{*}{1} &   $ \Delta_{\nu_1}$       & 2.099 & 1.948 & 1.045 & 1.000 & 1.335 & 1.138 & 1.031 \\
& $ \Delta_{\nu_2}$ & 2.036 & 2.015 & 2.113 & 1.000 & 1.407 & 1.179 & 0.988 \\
& $ \Delta_{\nu_{1,2}}$ & 2.008 & 2.044 & 2.016 & 1.000 & 1.423 & 1.091 & 1.014 \\
&   ${\Delta}_{\nu_1^1}$        & 2.051 & 2.014 & 1.563 & 1.000 & 1.402 & 1.134 & 1.029 \\
&   ${\Delta}_{\nu_{-1}^1}$       &  2.057 & 1.978 & 1.498 & 1.000 & 1.357 & 1.095 & 1.017 \\
\\
\multirow{5}{*}{2} &   $ \Delta_{\nu_1}$       & 2.327 & 2.168 & 1.044 & 1.000 & 1.546 & 1.249 & 1.232 \\
& $ \Delta_{\nu_2}$ & 2.254 & 2.259 & 2.355 & 1.000 & 1.619 & 1.312 & 1.209 \\
& $ \Delta_{\nu_{1,2}}$ & 2.249 & 2.287 & 2.173 & 1.000 & 1.646 & 1.225 & 1.250 \\
&   ${\Delta}_{\nu_1^1}$        & 2.285 & 2.265 & 1.634 & 1.000 & 1.599 & 1.260 & 1.227 \\
&   ${\Delta}_{\nu_{-1}^1}$       &  2.291 & 2.190 & 1.585 & 1.000 & 1.593 & 1.215 & 1.255 \\
\\
\multirow{5}{*}{3} &   $ \Delta_{\nu_1}$       & 2.042 & 1.996 & 1.792 & 1.000 & 1.422 & 1.206 & 1.124 \\
& $ \Delta_{\nu_2}$ & 1.576 & 1.527 & 1.480 & 1.000 & 1.221 & 1.140 & 1.109 \\
& $ \Delta_{\nu_{1,2}}$ & 3.113 & 2.982 & 1.943 & 1.000 & 1.900 & 1.337 & 1.187 \\
&   ${\Delta}_{\nu_1^1}$        & 3.401 & 3.351 & 2.237 & 1.000 & 1.979 & 1.410 & 1.225 \\
&   ${\Delta}_{\nu_{-1}^1}$       & 1.899 & 1.802 & 1.619 & 1.000 & 1.388 & 1.166 & 1.103 \\
\\
\multirow{5}{*}{4} &   $ \Delta_{\nu_1}$       & 1.311 & 1.305 & 1.252 & 1.000 & 1.100 & 1.070 & 1.194 \\
& $ \Delta_{\nu_2}$ & 1.218 & 1.210 & 1.167 & 1.000 & 1.063 & 1.064 & 1.057 \\
& $ \Delta_{\nu_{1,2}}$ & 1.296 & 1.289 & 1.152 & 1.000 & 1.184 & 1.084 & 1.191 \\
&   ${\Delta}_{\nu_1^1}$        & 1.416 & 1.401 & 1.259 & 1.000 & 1.158 & 1.080 & 1.249 \\
&   ${\Delta}_{\nu_{-1}^1}$       &  1.201 & 1.202 & 1.150 & 1.000 & 1.128 & 1.075 & 1.140 \\
\\
\multirow{5}{*}{5} &   $ \Delta_{\nu_1}$       & 1.603 & 1.606 & 1.315 & 1.000 & 1.280 & 1.169 & 1.375 \\
& $ \Delta_{\nu_2}$ & 1.444 & 1.458 & 1.378 & 1.000 & 1.225 & 1.173 & 1.235 \\
& $ \Delta_{\nu_{1,2}}$ & 1.607 & 1.598 & 1.351 & 1.000 & 1.370 & 1.184 & 1.390 \\
&   ${\Delta}_{\nu_1^1}$        & 1.802 & 1.797 & 1.415 & 1.000 & 1.353 & 1.192 & 1.441 \\
&   ${\Delta}_{\nu_{-1}^1}$       &  1.434 & 1.434 & 1.262 & 1.000 & 1.301 & 1.164 & 1.332 \\
\\
\multirow{5}{*}{6} &   $ \Delta_{\nu_1}$       & 1.119 & 1.122 & 1.116 & 1.000 & 1.055 & 1.021 & 1.065 \\
& $ \Delta_{\nu_2}$ & 1.051 & 1.042 & 1.056 & 1.000 & 1.026 & 0.991 & 0.989 \\
& $ \Delta_{\nu_{1,2}}$ & 1.107 & 1.104 & 1.077 & 1.000 & 1.074 & 0.994 & 1.018 \\
&   ${\Delta}_{\nu_1^1}$        & 1.096 & 1.100 & 1.088 & 1.000 & 1.058 & 1.005 & 1.051 \\
&   ${\Delta}_{\nu_{-1}^1}$       &  1.197 & 1.177 & 1.137 & 1.000 & 1.092 & 1.017 & 0.996 \\
\bottomrule
\end{tabular}
\end{adjustbox}
\caption{Ratio of MSEs relative to MT}
\label{table:mse}
\end{table}

\subsection{Inference}\label{sec:sims-inference}
In this section, we study the finite sample properties of several different tests of the null hypothesis $H_0: \Delta_{\nu} = 0$ for various choices of $\nu$, against the alternative hypotheses implied by setting $\tau = 0.2$. In this section we restrict our attention to five assignment mechanisms: \textbf{B-B}, \textbf{C}, \textbf{MT}, \textbf{Large-2} and \textbf{Large-4}. We exclude \textbf{MP-B} because it is a non-standard experimental design for which we have not developed an inference procedure. We also exclude the re-randomization design (\textbf{RE}) because, although it is a widely studied design, the inferential results in \cite{li2020rerandomization} are derived in a finite population framework which is distinct from our super-population framework, and their resulting limiting distribution is non-normal. 

In each case we perform our hypothesis tests at a significance level of $0.05$. For design \textbf{B-B}, tests are performed using a standard $t$-test. For designs \textbf{C}, \textbf{Large-2} and \textbf{Large-4} the tests are constructed using the asymptotic normality result from Theorem \ref{thm:block_factorial} combined with variance estimators constructed using the same plug-in method as in \cite{bugni2018} and \cite{bugni2019inference}. For design {\bf MT} the test is constructed as described in Theorem \ref{thm:V_const}. Table \ref{table:rej.prob} displays the rejection probabilities under the null and alternative hypotheses, computed from 2,000 Monte Carlo replications. The results show that the rejection probabilities are universally around 0.05 under the null hypothesis, which verifies the validity of our tests across all the designs. Under the alternative hypotheses implied by $\tau=0.2$, the rejection probabilities vary substantially across the different designs, outcome models and parameters. However, our matched tuples design displays the highest power for almost all parameters and model specifications. 

\begin{table}[ht!]
\centering
\setlength{\tabcolsep}{4pt}
\begin{adjustbox}{max width=0.9\linewidth,center}
\begin{tabular}{lllllllllllll}
\toprule
&       &  \multicolumn{5}{c}{Under $H_0$} & & \multicolumn{5}{c}{Under $H_1$} \\ \cmidrule{3-7} \cmidrule{9-13}
Model              & Parameter & \textbf{B-B} & \textbf{C} & \textbf{MT} & \textbf{Large-2} & \textbf{Large-4} &  & \textbf{B-B} & \textbf{C} & \textbf{MT} & \textbf{Large-2} & \textbf{Large-4}    \\ \midrule
\multirow{5}{*}{1} &   $ \Delta_{\nu_1}$       & 0.057 & 0.049 & 0.051 & 0.050 & 0.046 &   & 0.790 & 0.803 & 0.977 & 0.915 & 0.963 \\
& $ \Delta_{\nu_2}$ & 0.052 & 0.059 & 0.046 & 0.060 & 0.058 &  & 0.371 & 0.403 & 0.675 & 0.534 & 0.593 \\
& $ \Delta_{\nu_{1,2}}$ & 0.049 & 0.059 & 0.049 & 0.059 & 0.043 &   & 0.081 & 0.093 & 0.126 & 0.100 & 0.106 \\
&   ${\Delta}_{\nu_1^1}$        & 0.052 & 0.043 & 0.048 & 0.064 & 0.040 &   & 0.646 & 0.656 & 0.921 & 0.816 & 0.884 \\
&   ${\Delta}_{\nu_{-1}^1}$       &  0.056 & 0.051 & 0.044 & 0.057 & 0.048 &   & 0.361 & 0.333 & 0.594 & 0.499 & 0.545 \\
\\
\multirow{5}{*}{2} &   $ \Delta_{\nu_1}$       & 0.053 & 0.043 & 0.049 & 0.048 & 0.045 &   & 0.738 & 0.737 & 0.976 & 0.875 & 0.951 \\
& $ \Delta_{\nu_2}$ & 0.056 & 0.061 & 0.046 & 0.059 & 0.056 &   & 0.341 & 0.377 & 0.670 & 0.483 & 0.551 \\
& $ \Delta_{\nu_{1,2}}$ & 0.052 & 0.065 & 0.050 & 0.060 & 0.044 &  & 0.082 & 0.091 & 0.126 & 0.101 & 0.095 \\
&   ${\Delta}_{\nu_1^1}$        & 0.049 & 0.051 & 0.046 & 0.057 & 0.036 &   & 0.597 & 0.610 & 0.919 & 0.758 & 0.840 \\
&   ${\Delta}_{\nu_{-1}^1}$       &  0.056 & 0.051 & 0.046 & 0.054 & 0.048 &   & 0.340 & 0.310 & 0.598 & 0.436 & 0.500 \\
\\
\multirow{5}{*}{3} &   $ \Delta_{\nu_1}$       & 0.054 & 0.056 & 0.050 & 0.053 & 0.052 &   & 0.571 & 0.570 & 0.837 & 0.705 & 0.787 \\
& $ \Delta_{\nu_2}$ & 0.056 & 0.057 & 0.056 & 0.057 & 0.059 &   & 0.235 & 0.259 & 0.361 & 0.286 & 0.323 \\
& $ \Delta_{\nu_{1,2}}$ & 0.051 & 0.051 & 0.052 & 0.062 & 0.047 &   & 0.060 & 0.064 & 0.116 & 0.091 & 0.082 \\
&   ${\Delta}_{\nu_1^1}$        & 0.048 & 0.051 & 0.046 & 0.061 & 0.035 &   & 0.402 & 0.421 & 0.885 & 0.624 & 0.762 \\
&   ${\Delta}_{\nu_{-1}^1}$       &  0.061 & 0.047 & 0.060 & 0.056 & 0.057 &   & 0.255 & 0.234 & 0.374 & 0.310 & 0.340 \\
\\
\multirow{5}{*}{4} &   $ \Delta_{\nu_1}$       & 0.049 & 0.051 & 0.045 & 0.045 & 0.050 &   & 0.908 & 0.905 & 0.968 & 0.956 & 0.957 \\
& $ \Delta_{\nu_2}$ & 0.051 & 0.052 & 0.051 & 0.051 & 0.058 &   & 0.488 & 0.520 & 0.604 & 0.569 & 0.559 \\
& $ \Delta_{\nu_{1,2}}$ & 0.056 & 0.052 & 0.049 & 0.065 & 0.045 &   & 0.092 & 0.102 & 0.126 & 0.117 & 0.111 \\
&   ${\Delta}_{\nu_1^1}$        & 0.050 & 0.048 & 0.051 & 0.054 & 0.045 &   & 0.762 & 0.785 & 0.908 & 0.865 & 0.886 \\
&   ${\Delta}_{\nu_{-1}^1}$       &  0.044 & 0.055 & 0.048 & 0.052 & 0.046 &   & 0.498 & 0.472 & 0.544 & 0.528 & 0.523 \\
\\
\multirow{5}{*}{5} &   $ \Delta_{\nu_1}$       & 0.054 & 0.054 & 0.045 & 0.045 & 0.043 &   & 0.844 & 0.847 & 0.964 & 0.912 & 0.937 \\
& $ \Delta_{\nu_2}$ & 0.053 & 0.056 & 0.051 & 0.048 & 0.053 &   & 0.416 & 0.445 & 0.589 & 0.491 & 0.505 \\
& $ \Delta_{\nu_{1,2}}$ & 0.052 & 0.054 & 0.049 & 0.059 & 0.049 &   & 0.092 & 0.099 & 0.124 & 0.110 & 0.099 \\
&   ${\Delta}_{\nu_1^1}$        & 0.051 & 0.052 & 0.049 & 0.058 & 0.043 &   & 0.674 & 0.688 & 0.911 & 0.810 & 0.847 \\
&   ${\Delta}_{\nu_{-1}^1}$       &  0.050 & 0.062 & 0.049 & 0.056 & 0.049 &   & 0.416 & 0.403 & 0.523 & 0.461 & 0.474 \\
\\
\multirow{5}{*}{6} &   $ \Delta_{\nu_1}$       & 0.050 & 0.050 & 0.043 & 0.058 & 0.043 &   & 0.129 & 0.128 & 0.122 & 0.115 & 0.130 \\
& $ \Delta_{\nu_2}$ & 0.053 & 0.059 & 0.057 & 0.057 & 0.051 &   & 0.074 & 0.086 & 0.088 & 0.079 & 0.080 \\
& $ \Delta_{\nu_{1,2}}$ & 0.047 & 0.046 & 0.052 & 0.053 & 0.044 &   & 0.052 & 0.046 & 0.052 & 0.057 & 0.050 \\
&   ${\Delta}_{\nu_1^1}$        & 0.049 & 0.046 & 0.049 & 0.051 & 0.043 &   & 0.082 & 0.083 & 0.077 & 0.082 & 0.081 \\
&   ${\Delta}_{\nu_{-1}^1}$       &  0.059 & 0.056 & 0.058 & 0.059 & 0.056 &   & 0.140 & 0.113 & 0.125 & 0.131 & 0.135 \\
\bottomrule
\end{tabular}
\end{adjustbox}
\caption{Rejection probabilities under the null and alternative hypothesis}
\label{table:rej.prob}
\end{table}

\subsection{Experiments with More Factors and Covariates}\label{sec:sims-multcovs}
In this section we repeat the previous simulation exercises while varying the number of factors $K$ and the number of observed covariates $\mathrm{dim}(X_i)$. The data generating process is constructed as follows:
\begin{equation*}
Y_i(d) = \begin{cases} \tau d^{(1)} +  \tilde{X}_i' \beta + \epsilon_{i}, & \text { if } K=1 \\ \tau\cdot\left(d^{(1)} + \frac{\sum_{k=2}^K d^{(k)} }{K-1}\right) + \gamma_{d}\tilde{X}_i' \beta + \epsilon_{i}, & \text { if } K \geq 2 \end{cases}
\end{equation*}
where $\tau \in \{0, 0.1\}$, $d = (d^{(1)}, \ldots, d^{(K)})$ and $d^{(k)}\in \{-1, 1\}$ represents the treatment status of the $k$-th factor. We set $\gamma_{d} = 1$ if $d^{(2)}=1$, $\gamma_{d} = -1$ otherwise, in order to ensure the conditional means are heterogeneous in the second factor. $\tilde{X}_i$ contains $9$ covariates, out of which the first $\text{dim}(X_i)$ covariates are observed and used for the experimental designs. The distributions of $\tilde{X}_i, \epsilon_i$ and the values of $\beta$ are calibrated using data obtained from \cite{rubin2016}, who study the covariate balancing properties of $2^K$ factorial re-randomization designs using data from the New York Department of Education (NYDE). Details on the empirical context and construction of the data generating process are provided in Appendix \ref{sec:calibrated_details}.

To construct our matched tuples of size $2^K$ when $\mathrm{dim}(X_i) > 1$, we employ the recursive pairing algorithm described in Section \ref{sec:setup} using the Mahalanobis distance. We emphasize, however, that this approach is not guaranteed to be optimal, and we leave the study of potentially more effective matching algorithms to future work. 

In addition to the standard matched tuples design (\textbf{MT}), we also include a matched tuples design with a \emph{ replicate} for each treatment as described in Section \ref{sec:replicate}, denoted by \textbf{MT2}. For example, in the \textbf{MT2} design with two factors, units are matched into groups of \emph{eight}, and two units receive each factor combination. We also continue to consider the alternative designs (\textbf{C}, \textbf{Large-4}, \textbf{MP-B} and \textbf{RE}) from Section \ref{sec:sims-mse}. When constructing the strata for \textbf{Large-4}, we stratify on one covariate drawn at random from the set of available covariates.



In Table \ref{table:more-factor-x-mse} we report the ratio of the MSE of each design relative to the MSE of {\bf MT} when $\text{dim}(X_i) = 1$ and $K = 1$ (computed from 4,000 Monte Carlo replications). For all experiments in this section, the number of observations is fixed to be 1,280 so that we have 20 matched tuples of size 64 when $K=6$. Our simulation results are consistent with those in Section \ref{sec:sims-mse}: \textbf{MT} displays the lowest MSE across almost all model specifications. Although \textbf{MT2} generally produces larger MSEs than \textbf{MT}, it still performs favorably relative to the other designs. For methods that use an increasing number of covariates when $\mathrm{dim}(X_i)$ increases (\textbf{MT}, \textbf{MT2}, \textbf{MP-B} and \textbf{RE}), we observe that the MSE in fact \emph{increases} with the number of available covariates. We expect this is because (as shown in Appendix \ref{sec:calibrated_details}) the first covariate is a much stronger predictor of the control outcome than the other available covariates, which are relatively uninformative.

\begin{table}[ht!]
\centering
\setlength{\tabcolsep}{2.5pt}
\begin{adjustbox}{max width=\linewidth,center}
\begin{tabular}{ccccccccccccccccc}
\toprule
$\mathrm{dim}(X_i)$  & Method & $K=1$ & $K=2$ & $K=3$ & $K=4$ & $K=5$ & $K=6$ & & Method & $K=1$ & $K=2$ & $K=3$ & $K=4$ & $K=5$ & $K=6$ \\ 
\midrule
1 & \multirow{5}{*}{\textbf{MT}} & 1.000 & 1.003 & 1.006 & 1.113 & 1.297 & 1.945 & & \multirow{5}{*}{\textbf{C}} & 9.151 & 8.554 & 8.642 & 8.939 & 9.015 & 9.181 \\
2 &  & 1.027 & 1.052 & 1.107 & 1.180 & 1.463 & 2.293  & & &  9.120 & 8.528 & 8.568 & 8.867 & 9.053 & 9.114 \\
4 &  & 1.043 & 1.130 & 1.420 & 1.687 & 2.170 & 3.338 & & &  8.968 & 8.364 & 8.569 & 8.868 & 8.949 & 8.765 \\
6 &  & 1.192 & 1.495 & 1.763 & 2.241 & 3.097 & 4.304 & &  & 8.945 & 8.327 & 8.588 & 8.994 & 9.081 & 8.853 \\
9 & & 1.284 & 1.702 & 2.047 & 2.781 & 3.337 & 4.081 & &  & 8.934 & 8.309 & 8.600 & 8.788 & 8.915 & 8.526 \\
\\
1 & \multirow{5}{*}{\textbf{MT2}} & 1.017 & 1.049 & 1.074 & 1.297 & 1.916 & 2.903 & & \multirow{5}{*}{\textbf{Large-4}} & 4.393 & 4.605 & 4.674 & 4.634 & 4.393 & 4.381  \\
2 &  & 1.044 & 1.086 & 1.212 & 1.547 & 2.200 & 3.585 & & &  6.523 & 6.926 & 6.745 & 6.704 & 6.521 & 6.367 \\
4 &  & 1.224 & 1.332 & 1.620 & 2.231 & 3.379 & 4.799 & & &  7.321 & 8.100 & 7.407 & 7.559 & 7.542 & 7.399 \\
6 &  & 1.451 & 1.901 & 2.339 & 3.061 & 4.020 & 5.721  & &  & 8.143 & 8.137 & 7.644 & 7.801 & 8.288 & 7.906 \\
9 & & 1.609 & 2.140 & 2.693 & 3.231 & 4.387 & 6.903  & &  & 8.093 & 8.075 & 8.170 & 7.799 & 8.129 & 8.402 \\
\\
1 & \multirow{5}{*}{\textbf{MP-B}} & 0.991 & 8.693 & 8.807 & 8.964 & 8.991 & 8.829  & & \multirow{5}{*}{\textbf{RE}} & 1.073 & 1.091 & 1.296 & 2.032 & 3.040 & 3.640 \\
2 &  & 0.978 & 8.854 & 8.897 & 8.863 & 8.811 & 9.072 & & & 1.090 & 1.069 & 1.955 & 3.284 & 4.282 & 5.094 \\
4 &  & 0.967 & 8.970 & 8.711 & 9.020 & 8.855 & 8.749 & & &  1.320 & 1.410 & 3.278 & 4.640 & 5.504 & 6.270  \\
6 &  & 1.175 & 9.148 & 8.753 & 8.941 & 8.774 & 8.596  & &  & 1.961 & 1.886 & 3.976 & 5.648 & 6.223 & 6.759 \\
9 & & 1.227 & 8.793 & 8.989 & 9.444 & 9.227 & 8.273 & &  & 2.515 & 2.566 & 4.957 & 6.265 & 6.676 & 7.455  \\
\bottomrule
\end{tabular}
\end{adjustbox}
\caption{Ratio of MSEs relative to MT using a single factor and covariate}
\label{table:more-factor-x-mse}
\end{table}

In Table \ref{table:more-factor-x}, we compute the rejection probabilities when testing the null hypothesis $H_0: \Delta_{\nu_1} = 0$ against the alternative implied by setting $\tau = 0.1$, for various choices of $K$ and $\mathrm{dim}(X_i)$ (computed from 1,000 Monte Carlo replications). Under the null hypothesis, we observe that our tests under design \textbf{MT} become conservative as $\mathrm{dim}(X_i)$ and $K$ increase. In particular, we notice a large difference between $K = 4$ and $K = 5$. However, despite being conservative, \textbf{MT} still displays favorable power properties relative to \textbf{C} and \textbf{Large-4} for all but the largest choices of $K$.

Our next observation is that our tests under design \textbf{MT2} remain exact even as $\mathrm{dim}(X_i)$ and $K$ both increase. As we explain in Section \ref{sec:replicate}, we suspect that our challenges for inference using \textbf{MT} come from poor estimation of the variance, which seems to be alleviated in \textbf{MT2}, where the number of observations receiving each treatment within a tuple are doubled. As a result of this exactness, \textbf{MT2} achieves higher power than \textbf{MT} when $\mathrm{dim}(X_i)$ and $K$ are large. To further explore these power improvements, Figure \ref{fig:power_plots}  presents power plots for three specific choices of $K$ and $\mathrm{dim}(X_i)$ with $\tau$ ranging from 0 to 0.1 (Figure \ref{fig:power_plots2} in the appendix presents power plots for alternatives implied by larger values than $\tau = 0.1$). First, when $\mathrm{dim}(X_i)$ and $K$ are small, for instance $\mathrm{dim}(X_i)=K=1$, we observe no significant difference between the power plots generated by \textbf{MT} and \textbf{MT2}. However, when the dimension of the covariates and factors are both large, for instance $\mathrm{dim}(X_i)=6, K=4$, \textbf{MT2} dominates \textbf{MT} for all alternative hypotheses. Therefore, our recommendation to practitioners is to consider a matched tuples design when working with few treatments and covariates, but to consider the replicated design when dealing with a large number of treatments and/or covariates.

\begin{table}[ht!]
\centering
\setlength{\tabcolsep}{3pt}
\begin{adjustbox}{max width=\linewidth,center}
\begin{tabular}{ccccccccccccccc}
\toprule
& &    \multicolumn{6}{c}{Under $H_0$} & & \multicolumn{6}{c}{Under $H_1$} \\ \cmidrule{3-8} \cmidrule{10-15}
Method & $\mathrm{dim}(X_i)$ & $K=1$ & $K=2$ & $K=3$ & $K=4$ & $K=5$ & $K=6$ & & $K=1$ & $K=2$ & $K=3$ & $K=4$ & $K=5$ & $K=6$ \\ 
\midrule
\multirow{5}{*}{\textbf{MT}}& 1 & 0.049 & 0.045 & 0.033 & 0.023 & 0.009 & 0.008 & & 0.998 & 1.000 & 1.000 & 0.997 & 0.980 & 0.837  \\
& 2 & 0.047 & 0.043 & 0.041 & 0.018 & 0.008 & 0.002 & & 0.999 & 0.998 & 0.997 & 0.997 & 0.935 & 0.732 \\
& 4 & 0.040 & 0.029 & 0.031 & 0.011 & 0.009 & 0.008 & & 1.000 & 1.000 & 0.979 & 0.946 & 0.794 & 0.583 \\
& 6 & 0.037 & 0.018 & 0.010 & 0.022 & 0.010 & 0.007 & & 0.999 & 0.989 & 0.936 & 0.870 & 0.668 & 0.479 \\
& 9 & 0.041 & 0.026 & 0.016 & 0.019 & 0.014 & 0.003 & & 0.988 & 0.961 & 0.895 & 0.810 & 0.674 & 0.319 \\
\\
\multirow{5}{*}{\textbf{MT2}} & 1 & 0.054 & 0.054 & 0.044 & 0.059 & 0.047 & 0.052 & & 1.000 & 0.999 & 1.000 & 0.996 & 0.973 & 0.858 \\
& 2 & 0.048 & 0.053 & 0.041 & 0.058 & 0.039 & 0.055 & & 1.000 & 0.999 & 1.000 & 0.985 & 0.943 & 0.784 \\
& 4 & 0.075 & 0.048 & 0.054 & 0.056 & 0.060 & 0.046 & & 0.996 & 0.993 & 0.981 & 0.951 & 0.843 & 0.673 \\
& 6 & 0.053 & 0.067 & 0.046 & 0.054 & 0.045 & 0.046 & & 0.988 & 0.967 & 0.926 & 0.857 & 0.744 & 0.579 \\
& 9 & 0.065 & 0.050 & 0.053 & 0.059 & 0.060 & 0.047 & & 0.983 & 0.944 & 0.872 & 0.840 & 0.704 & 0.494 \\
\\
\multirow{5}{*}{\textbf{C}} & 1 & 0.062 & 0.054 & 0.041 & 0.056 & 0.059 & 0.069 & & 0.437 & 0.449 & 0.410 & 0.445 & 0.463 & 0.459 \\
& 2 & 0.063 & 0.049 & 0.038 & 0.051 & 0.065 & 0.068 & & 0.434 & 0.450 & 0.410 & 0.442 & 0.459 & 0.459 \\
& 4 & 0.064 & 0.050 & 0.038 & 0.048 & 0.055 & 0.057 & & 0.425 & 0.448 & 0.400 & 0.443 & 0.457 & 0.468 \\
& 6 & 0.066 & 0.052 & 0.045 & 0.048 & 0.054 & 0.055 & & 0.430 & 0.437 & 0.409 & 0.436 & 0.437 & 0.463 \\
& 9 & 0.063 & 0.042 & 0.050 & 0.033 & 0.054 & 0.048 & & 0.417 & 0.439 & 0.420 & 0.433 & 0.433 & 0.448 \\
\\
\multirow{5}{*}{\textbf{Large-4}} & 1 & 0.050 & 0.044 & 0.059 & 0.061 & 0.053 & 0.057 & & 0.685 & 0.699 & 0.701 & 0.683 & 0.730 & 0.770 \\
& 2 & 0.046 & 0.050 & 0.043 & 0.052 & 0.044 & 0.065 & & 0.560 & 0.564 & 0.575 & 0.585 & 0.582 & 0.634 \\
& 4 & 0.053 & 0.064 & 0.039 & 0.059 & 0.056 & 0.062 & & 0.497 & 0.490 & 0.486 & 0.527 & 0.521 & 0.577 \\
& 6 & 0.055 & 0.053 & 0.049 & 0.057 & 0.059 & 0.071 & & 0.462 & 0.444 & 0.495 & 0.519 & 0.520 & 0.553 \\
& 9 & 0.044 & 0.041 & 0.056 & 0.051 & 0.049 & 0.076 & & 0.457 & 0.451 & 0.493 & 0.490 & 0.511 & 0.571 \\
\toprule
\end{tabular}
\end{adjustbox}
\caption{Rejection probabilities when testing $H_0: \Delta_{\nu_1} = 0$ under the null and alternative hypothesis}
\label{table:more-factor-x}
\end{table}

\begin{figure}[ht!]
\centering
\includegraphics[width=\textwidth]{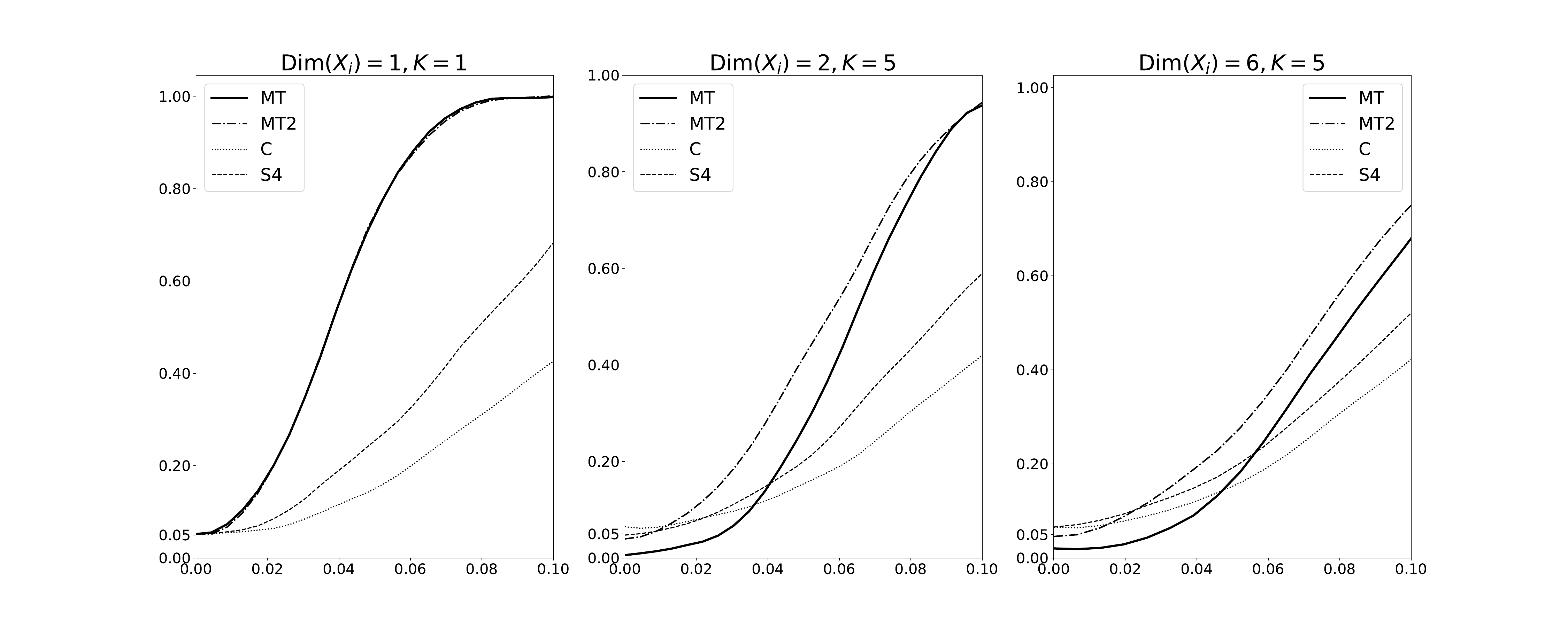}
\caption{Rejection probability under various choices of $\tau$}
\label{fig:power_plots}
\end{figure}

\section{Empirical Application}\label{sec:application}
In this section, we illustrate the inference procedures introduced in Section \ref{sec:main} using the data collected  in \cite{McKenzie2014}\footnote{The original paper features six rounds of surveys which were pooled in the final analysis. We perform our analysis exclusively on the data obtained in the sixth round in order to avoid complications related to time-series dependence across rounds. For simplicity, we additionally drop quadruplets with missing values, and 4 ``leftover'' groups whose sizes range from 5 to 8 firms. This results in a final sample of 120 quadruplets, or $4n = 480$. Further results on the long-run effects (collected in a seventh survey wave) are contained in Table \ref{table:application-wave7} in Section \ref{sec:additional-simulations} of the appendix.}. \cite{McKenzie2014} conduct a randomized experiment in order to investigate the effects of several capital aid programs on the profits of small businesses in Ghana. In their experiment, there are three treatment arms, where (in our notation) $D_i = 1$ indicates that the $i$th firm is untreated, $D_i = 2$ indicates being offered cash, and $D_i = 3$ indicates being offered in-kind grants. The null hypotheses of interest are
\begin{equation} \label{eq:h0-app}
    H_0^d: E[Y_i(1)] = E[Y_i(d)] \text{ versus } H_1: E[Y_i(1)] \neq E[Y_i(d)]
\end{equation}
for $d \in \{2, 3\}$, as well as
\begin{equation} \label{eq:h0-app-23}
H_0^{2, 3}: E[Y_i(2)] = E[Y_i(3)] \text{ versus } H_1: E[Y_i(2)] \neq E[Y_i(3)]~.
\end{equation}

In their experimental design, blocks are defined by quadruplets, where each quadruplet contains \emph{two} untreated units with $D_i = 1$, one treated unit with $D_i = 2$, and one treated unit with $D_i = 3$. Despite the slight departure from the framework presented in Sections \ref{sec:setup}--\ref{sec:main}, in that there are two untreated units in each quadruplet, we show in Appendix \ref{sec:app_details} that a slight modification of the variance estimator in Theorem \ref{thm:V_const} produces a valid test for \eqref{eq:h0-app}--\eqref{eq:h0-app-23}. Specifically, we pretend that there are four treatment levels in each quadruplet, while the first two are in fact controls. Then, by setting generating vectors $\nu^{2}=(-1/2,-1/2,1,0)$, $\nu^{3} = (-1/2, -1/2, 0, 1)$, and $\nu^{2,3}=(0,0,-1,1)$ and proceeding with the testing procedure in Theorem \ref{thm:V_const}, we obtain valid tests for $H_0^d$ and $H_0^{2,3}$. For each of the hypotheses in \eqref{eq:h0-app}--\eqref{eq:h0-app-23}, we implement the following tests:
\begin{itemize}
\item[---] A $t$-test based on the OLS estimator in a linear regression of $Y$ on $1$, $I \{D_i = 2\}$, and $I \{D_i = 3\}$, together with the usual heteroskedasticity-robust variance estimator.
\item[---] The test introduced in Proposition \ref{prop:application}, which implements the test from Theorem \ref{thm:V_const} as described above to accommodate for the fact that there are two untreated units in each block.
\end{itemize}
We note that \cite{McKenzie2014} test \eqref{eq:h0-app} and \eqref{eq:h0-app-23} using a $t$-test obtained from a linear regression of outcomes on treatment indicators and block fixed effects. However, as was shown in Theorem \ref{thm:sfe}, such a procedure is not guaranteed to be valid. On the other hand, we expect that the $t$-test obtained from a linear regression without block fixed effects should be conservative for testing \eqref{eq:h0-app}--\eqref{eq:h0-app-23}  in light of the observations made in Example \ref{ex:matched-triples} and the fact that this test coincides with a standard two-sample $t$-test.

Our results are presented in Table \ref{table:application-wave6}. The point estimates of the two methods are identical because the OLS estimator coincides with the difference-in-means estimator. However, the standard errors obtained from our variance estimator are always smaller than the heteroskedasticy-robust standard errors. For example, when testing \eqref{eq:h0-app} for $d = 3$ among the female subsample, the standard error produced from our variance estimator is 15.21 whereas the heteroskedasticy robust standard error is 18.13. We note that overall the improvements are modest; this suggests that the conditional expectation of the outcomes does not vary substantially with the observable characteristics in this survey wave. This is further corroborated by the calibrated simulations presented in Table \ref{table:finite_population_empirical} in Appendix \ref{sec:additional-simulations}.

\begin{table}[ht!]
\centering
\setlength{\tabcolsep}{4pt}
\caption{Point estimates and standard errors for testing the treatment effects of cash and in-kind grants using different methods (wave 6)}
\begin{adjustbox}{max width=\linewidth,center}
\begin{tabular}{cccccccccccc}
\toprule
& & &   \multicolumn{5}{c}{} & & High Initial & & Low Initial \\ \cmidrule{4-8}  \cmidrule{10-10} \cmidrule{12-12} 
&  &  & All Firms & & Males &  & Females & & Profit Women & & Profit Women \\ \cmidrule{4-4} \cmidrule{6-6} \cmidrule{8-8} \cmidrule{10-10} \cmidrule{12-12}
&  &  & (1) & & (2) &  & (3) & & (4) & & (5) \\
\midrule
& &Cash treatment & 19.64  & & 24.84  & & 16.30  & & 33.09  & & 7.01 \\
OLS& &  & (15.42)  & & (27.29)  & & (18.13)  & & (42.56)  & & (11.58) \\
(standard $t$-test) & & In-kind treatment & 20.26  & & 4.48  & & 30.42  & & 65.36  & & 11.10 \\
& &  & (15.67)  & & (18.42)  & & (22.83)  & & (53.28)  & & (15.31) \\
&& Cash$=$in-kind ($p$-val) & 0.975  & & 0.493  & & 0.600  & & 0.610  & & 0.817 \\
\\
\multirow{5}{*}{} & & Cash treatment & 19.64  & & 24.84  & & 16.30  & & 33.09  & & 7.01 \\
Difference-in-means& &  & (14.24)  & & (26.05)  & & (15.21)  & & (39.27)  & & (11.15) \\
(adjusted $t$-test) & & In-kind treatment & 20.26  & & 4.48  & & 30.42  & & 65.36  & & 11.10 \\
& &  & (15.24)  & & (17.79)  & & (21.97)  & & (48.27)  & & (14.99) \\
&& Cash$=$in-kind ($p$-val) & 0.974  & & 0.468  & & 0.567  & & 0.576  & & 0.815 \\
\toprule
\end{tabular}
\end{adjustbox}
\label{table:application-wave6}
\begin{tablenotes}
\item Note: The results in this table are based on the data from the sixth wave of data collection. For each treatment and each subsample, the number in the first row is the point estimate and that in the second row is the standard error. For testing the equality of the average potential outcomes under the two values of treatment, we report the $p$-values as in \cite{McKenzie2014}.
\end{tablenotes}
\end{table}

\section{Recommendations for Empirical Practice} \label{sec:rec}
We conclude with some recommendations for empirical practice based on our theoretical results as well as the simulation study above. For inference about the linear contrast of expected outcomes given by $\Delta_{\nu}$ in a matched tuples design, we recommend the test $\phi_n^{\nu}$ defined in Section \ref{sec:main_tuple}: our simulations results show that this test does a good job of controlling size in large samples (approximately 80 blocks). We have shown that tests based on the heteroskedasticity-robust variance estimator from a linear regression of outcomes on treatment and block fixed effects may be \emph{invalid}, in the sense of having rejection probability strictly greater than the nominal level under the null hypothesis. Tests based on the heteroskedasticity-robust variance or block-cluster variance estimators from a linear regression of outcomes on treatment are valid but potentially conservative, which would result in a loss of power relative to our proposed test. 

We also find that matched tuples designs have favorable efficiency properties relative to other popular designs (with a specific illustration in the setting of $2^K$ factorial designs). However, this comes with the caveat that when dealing with a large number of treatments (in our simulations, this translated to having fewer than 80 blocks) and/or large number of covariates, practitioners may want to consider the replicated matched tuples design introduced in Section \ref{sec:replicate}, as our simulations suggest that this design may have more robust size control, which translates to better power in such cases.

\clearpage
\appendix
\small

\section{Additional Details}

\subsection{Details for Section \ref{sec:application}}\label{sec:app_details}
\begin{proposition}\label{prop:application}
Consider the setting with three treatment statuses $\{1, 2, 3\}$, where $1$ corresponds to being untreated and $2$ and $3$ correspond to two treatments. In a matched quadruplets design where each quadruplet has two untreated units and one unit for each treatment, the test introduced in Section \ref{sec:main_tuple} with $\mathcal{D}^\prime = \{1,2,3,4\}$ and 
\begin{equation*}
    \nu = 
\begin{pmatrix}
-1/2 & -1/2 & 1 & 0  \\
-1/2 & -1/2 & 0 & 1  \\
0 & 0 & -1 & 1
\end{pmatrix}
\end{equation*}
is valid for testing \eqref{eq:h0-app}--\eqref{eq:h0-app-23} at level $\alpha \in (0, 1)$.
\end{proposition}

\begin{proof}[\sc Proof of Proposition \ref{prop:application}]
Consider a design of matched-quadruplets with two treatments $d=2,3$ and two controls $d=1$, i.e a quadruplet consisting of $(1,1,2,3)$. The difference-in-mean estimator for the effect of the first treatment $d=2$ is
\begin{equation*}
    \hat{\Delta}_2 = \frac{1}{n}\sum_{i=1}^{4n} I\{D_i=2\} Y_i - \frac{1}{2 n}\sum_{i=1}^{4n} I\{D_i=1\} Y_i
\end{equation*}
Note that
\[ \sqrt n \left(\hat{\Delta}_2 - \Delta_2(Q)\right) = A_{n,2} + C_{n,2} - (A_{n,1} + C_{n,1})~, \]
where
\begin{align*}
A_{n, 2} & = \frac{1}{\sqrt n} \sum_{1 \leq i \leq 4 n} I \{D_i = 2\} (Y_i(2) - E[Y_i(2) | X^{(n)}, D^{(n)}]) \\
C_{n, 2} & = \frac{1}{\sqrt n} \sum_{1 \leq i \leq 4 n} I \{D_i = 2\} (E[Y_i(2) | X^{(n)}, D^{n}] - E[Y_i(2)])~.
\end{align*}
and
\begin{align*}
A_{n, 1} & = \frac{1}{2\sqrt n} \sum_{1 \leq i \leq 4 n} I \{D_i = 1\} (Y_i(1) - E[Y_i(1) | X^{(n)}, D^{(n)}]) \\
C_{n, 1} & = \frac{1}{2\sqrt n} \sum_{1 \leq i \leq 4 n} I \{D_i = 1\} (E[Y_i(1) | X^{(n)}, D^{n}] - E[Y_i(1)])~.
\end{align*}
Let $I_{j}$ denote the set of indices for the two untreated units in the $j$-th tuple. Note
\begin{align*}
& \var[A_{n, 1} | X^{(n)}, D^{(n)}] \\
& = \frac{1}{2\cdot 2n} \sum_{1 \leq i \leq 4 n} I \{D_i = 1\} \var[Y_i(1) | X_i] \\
& = \frac{1}{2\cdot 4n} \sum_{1 \leq i \leq 4 n} \var[Y_i(1) | X_i] - \frac{1}{8 n} \sum_{1 \leq j \leq n} \frac{1}{2} \sum_{i_j \in I_j} \sum_{k \in \lambda_j: k \not\in I_j} (\var[Y_k(1) | X_k] - \var[Y_{i_j} | X_{i_j}])
\end{align*}
It follows from similar arguments as in the proof of Theorem \ref{thm:main_delta} that the second term goes to zero. Therefore,
\begin{equation*}
    \var[A_{n, 1} | X^{(n)}, D^{(n)}] \stackrel{P}{\to} \frac{1}{2} E[\var[Y_i(1) | X_i]] ~.
\end{equation*}
It therefore follows from Lemma S.1.2 of \cite{bai2021inference} that
\begin{equation*}
    \gamma\left(\left(\left(A_{n, 2}, A_{n,1}\right)' | X^{(n)}, D^{(n)}\right), N\left(0,  \begin{bmatrix} E[\var[Y_i(2) | X_i]] & 0\\ 0 & \frac{1}{2} E[\var[Y_i(1) | X_i]] \end{bmatrix}\right)
    \right) \stackrel{P}{\to} 0~,
\end{equation*}
where $\gamma$ is any metric that metrizes weak convergence. 

Next, note
\begin{align*}
    E[C_{n, 2} | X^{(n)}] &= \frac{1}{\sqrt n} \sum_{1 \leq i \leq 4 n} \frac{1}{4} (E[Y_i(2) | X^{(n)}] - E[Y_i(2)]) = \frac{1}{4 \sqrt n} \sum_{1 \leq i \leq 4 n} (\Gamma_2(X_i) - \Gamma_2) \\
    E[C_{n, 1} | X^{(n)}] &= \frac{1}{2\sqrt n} \sum_{1 \leq i \leq 4 n} \frac{1}{2} (E[Y_i(1) | X^{(n)}] - E[Y_i(1)]) = \frac{1}{4 \sqrt n} \sum_{1 \leq i \leq 4 n} (\Gamma_1(X_i) - \Gamma_1) ~.
\end{align*}
Therefore,
\[ (C_{n, 2}, C_{n,1})^\prime \stackrel{d}{\to} N\left(0, \frac{1}{4} \begin{bmatrix} \var(\Gamma_2(X_i)) & \cov(\Gamma_2(X_i), \Gamma_{1}(X_i))\\ \cov(\Gamma_2(X_i), \Gamma_{1}(X_i)) & \var(\Gamma_1(X_i))] \end{bmatrix} \right)~. \]
It then follows from Lemma S.1.2 of \cite{bai2021inference} that
\begin{equation*}
    \sqrt n \left(\hat{\Delta}_2 - \Delta_2(Q)\right) \stackrel{d}{\to} N \left(0, \mathbb{V}_2\right)~,
\end{equation*}
where
\[ \mathbb{V}_2 = E[\var[Y_i(2) | X_i]] + \frac{1}{2} E[\var[Y_i(1) | X_i]] + \frac{1}{4} \left(\var(\Gamma_2(X_i)) + \var(\Gamma_1(X_i)) - 2 \cov(\Gamma_2(X_i), \Gamma_{1}(X_i)) \right)~. \]

Now, suppose we pretend the two untreated units are assigned to two distinct treatment levels and denote the two untreated levels and two treated levels by $d \in \{1, 2, 3, 4\}$, where $d=1,2$ actually corresponds to the untreated units. Our estimand can then be defined as
\begin{equation*}
    \Tilde{\Delta}_2(Q) = \Gamma_3(Q) - \frac{1}{2}\left(\Gamma_1(Q) +\Gamma_2(Q)\right)
\end{equation*}
Applying the existing results in Theorem \ref{thm:main_delta} with $\nu = (-1/2, -1/2,1,0)$. It follows that
\begin{equation*}
    \sqrt n \left(\hat{\Delta}_2 - \Tilde{\Delta}_2(Q)\right) \stackrel{d}{\to} N \left(0, \Tilde{\mathbb{V}}_2\right)~,
\end{equation*}
where
\begin{align*}
    \Tilde{\mathbb{V}}_2 &= E[\var[Y_i(3) | X_i]] + \frac{1}{4}\left( E[\var[Y_i(1) | X_i]] + E[\var[Y_i(2) | X_i]]\right) \\
    &+ \frac{1}{4} \bigg(\var(\Gamma_3(X_i)) + \frac{1}{4}\var(\Gamma_1(X_i)) + \frac{1}{4}\var(\Gamma_2(X_i)) + \frac{1}{2}\cov(\Gamma_2(X_i), \Gamma_{1}(X_i)) \\
    & -\cov(\Gamma_3(X_i), \Gamma_{1}(X_i)) -\cov(\Gamma_3(X_i), \Gamma_{2}(X_i))\bigg) \\
    &= \mathbb{V}_2~,
\end{align*}
where the last equality follows by setting $d=1,2,3$ to $d=1,1,2$. The same argument holds for $v=(-1/2,-1/2,0,1)$. As for $\nu=(0,0,-1,1)$, the estimation and inference of the third and fourth arms is not affected by treatment status in the first two arms. Therefore, pretending two controls are two different treatment levels yields the true asymptotic variance, meaning that the inference is still valid.
\end{proof}

\section{Proofs of Main Results}

\subsection{Proof of Theorem \ref{thm:main_delta}}
We derive the limiting distribution of $\sqrt n(\hat \Gamma_n(d) - \Gamma_d(Q): d \in \mathcal D)$, from which the conclusion of the theorem follows by an application of the continuous mapping theorem. Note that
\[ \sqrt n(\hat \Gamma_n(d) - \Gamma_d(Q): d \in \mathcal D)' = A_n + C_n~, \]
where $A_n = (A_{n, d}: d \in \mathcal D)'$, $C_n = (C_{n, d}: d \in \mathcal D)'$, and
\begin{align*}
A_{n, d} & = \frac{1}{\sqrt n} \sum_{1 \leq i \leq |\mathcal D|n} I \{D_i = d\} (Y_i(d) - E[Y_i(d) | X^{(n)}, D^{(n)}]) \\
C_{n, d} & = \frac{1}{\sqrt n} \sum_{1 \leq i \leq |\mathcal D|n} I \{D_i = d\} (E[Y_i(d) | X^{(n)}, D^{n}] - E[Y_i(d)])~.
\end{align*}
Note that conditional on $X^{(n)}, D^{(n)}$, $C_{n, d}$'s are constants, and $A_{n, d}$'s are independent. By Assumption \ref{as:D}, for $d \in \mathcal D$, $E[Y_i(d) | X^{(n)}, D^{(n)}] = E[Y_i(d) | X_i]$. Fix $d \in \mathcal D$. Let $i_j \in \lambda_j$ be such that $D_{i_j} = d$. Note
\begin{align*}
& \var[A_{n, d} | X^{(n)}, D^{(n)}] \\
& = \frac{1}{n} \sum_{1 \leq i \leq |\mathcal D| n} I \{D_i = d\} \var[Y_i(d) | X_i] \\
& = \frac{1}{|\mathcal D| n} \sum_{1 \leq i \leq |\mathcal D| n} \var[Y_i(d) | X_i] - \frac{1}{|\mathcal D| n} \sum_{1 \leq j \leq n} \sum_{k \in \lambda_j: k \neq i_j} (\var[Y_k(d) | X_k] - \var[Y_{i_j}(d) | X_{i_j}])
\end{align*}
where the first equality follows from Assumption \ref{as:D}. By Assumption \ref{as:Q}(b) and the weak law of large numbers,
\[ \frac{1}{|\mathcal D| n} \sum_{1 \leq i \leq |\mathcal D| n} \var[Y_i(d) | X_i] \stackrel{P}{\to} E[\var[Y_i(d) | X_i]]~. \]
By Assumptions \ref{as:Q}(c) and \ref{as:close}, we have
\begin{align*}
& \Big | \frac{1}{|\mathcal D| n} \sum_{1 \leq j \leq n} \sum_{k \in \lambda_j: k \neq i_j} (\var[Y_k(d) | X_k] - \var[Y_{i_j}(d) | X_{i_j}]) \Big | \\
& \leq \frac{1}{|\mathcal D| n} \sum_{1 \leq j \leq n} \sum_{k \in \lambda_j: k \neq i_j} |\var[Y_k(d) | X_k] - \var[Y_{i_j}(d) | X_{i_j}]| \\
& \lesssim \frac{1}{n} \sum_{1 \leq j \leq n} \sum_{k \in \lambda_j: k \neq i_j} \|X_k - X_{i_j}\| \\
& \leq \frac{|\mathcal D| - 1}{n} \sum_{1 \leq j \leq n} \max_{i, k \in \lambda_j} \|X_i - X_k\| \stackrel{P}{\to} 0~.
\end{align*}
Therefore, $\var[A_{n, d} | X^{(n)}, D^{(n)}] \stackrel{P}{\to} E[\var[Y_i(d) | X_i]]$. We can then verify Lindeberg's condition as in the proof of Lemma S.1.4 of \cite{bai2021inference}. It follows that
\[ \gamma(((A_{n, d}: d \in \mathcal D)' | X^{(n)}, D^{(n)})), N(0, \mathbb V_1)) \stackrel{P}{\to} 0~, \]
where $\mathbb V_1 = \diag(E[\var[Y_i(d) | X_i]]: d \in \mathcal D)$ and $\gamma$ is any metric that metrizes weak convergence.

Next,
\[ E[C_{n, d} | X^{(n)}] = \frac{1}{|\mathcal D| \sqrt n} \sum_{1 \leq i \leq |\mathcal D| n} (\Gamma_d(X_i) - \Gamma_d) \]
and
\begin{align*}
\var[C_{n, d} | X^{(n)}] & = \frac{1}{n} \sum_{1 \leq j \leq n} \sum_{i \in \lambda_j} \frac{1}{|\mathcal D|} \Big ( \Gamma_d(X_i) - \frac{1}{|\mathcal D|} \sum_{k \in \lambda_j} \Gamma_d(X_k) \Big )^2 \\
& \lesssim \frac{1}{n} \sum_{1 \leq j \leq n} \max_{i, k \in \lambda_j} \|X_i - X_k\|^2 \stackrel{P}{\to} 0
\end{align*}
by Assumption \ref{as:Q}(c) and \ref{as:close}. Therefore, by repeating the argument which establishes (S.24) in the proof of Lemma S.1.4 of \cite{bai2021inference}, it follows that
\[ C_{n, d} = \frac{1}{|\mathcal D| \sqrt n} \sum_{1 \leq i \leq |\mathcal D|n} (\Gamma_d(X_i) - \Gamma_d) + o_P(1)~. \]
Therefore,
\[ (C_{n, d}: d \in \mathcal D) \stackrel{d}{\to} N(0, \mathbb V_2)~, \]
where $(\mathbb V_2)_{d, d'} = \frac{1}{|\mathcal D|} \cov(\Gamma_d(X_i), \Gamma_{d'}(X_i))$. It then follows from Lemma S.1.2 of \cite{bai2021inference} that
\[ \sqrt n(\hat \Gamma_n(d) - \Gamma_d: d \in \mathcal D)' \stackrel{d}{\to} N(0, \mathbb V_1 + \mathbb V_2)~. \]
The conclusion now follows. \qed

\subsection{Proof of Theorem \ref{thm:V_const}}
The conclusion follows from Lemmas \ref{lem:mean}--\ref{lem:rho-dd} together with the continuous mapping theorem. \qed

\subsection{Proof of Theorem \ref{thm:sfe}}
Define
\[ C_i = (I \{D_i = 2\}, \ldots, I \{D_i = |\mathcal D|\})'~. \]
To begin, note it follows from the Frisch-Waugh-Lovell theorem and Assumption \ref{as:D} that
\[ \begin{pmatrix}
\hat \beta_n(2) \\
\vdots \\
\hat \beta_n(|\mathcal D|)
\end{pmatrix} = \left ( \sum_{1 \leq i \leq |\mathcal D| n} \tilde C_i \tilde C_i' \right )^{-1} \sum_{1 \leq i \leq |\mathcal D| n} \tilde C_i Y_i~, \]
where
\[ \tilde C_i = \left ( I \{D_i = 2\} - \frac{1}{|\mathcal D|}, \ldots, I \{D_i = |\mathcal D|\} - \frac{1}{|\mathcal D|} \right )'~. \]
Next, note for
\[ \sum_{1 \leq i \leq |\mathcal D| n} \tilde C_i \tilde C_i'~, \]
the diagonal entries are $\frac{|\mathcal D| - 1}{|\mathcal D|} n$ and the off-diagonal entries are $- \frac{1}{|\mathcal D|} n$. It follows from element calculation that the diagonal entries of $\left ( \sum_{1 \leq i \leq |\mathcal D| n} \tilde C_i \tilde C_i' \right )^{-1}$ are $\frac{2}{n}$ and the off-diagonal entries are $\frac{1}{n}$. Furthermore,
\[ \sum_{1 \leq i \leq |\mathcal D| n} \tilde C_i Y_i = \begin{pmatrix}
n \hat \Gamma_n(2) - \frac{1}{|\mathcal D|} \sum_{1 \leq i \leq |\mathcal D| n} Y_i \\
\vdots \\
n \hat \Gamma_n(|\mathcal D|) - \frac{1}{|\mathcal D|} \sum_{1 \leq i \leq |\mathcal D| n} Y_i
\end{pmatrix}~. \]
Therefore, for $d \in \mathcal D \backslash \{1\}$,
\begin{align*}
\hat \beta_n(d) & = \frac{2}{n} \Big ( n \hat \Gamma_n(d) - \frac{1}{|\mathcal D|} \sum_{1 \leq i \leq |\mathcal D| n} Y_i \Big ) + \frac{1}{|\mathcal D|} \sum_{d' \in \mathcal D \backslash \{1, d\}} \hat \Gamma_n(d') - \frac{|\mathcal D| - 2}{|\mathcal D| n} \sum_{1 \leq i \leq |\mathcal D| n} Y_i \\
& = \hat \Gamma_n(d) - \hat \Gamma_n(1)~.
\end{align*}
The first conclusion of the theorem then follows.

Next, note by the properties of the OLS estimator that
\begin{align*}
\hat \delta_{j, n} & = \left ( \sum_{1 \leq i \leq |\mathcal D| n} I \{i \in \lambda_j\} \right )^{-1} \sum_{1 \leq i \leq |\mathcal D| n} I \{i \in \lambda_j\} \left ( Y_i - \sum_{d \in \mathcal D \backslash \{1\}} \hat \beta_n(d) I \{D_i = d\} \right ) \\
& = \frac{1}{|\mathcal D|} \sum_{i \in \lambda_j} Y_i - \frac{1}{|\mathcal D|} \sum_{d \in \mathcal D \backslash \{1\}} \hat \beta_n(d)~.
\end{align*}
Therefore,
\[ \hat \epsilon_i = \begin{cases}
Y_i - \sum_{1 \leq j \leq n} I\{i\in \lambda_j \} \frac{1}{|\mathcal D|} \sum_{k \in \lambda_j} Y_k + \frac{1}{|\mathcal D|} \sum_{d' \in \mathcal D \backslash \{1\}} \hat \beta_n(d')~, \text{ if } D_i = 1 \\
Y_i - \hat \beta_n(d) - \sum_{1 \leq j \leq n} I\{i\in \lambda_j \} \frac{1}{|\mathcal D|} \sum_{k \in \lambda_j} Y_k + \frac{1}{|\mathcal D|} \sum_{d' \in \mathcal D \backslash \{1\}} \hat \beta_n(d')~, \text{ if } D_i = d \neq 1~.
\end{cases} \]
Note it follows from Lemma \ref{lem:fwl} that the heteroskedasticity-robust variance estimator of $(\hat \beta_n(2), \ldots, \hat \beta_n(|\mathcal D|))'$ equals
\[ \left ( \sum_{1 \leq i \leq |\mathcal D| n} \tilde C_i \tilde C_i' \right )^{-1} \left ( \sum_{1 \leq i \leq |\mathcal D| n} \hat \epsilon_i^2 \tilde C_i \tilde C_i' \right ) \left ( \sum_{1 \leq i \leq |\mathcal D| n} \tilde C_i \tilde C_i' \right )^{-1}~. \]
For $d \in \mathcal D \backslash \{1\}$, the corresponding $(d - 1)$-th diagonal term of $\mathbb A = \sum_{1 \leq i \leq |\mathcal D| n} \hat \epsilon_i^2 \tilde C_i \tilde C_i'$ equals
\[ \mathbb A_d = \sum_{1 \leq i \leq |\mathcal D| n} I \{D_i = 1\} \frac{1}{|\mathcal D|^2} \hat \epsilon_i^2 + \sum_{1 \leq i \leq |\mathcal D| n} I \{D_i = d\} \frac{(|\mathcal D| - 1)^2}{|\mathcal D|^2} \hat \epsilon_i^2 + \sum_{\tilde d \in \mathcal D \backslash \{1, d\}} \sum_{1 \leq i \leq |\mathcal D| n} I \{D_i = \tilde d\} \frac{1}{|\mathcal D|^2} \hat \epsilon_i^2~. \]
For $\tilde d \neq \check d \in \mathcal D \backslash \{1\}$, the correponding $(\tilde d - 1, \check d - 1)$-th term of $\sum_{1 \leq i \leq |\mathcal D| n} \hat \epsilon_i^2 \tilde C_i \tilde C_i'$ equals
\begin{multline*}
\mathbb A_{\tilde d, \check d} = \sum_{1 \leq i \leq |\mathcal D| n} I \{D_i = 1\} \frac{1}{|\mathcal D|^2} \hat \epsilon_i^2 + \sum_{1 \leq i \leq |\mathcal D| n} I \{D_i = \tilde d\} \frac{-(|\mathcal D| - 1)}{|\mathcal D|^2} \hat \epsilon_i^2 \\
+ \sum_{1 \leq i \leq |\mathcal D| n} I \{D_i = \check d\} \frac{-(|\mathcal D| - 1)}{|\mathcal D|^2} \hat \epsilon_i^2 + \sum_{d' \in \mathcal D \backslash \{1, \tilde d, \check d\}} \sum_{1 \leq i \leq |\mathcal D| n} I \{D_i = d'\} \frac{1}{|\mathcal D|^2} \hat \epsilon_i^2~.
\end{multline*}
Therefore,
\begin{align*}
\hat {\mathbb V}_n^{\rm sfe}(d, 1) & = \frac{4}{n^2} \mathbb A_d + \frac{1}{n^2} \sum_{\tilde d \in \mathcal D \backslash \{1, d\}} \mathbb A_{\tilde d} + \frac{4}{n^2} \sum_{\tilde d \in \mathcal D \backslash \{1, d\}} \mathbb A_{d, \tilde d} + \frac{2}{n^2} \sum_{\tilde d < \check d \in \mathcal D \backslash \{1, d\}} \mathbb A_{\tilde d, \check d} \\
& = \frac{4 + |\mathcal D| - 2  + 4 (|\mathcal D| - 2) + 2 (|\mathcal D| - 2) (|\mathcal D| - 3) / 2}{n^2} \sum_{1 \leq i \leq |\mathcal D| n} I \{D_i = 1\} \frac{1}{|\mathcal D|^2} \hat \epsilon_i^2 \\
& \hspace{3em} + \frac{1}{n^2} \sum_{1 \leq i \leq |\mathcal D| n} I \{D_i = d\} \frac{4 (|\mathcal D| - 1)^2 + |\mathcal D| - 2 - 4 (|\mathcal D| - 1) (|\mathcal D| - 2) + 2 (|\mathcal D| - 2) (|\mathcal D| - 3) / 2}{|\mathcal D|^2} \hat \epsilon_i^2 \\
& \hspace{3em} + \frac{1}{n^2} \sum_{\tilde d \in \mathcal D \backslash \{1, d\}} \sum_{1 \leq i \leq |\mathcal D| n} I \{D_i = \tilde d\} \\
& \hspace{3em} \times \frac{4 + (|\mathcal D| - 1)^2 + |\mathcal D| - 3 - 4(|\mathcal D| - 1) + 4(|\mathcal D| - 3) - 2 (|\mathcal D| - 1) (|\mathcal D| - 3) + 2 (|\mathcal D| - 3)(|\mathcal D| - 4) / 2}{|\mathcal D|^2} \hat \epsilon_i^2 \\
& = \frac{1}{n^2} \sum_{1 \leq i \leq |\mathcal D| n} I \{D_i = 1\} \hat \epsilon_i^2 + \frac{1}{n^2} \sum_{1 \leq i \leq |\mathcal D| n} I \{D_i = d\} \hat \epsilon_i^2 \\
& = \frac{1}{n^2} \sum_{1 \leq i \leq |\mathcal D| n} I \{D_i = 1\} \left ( Y_i - \sum_{1 \leq j \leq n} I\{i\in \lambda_j \} \frac{1}{|\mathcal D|} \sum_{k \in \lambda_j} Y_k + \frac{1}{|\mathcal D|} \sum_{d' \in \mathcal D \backslash \{1\}} \hat \beta_n(d') \right )^2 \\
& \hspace{3em} +\frac{1}{n^2} \sum_{1 \leq i \leq |\mathcal D| n} I \{D_i = d\} \left ( Y_i - \hat \beta_n(d) - \sum_{1 \leq j \leq n} I\{i\in \lambda_j \} \frac{1}{|\mathcal D|} \sum_{k \in \lambda_j} Y_k + \frac{1}{|\mathcal D|} \sum_{d' \in \mathcal D \backslash \{1\}} \hat \beta_n(d') \right )^2 \\
& = \frac{1}{n^2} \sum_{1 \leq i \leq |\mathcal D| n} I \{D_i = 1\} \left ( Y_i - \hat \Gamma_n(1) - \sum_{1 \leq j \leq n} I\{i\in \lambda_j \} \frac{1}{|\mathcal D|} \sum_{k \in \lambda_j} Y_k + \frac{1}{|\mathcal D|} \sum_{d' \in \mathcal D} \hat \Gamma_n(d') \right )^2 \\
& \hspace{3em} + \frac{1}{n^2} \sum_{1 \leq i \leq |\mathcal D| n} I \{D_i = d\} \left ( Y_i - \hat \Gamma_n(d) - \sum_{1 \leq j \leq n} I\{i \in \lambda_j\} \frac{1}{|\mathcal D|} \sum_{k \in \lambda_j} Y_k + \frac{1}{|\mathcal D|} \sum_{d' \in \mathcal D} \hat \Gamma_n(d') \right )^2 \\
& = \frac{1}{n^2} \sum_{1 \leq i \leq |\mathcal D| n} I \{D_i = 1\} \left ( Y_i - \hat \Gamma_n(1) - \sum_{1 \leq j \leq n} I\{i \in \lambda_j\} \frac{1}{|\mathcal D|} \sum_{k \in \lambda_j} Y_k + \frac{1}{|\mathcal D|} \sum_{d' \in \mathcal D} \hat \Gamma_n(d') \right )^2 \\
& \hspace{3em} + \frac{1}{n^2} \sum_{1 \leq i \leq |\mathcal D| n} I \{D_i = d\} \left ( Y_i - \hat \Gamma_n(d) - \sum_{1 \leq j \leq n} I\{i \in \lambda_j\} \frac{1}{|\mathcal D|} \sum_{k \in \lambda_j} Y_k + \frac{1}{|\mathcal D|} \sum_{d' \in \mathcal D} \hat \Gamma_n(d') \right )^2 \\
& = \frac{1}{n^2} \sum_{1 \leq j \leq n} \left ( \sum_{i \in \lambda_j} \left ( I \{D_i = 1\} - \frac{1}{|\mathcal D|} \right ) Y_i \right )^2 - \frac{1}{n} \left ( \hat \Gamma_n(1) - \frac{1}{|\mathcal D|} \sum_{d' \in \mathcal D} \hat \Gamma_n(d') \right )^2 \\
& \hspace{3em} + \frac{1}{n^2} \sum_{1 \leq j \leq n} \left ( \sum_{i \in \lambda_j} \left ( I \{D_i = d\} - \frac{1}{|\mathcal D|} \right ) Y_i \right )^2 - \frac{1}{n} \left ( \hat \Gamma_n(d) - \frac{1}{|\mathcal D|} \sum_{d' \in \mathcal D} \hat \Gamma_n(d') \right )^2~,
\end{align*}
where in the last equality we used the fact that for $d \in \mathcal D$,
\[ \sum_{1 \leq i \leq |\mathcal D| n} I \{D_i = d\} \left ( Y_i - \sum_{1 \leq j \leq n} I\{i \in \lambda_j\} \frac{1}{|\mathcal D|} \sum_{k \in \lambda_j} Y_k \right ) \left ( \hat \Gamma_n(d) - \frac{1}{|\mathcal D|} \sum_{d' \in \mathcal D} \hat \Gamma_n(d') \right ) = n \left ( \hat \Gamma_n(d) - \frac{1}{|\mathcal D|} \sum_{d' \in \mathcal D} \hat \Gamma_n(d') \right )^2~. \]

It follows from Assumptions \ref{as:Q} and \ref{as:close} as well as Lemmas \ref{lem:mean}--\ref{lem:rho-dd} that as $n \to \infty$,
\begin{align*}
    & \hat \Gamma_n(d) \stackrel{P}{\to} E[Y_i(d)] \text{ for all } d \in \mathcal D \\
    & \frac{1}{n} \sum_{1 \leq j \leq n}  \sum_{i \in \lambda_j} I \{D_i = d\} Y_i^2 \stackrel{P}{\to} E[Y_i^2(d)] \\
    & \frac{1}{n} \sum_{1 \leq j \leq n} \left ( \sum_{i \in \lambda_j} I \{D_i = d\} Y_i \right ) \left ( \sum_{i \in \lambda_j} I \{D_i = d'\} Y_i \right ) \stackrel{P}{\to} E[\Gamma_d(X_i) \Gamma_{d'}(X_i)] \text{ for all } d \neq d' \in \mathcal D~.
\end{align*}
Therefore,
\begin{align*}
    n \hat {\mathbb V}_n^{\rm sfe}(d, 1) & \stackrel{P}{\to} \var \left [ \Gamma_1(X_i) - \frac{1}{|\mathcal D|} \sum_{d' \in \mathcal D} \Gamma_{d'}(X_i) \right ] + \left ( 1 - \frac{1}{|\mathcal D|} \right )^2 E[\var[Y_i(1) | X_i]] + \frac{1}{|\mathcal D|^2} \sum_{d' \in \mathcal D \backslash \{1\}} E[\var[Y_i(d') | X_i]] \\
    & \hspace{3em} + \var \left [ \Gamma_d(X_i) - \frac{1}{|\mathcal D|} \sum_{d' \in \mathcal D} \Gamma_{d'}(X_i) \right ] + \left ( 1 - \frac{1}{|\mathcal D|} \right )^2 E[\var[Y_i(d) | X_i]] + \frac{1}{|\mathcal D|^2} \sum_{d' \in \mathcal D \backslash \{d\}} E[\var[Y_i(d') | X_i]]~.
\end{align*}
Finally, note by Theorem \ref{thm:main_delta} that the actual limiting variance for $\hat \Gamma_n(d) - \hat \Gamma_n(1)$ is
\[ E\left[\var[Y_i(d) | X_i]] + E[\var[Y_i(1) | X_i] \right] + \frac{1}{|\mathcal D|} E\left[ \left( (\Gamma_d(X_i) - \Gamma_d) - (\Gamma_1(X_i) - \Gamma_1) \right)^2 \right]~. \]
Consider the special case where $E[\var[Y_i(d') | X_i]]$ are identical across $d' \in \mathcal D$ and
\[ \Gamma_1(X_i) = \Gamma_d(X_i) = \frac{1}{|\mathcal D|} \sum_{d' \in \mathcal D} \Gamma_{d'}(X_i) \text{ with probability one}~. \]
Then, the probability limit of $n \hat {\mathbb V}_n^{\rm sfe}(d, 1)$ is clearly strictly smaller than the actual limiting variance for $\hat \Gamma_n(d) - \hat \Gamma_n(1)$.

For variance estimator HC $1$, consider the special case where $\Gamma_d(X_i)$ are identical across $d \in \mathcal D$, $E[\var[Y_i(d) | X_i]] > 0$, $E[\var[Y_i(1) | X_i]] > 0$, and $E[\var[Y_i(d') | X_i]]$ is zero for all $d' \in \mathcal D \backslash \{1, d\}$. Then,
\begin{align*}
   n \hat{\mathbb V}_n^{\rm sfe}(d, 1) \times \frac{|\mathcal D| n}{|\mathcal D| n - (|\mathcal D| - 1 + n)} & \stackrel{P}{\to} \frac{|\mathcal D|}{|\mathcal D| - 1} \left ( \left ( 1 - \frac{1}{|\mathcal D|} \right )^2 + \frac{1}{|\mathcal D|^2} \right ) (E[\var[Y_i(d) | X_i]] + E[\var[Y_i(1) | X_i]]) \\
   & = \frac{|\mathcal D|^2 - 2 |\mathcal D| + 2}{|\mathcal D|^2 - |\mathcal D|} (E[\var[Y_i(d) | X_i]] + E[\var[Y_i(1) | X_i]])~.
\end{align*}
Note that
\[ \frac{|\mathcal D|^2 - 2 |\mathcal D| + 2}{|\mathcal D|^2 - |\mathcal D|} < 1 \]
if and only if $|\mathcal D| > 2$. By a continuity argument, the result then follows for the case where $E[\var[Y_i(d') | X_i]]$ is sufficiently close to zero for all $d' \in \mathcal D \backslash \{1, d\}$. \qed

\subsection{Proof of Theorem \ref{thm:bcve}}
First, note that
\begin{equation*}
    \left(\frac{1}{n }\sum_{1\leq j \leq n} \sum_{i \in \lambda_j} C_i C_i'\right)^{-1}
= \begin{pmatrix}
|\mathcal{D}| & 1 & 1 & \dots & 1 \\
1 & 1 & 0 &\dots & 0 \\
1 & 0 & 1 &\dots & 0 \\
\vdots & \vdots & \vdots &\ddots &\vdots \\
1 & 0 & 0 & \dots & 1 \\
\end{pmatrix}^{-1} 
= \begin{pmatrix}
1 & -1 & -1 & \dots & -1 \\
-1 & 2 & 1 &\dots & 1 \\
-1 & 1 & 2 &\dots & 1 \\
\vdots & \vdots & \vdots &\ddots &\vdots \\
-1 & 1 & 1 & \dots & 2 \\
\end{pmatrix}~,
\end{equation*}
and note that
\begin{equation*}
\sum_{i\in\lambda_j} \hat \epsilon_i C_i = 
    \begin{pmatrix}
\sum_{i \in \lambda_j}\sum_{d \in \mathcal{D}\backslash \{1\}} (Y_i - \hat \gamma_n(d)) I\{D_i = d\}+ Y_i I\{D_i=1\}- \hat\gamma_n(1)  &  \\
\sum_{i \in \lambda_j} (Y_i - \hat \gamma_n(2)- \hat\gamma_n(1)) I\{D_i = 2\} &  \\
\sum_{i \in \lambda_j} (Y_i - \hat \gamma_n(3)- \hat\gamma_n(1)) I\{D_i = 3\} &  \\
\vdots &  \\
\sum_{i \in \lambda_j} (Y_i - \hat \gamma_n(|\mathcal{D}|)- \hat\gamma_n(1)) I\{D_i = |\mathcal{D}|\} &  \\
\end{pmatrix}~.
\end{equation*}
Combining these expressions, it follows that the $d$-th diagonal element of $n\cdot\hat{\mathbb{V}}_n^{\rm bcve}$ is equal to
\begin{align*}
n\cdot\hat{\mathbb{V}}^{\rm bcve}_n(d) &= 
    \frac{1}{n}\sum_{1\leq j \leq n} \left ( \sum_{i \in \lambda_j} (Y_i - \hat \gamma_n(d) - \hat\gamma_n(1)) I\{D_i = d\} -  \sum_{i \in \lambda_j} (Y_i - \hat \gamma_n(1)) I\{D_i = 1\} \right )^2 \\
    &= \frac{1}{n} \sum_{1 \leq j \leq n} \left ( \sum_{i \in \lambda_j} Y_i I \{D_i = d\} - \sum_{i \in \lambda_j} Y_i I \{D_i = 1\} \right )^2 - (\hat \Gamma_n(d) - \hat \Gamma_n(1))^2~.
\end{align*}
Where the second equality exploits the fact that $\hat \gamma_n(d) = \hat \Gamma_n(d) - \hat \Gamma_n(1)$ for $d \in \mathcal{D} \backslash \{1\}$ and $\hat \gamma_n(1) =\hat\Gamma_n(1)$. It thus follows from Lemmas \ref{lem:mean}--\ref{lem:rho-dd'} and the continuous mapping theorem that
\[n\cdot\hat{\mathbb{V}}^{\rm bcve}_n(d) \xrightarrow{p} E[\var[Y_i(d) | X_i]] + E[\var[Y_i(1) | X_i]] + E\left[ \left( (\Gamma_d(X_i) - \Gamma_d) - (\Gamma_1(X_i) - \Gamma_1) \right)^2 \right]~.\]
Next, note that by Theorem \ref{thm:main_delta}, the actual limiting variance of $\hat{\Gamma}_n(d) - \hat{\Gamma}_n(1)$ is given by
\[ E\left[\var[Y_i(d) | X_i]] + E[\var[Y_i(1) | X_i] \right] + \frac{1}{|\mathcal D|} E\left[ \left( (\Gamma_d(X_i) - \Gamma_d) - (\Gamma_1(X_i) - \Gamma_1) \right)^2 \right]~. \]
Therefore, the test defined in \eqref{eq:stat-bcve} is conservative unless 
\[E\left[ \left( (\Gamma_d(X_i) - \Gamma_d) - (\Gamma_1(X_i) - \Gamma_1) \right)^2 \right] = 0~,\]
as desired. \qed

\subsection{Proof of Theorem \ref{thm:replicate-delta}}
The proof is similar to the proof of Theorem \ref{thm:main_delta}, with the difference being that two units are assigned to each treatment status in each block. The necessary modification follows from arguing similarly as in Lemma B.3 of \cite{bai2022optimality} and is omitted. \qed

\subsection{Proof of Theorem \ref{thm:replicate-rho}}
First note
\begin{align*}
    & E[\tilde \rho_n(d, d) | X^{(n)}] \\
    & = \frac{2}{n} \sum_{1 \leq j \leq \lfloor n / 2 \rfloor} \frac{1}{\binom{2 |\mathcal D|}{2}} \sum_{i < l, i, l \in \lambda_j} E[Y_i(d) Y_l(d) | X^{(n)}] \\
    & = \frac{2}{n} \sum_{1 \leq j \leq \lfloor n / 2 \rfloor} \frac{1}{\binom{2 |\mathcal D|}{2}} \sum_{i < l, i, l \in \lambda_j} E[Y_i(d) | X_i] E[Y_l(d) | X_l]~,
\end{align*}
where the first equality follows from the conditional independence assumption in Assumption \ref{as:D} and the fact that in each block, there are $\binom{2 |\mathcal D|}{2}$ ways to choose $2$ units out of $2 |\mathcal D|$ units and assign them to treatment arm $d$, and the second equality follows from the fact that conditional on $X^{(n)}$, $Y^{(n)}(d)$ are i.i.d.\ across units. \eqref{eq:replicate-consistent} then follows by arguing similarly as in the proof of Lemma \ref{lem:rho-dd} below (see also Section 4.7 of \cite{bai2022optimality}). \qed

\subsection{Proof of Theorem \ref{thm:block_factorial}}
First we show that
\[\sqrt{n}(\hat{\Delta}_{\nu,n} - \Delta_{\nu}(Q)) \xrightarrow{d} N(0, \sigma^2_{h,\nu})~,\]
under the stratified factorial design defined by $h(\cdot)$. To show this, we derive the limiting distribution of $\sqrt n(\hat \Gamma_n(d) - \Gamma_d(Q): d \in \mathcal D)$.
To that end, note that
\[\sqrt{n}(\hat{\Gamma}_n(d) - \Gamma_d(Q):d \in \mathcal{D})' = A_n + B_n + C_n + o_P(1)~,\]
where $A_n = (A_{n,d}: d \in \mathcal{D})'$, $B_n = (B_{n,d}: d \in \mathcal{D})'$, $C_n = (C_{n,d}: d \in \mathcal{D})'$, with
\begin{align*}
A_{n,d} &= \sqrt{|\mathcal{D}|}\frac{1}{\sqrt{J_n}}\sum_{1\le i\le J_n}(Y_i(d) - E[Y_i(d)|h(X_i)])I\{D_i = d\}\\
B_{n,d} &=  \sqrt{|\mathcal{D}|}\frac{1}{\sqrt{J_n}}\sum_{1 \le i \le J_n}(I\{D_i = d\} - \pi)(E[Y_i(d)|h(X_i)] - E[Y_i(d)]) \\
C_{n,d} &=  \sqrt{|\mathcal{D}|}\frac{1}{\sqrt{J_n}}\sum_{1 \le i \le J_n}\pi(E[Y_i(d)|h(X_i)] - E[Y_i(d)])~,
\end{align*}
where $\pi := \frac{1}{|\mathcal{D}|}$. Re-writing each of these terms using the fact that \[E[Y_i(d)|h(X_i)] = \sum_{1 \le s \le S}E[Y_i(d)|h(X_i)]I\{h(X_i) = s\} = \sum_{1 \le s \le S}E[Y_i(d)|h(X_i) = s]I\{h(X_i) = s\}~,\] we obtain
\begin{align*}
A_{n,d} &= \sqrt{|\mathcal{D}|}\sum_{1 \le s \le S}\frac{1}{\sqrt{J_n}}\sum_{1\le i\le J_n}(E[Y_i(d)|h(X_i)] - E[Y_i(d)])I\{D_i = d, h(X_i) = s\} \\
B_{n,d} &=  \sqrt{|\mathcal{D}|}\sum_{1 \le s \le S}(E[Y_i(d)|h(X_i) = s] - E[Y_i(d)])\frac{J_{n}(s)}{J_n}\sqrt{J_n}\left(\frac{J_{n,d}(s)}{J_{n}(s)} - \pi\right)\\
C_{n,d} &= \sqrt{|\mathcal{D}|} \sum_{1 \le s \le S}\pi(E[Y_i(d)|h(X_i) = s] - E[Y_i(d)])\sqrt{J_n}\left(\frac{J_{n}(s)}{J_n} - p(s)\right)~,
\end{align*}
where $J_{n}(s) = \sum_{1 \le i \le J_n}I\{h(X_i) = s\}$, $J_{n,d}(s) = \sum_{1 \le i \le J_n}I\{h(X_i) = s, D_i = d\}$, $p(s) = P(h(X_i) = s)$, and importantly for $C_{n,d}$ we have used the fact that
\[\sum_{1 \le s \le S}(E[Y_i(d)|h(X_i) = s] - E[Y_i(d)])p(s) = 0~,\]
which follows by the law of iterated expectations. By the law of large numbers, $J_n(s)/J_n \xrightarrow{p} p(s)$, and by the properties of stratified block randomization (see Example 3.4 in \cite{bugni2018}),
\[\sqrt{J_n}\left(\frac{J_{n,d}(s)}{J_{n}(s)} - \pi\right) \xrightarrow{p} 0~,\]
and hence we can conclude that $B_{n,d} \xrightarrow{p} 0 $ for every $d \in \mathcal{D}$. Using Lemma C.1. in \cite{bugni2019inference}, it can then be shown that 
\[\left(\begin{array}{c}
A_n \\ 
C_n
\end{array}\right) \xrightarrow{d} N\left(0, \begin{bmatrix} \mathbb V_{h,1} & 0\\ 0 & \mathbb V_{h,2} \end{bmatrix}\right)~,
\]
and hence the first result follows. Next, let $\nu$ be a $1 \times |\mathcal D|$ vector of constants, then it can be shown that
\[ \nu \mathbb V \nu' = \sum_{d \in \mathcal D} \nu_d^2 \mathrm{Var}[Y_i(d)] - \sum_{d \neq d' \in \mathcal D} \frac{1}{|\mathcal D|} \var[\nu_d E[Y_i(d) | X_i] - \nu_{d'} E[Y_i(d') | X_i]]~, \]
and
\[\nu \mathbb V_h \nu' = \sum_{d \in \mathcal D} \nu_d^2 \mathrm{Var}[Y_i(d)] - \sum_{d \neq d' \in \mathcal D} \frac{1}{|\mathcal D|} \var[\nu_d E[Y_i(d) | h(X_i)] - \nu_{d'} E[Y_i(d') | h(X_i)]]~.\]
It then follows from similar arguments to those used in the proof of Theorem C.2 of \cite{bai2022optimality} that $\nu \mathbb V \nu' \leq \nu \mathbb V_h \nu'$. In particular, note that 
\begin{align*}
& \var[\nu_d E[Y_i(d) | X_i] - \nu_{d'} E[Y_i(d') | X_i]] \\
& = E[(\nu_d E[Y_i(d) | X_i] - \nu_{d'} E[Y_i(d') | X_i] - (\nu_d E[Y_i(d)] - \nu_{d'} E[Y_i(d')]))^2] \\
& = E[(\nu_d E[Y_i(d) | X_i] - \nu_{d'} E[Y_i(d') | X_i] - (\nu_d E[Y_i(d) | h(X_i)] - \nu_{d'} E[Y_i(d') | h(X_i)]) \\
& \hspace{3em} + (\nu_d E[Y_i(d) | h(X_i)] - \nu_{d'} E[Y_i(d') | h(X_i)]) - (\nu_d E[Y_i(d)] - \nu_{d'} E[Y_i(d')]))^2] \\
& =  E[(\nu_d E[Y_i(d) | X_i] - \nu_{d'} E[Y_i(d') | X_i] - (\nu_d E[Y_i(d) | h(X_i)] - \nu_{d'} E[Y_i(d') | h(X_i)]))^2] \\
& \hspace{3em} + E[(\nu_d E[Y_i(d) | h(X_i)] - \nu_{d'} E[Y_i(d') | h(X_i)]) - (\nu_d E[Y_i(d)] - \nu_{d'} E[Y_i(d')]))^2]~,
\end{align*}
where the last equality follows because
\begin{align*}
    & E[(\nu_d E[Y_i(d) | X_i] - \nu_{d'} E[Y_i(d') | X_i] - (\nu_d E[Y_i(d) | h(X_i)] - \nu_{d'} E[Y_i(d') | h(X_i)])) \\
    & \hspace{3em} ((\nu_d E[Y_i(d) | h(X_i)] - \nu_{d'} E[Y_i(d') | h(X_i)]) - (\nu_d E[Y_i(d)] - \nu_{d'} E[Y_i(d')]))] \\
    & = E[E[(\nu_d E[Y_i(d) | X_i] - \nu_{d'} E[Y_i(d') | X_i] - (\nu_d E[Y_i(d) | h(X_i)] - \nu_{d'} E[Y_i(d') | h(X_i)])) \\
    & \hspace{3em} ((\nu_d E[Y_i(d) | h(X_i)] - \nu_{d'} E[Y_i(d') | h(X_i)]) - (\nu_d E[Y_i(d)] - \nu_{d'} E[Y_i(d')])) | h(X_i)]] \\
    & = E[E[(\nu_d E[Y_i(d) | X_i] - \nu_{d'} E[Y_i(d') | X_i] - (\nu_d E[Y_i(d) | h(X_i)] - \nu_{d'} E[Y_i(d') | h(X_i)])) | h(X_i)] \\
    & \hspace{3em} ((\nu_d E[Y_i(d) | h(X_i)] - \nu_{d'} E[Y_i(d') | h(X_i)]) - (\nu_d E[Y_i(d)] - \nu_{d'} E[Y_i(d')])) ] \\
    & = 0~,
\end{align*}
where the last equality follows from the law of iterated expectations. We can thus conclude that the matched tuples design is asymptotically more efficient than the large stratum design, in the sense that the difference in variances between the large stratum and matched tuples designs, $\mathbb V_h - \mathbb V$, is positive semidefinite. \qed


\subsection{Proof of Theorem \ref{thm:matched-pair}}
To begin, note that
\[ \hat\Delta_{\nu_k,n} = \frac{1}{n}\sum_{1 \leq i \leq J_n}\sum_{d \in \mathcal{D}}I\{\iota_k(d) = +1\}  I \{D_i = d\} Y_i(d) - \frac{1}{n}\sum_{1 \leq i \leq J_n}\sum_{d \in \mathcal{D}}I\{\iota_k(d) = -1\}  I \{D_i = d\} Y_i(d)~. \]
Let $A_i, 1 \leq i \leq J_n$ denote a sequence of i.i.d.\ random vectors, each of which is a $K - 1$ vector of i.i.d.\ Rademacher random variables. Further assume they are independent of $Y^{(n)}(d), d \in \mathcal D$, $D^{(n)}$, and $X^{(n)}$. Define $\iota_{-k}(d)$ as the vector of all entries of $\iota(d)$ except the $k$th entry. Then, we consider the following ``averaged" potential outcomes over these $K-1$ factors defined as follows:
\begin{equation*}
\tilde{Y}_i(+1) := \sum_{d \in \mathcal{D}}I\{\iota_k(d) = + 1\}I\{\iota_{-k}(d) = A_i\} Y_i(d)
\end{equation*}
\begin{equation*}
\tilde{Y}_i(-1) := \sum_{d \in \mathcal{D}}I\{\iota_k(d) = -1\}I\{\iota_{-k}(d) = A_i\}Y_i(d)~.
\end{equation*}
With this notation, define
\[ \tilde\Delta_{\nu_k,n} = \frac{1}{n} \sum_{1 \leq i \leq J_n}I\{\iota_k(D_i)= +1\} \tilde Y_i(+1) - \frac{1}{n} \sum_{1 \leq i \leq J_n}I\{\iota_k(D_i) = -1\} \tilde Y_i(-1) ~.\]
It then follows from the definition of the factor $k$-specific design that $\tilde{\Delta}_{\nu_k,n}$ has the same distribution as $\hat \Delta_{\nu_k, n}$. To see it, note
\[ \frac{1}{n}\sum_{1 \leq i \leq J_n}\sum_{d \in \mathcal{D}} I\{\iota_k(d) = +1\} I \{D_i = d\} Y_i(d) = \frac{1}{n}\sum_{1 \leq i \leq J_n} \sum_{d \in \mathcal{D}} I\{\iota_k(D_i) = +1\} I\{\iota_k(d) = + 1\} I \{\iota_{-k}(D_i) = \iota_{-k}(d)\} Y_i(d) \]
and
\[ \frac{1}{n} \sum_{1 \leq i \leq J_n}I\{\iota_k(D_i)= +1\} \tilde Y_i(+1) = \frac{1}{n} \sum_{1 \leq i \leq J_n} \sum_{d \in \mathcal{D}} I\{\iota_k(D_i) = +1\} I\{\iota_k(d)= +1\}I\{\iota_{-k}(d) = A_i\} Y_i(d)\]
and $\iota_{-k}(D_i)$ and $A_i$ follow the same distribution independently of everything else.

Note $\tilde \Delta_{\nu_k, n} / 2^{K - 1}$ can be thought of as the difference-in-means estimator where the treatment has two levels $+1$ and $-1$ and the potential outcomes are $\tilde Y_i(+1)$ and $\tilde Y_i(-1)$. The conditions in Lemma S.1.4 in \cite{bai2021inference} can be verified straightforwardly and therefore we have
\[\sqrt{2^{K-1}n} \left( \frac{\hat{\Delta}_{\nu_k, n}}{2^{K-1}} - \frac{\Delta_{\nu_k}(Q)}{2^{K-1}}  \right)\stackrel{d}{\to} N(0, \mathbb V_{\nu_k, mp})~,\]
where
\begin{align*}
\mathbb V_{\nu_k, mp} &:=  E[\var[\tilde{Y}_i(+1) | X_i]] +  E[\var[\tilde{Y}_i(-1) | X_i]]\\
&+ \frac{1}{2} E[(E[\tilde{Y}_i(+1) | X_i] - E[\tilde{Y}_i(+1)] - (E[\tilde{Y}_i(-1) | X_i] - E[\tilde{Y}_i(-1)]))^2]~.
\end{align*}
Note that by Assumption \ref{as:D},
\begin{align*}
E[\tilde{Y}_i(+1) | X_i] &= E\left[\sum_{d \in \mathcal{D}}I\{\iota_k(d) = +1\}I\{\iota_{-k}(d) = A_i\} Y_i(d) \Bigg | X_i\right] \\
&= \frac{1}{2^{K-1}}\sum_{d \in \mathcal{D}}I\{\iota_k(d) = +1\} \Gamma_d(X_i)~.
\end{align*}
Therefore,
\begin{align*}
&\frac{1}{2}E[(E[\tilde{Y}_i(+1) | X_i] - E[\tilde{Y}_i(+1)] - (E[\tilde{Y}_i(-1) | X_i] - E[\tilde{Y}_i(0)]))^2] \\
&=\frac{1}{2}\cdot\frac{1}{2^{2(K-1)}} E\left[\left(\sum_{d \in \mathcal{D}}I\{\iota_k(d) = +1\} (\Gamma_d(X_i) -\Gamma_d) -  \sum_{d \in \mathcal{D}}I\{\iota_k(d) = -1\} (\Gamma_d(X_i) -\Gamma_d)\right)^2\right] \\
&=\frac{1}{2}\cdot \frac{1}{2^{2(K-1)}} E\left[\left(\nu_k^\prime (\Gamma_d(X_i)-\Gamma_d: d\in\mathcal{D}))\right)^2\right] \\
&= \frac{2^{K-1}}{2^{2(K-1)}} \nu_k^\prime E\left[\frac{1}{2^K} \cov[\Gamma_d(X_i), \Gamma_{d'}(X_i)]\right]_{d,d'\in \mathcal{D}} \nu_k \\
&=  \frac{1}{2^{(K-1)}} \nu_k^\prime \mathbb{V}_2 \nu_k~.
\end{align*}
Moreover,
\begin{align*}
&\var[\tilde{Y}_i(+1) | X_i] \\
& = \var\left[\sum_{d \in \mathcal{D}}I\{\iota_k(d) = + 1\}I\{\iota_{-k}(d) = A_i\} Y_i(d) \Bigg| X_i\right]\\
& = E\left[\var\left[\sum_{d \in \mathcal{D}}I\{\iota_k(d) = + 1\}I\{\iota_{-k}(d) = A_i\} Y_i(d) \Bigg| X_i, A_i\right]\Bigg|X_i\right] \\
& \hspace{3em} + \var\left[E\left[\sum_{d \in \mathcal{D}}I\{\iota_k(d) = + 1\}I\{\iota_{-k}(d) = A_i\} Y_i(d) \Bigg| X_i, A_i\right]\Bigg|X_i\right] \\
&= E\left[\sum_{d \in \mathcal{D}}I\{\iota_k(d) = + 1\}I\{\iota_{-k}(d) = A_i\} \var[Y_i(d) | X_i] \Bigg|X_i\right] + \var\left[\sum_{d \in \mathcal{D}}I\{\iota_k(d) = + 1\}I\{\iota_{-k}(d) = A_i\}\Gamma_{d}(X_i) \Bigg|X_i\right] \\
&= \frac{1}{2^{K - 1}}\sum_{d \in \mathcal{D}:\iota_k(d) = + 1} \var[Y_i(d) | X_i] + \frac{1}{2^{K - 1}} \sum_{d \in \mathcal D: \iota_k(d) = +1} \left (\Gamma_d(X_i) -  \frac{1}{2^{K - 1}} \sum_{d' \in \mathcal D: \iota_k(d') = +1} \Gamma_{d'}(X_i) \right )^2 \\
&=\frac{1}{2^{K-1}} \sum_{d \in \mathcal{D}:\iota_k(d) = +1} \left(\var[Y_i(d) | X_i] +\left( \Gamma_{d}(X_i) - \frac{1}{2^{K-1}}\sum_{d' \in\mathcal{D}:\iota_k(d') = +1} \Gamma_{d'}(X_i)\right)^2 \right)~.
\end{align*}
A similar calculation applies to $\var[\tilde{Y}_i(-1) | X_i]$. Finally,
\begin{align*}
\mathbb V_{\nu_k, mp} & = \frac{1}{2^{K-1}} \sum_{d \in \mathcal{D}} E[\var[Y_i(d) | X_i]] + \frac{1}{2^{K-1}} \nu_k' \mathbb{V}_2 \nu_k \\
& \hspace{3em} +\frac{1}{2^{K-1}} E \left [ \sum_{d \in \mathcal{D}:\iota_k(d) = + 1}\left(\Gamma_{d}(X_i) - \frac{1}{2^{K-1}}\sum_{d \in \mathcal{D}:\iota_k(d) = + 1} \Gamma_d(X_i)\right)^2 \right ] \\
& \hspace{3em} +\frac{1}{2^{K-1}} E \left [ \sum_{d \in \mathcal{D}:\iota_k(d') = -1 }\left(\Gamma_{d'}(X_i) - \frac{1}{2^{K-1}}\sum_{d \in\mathcal{D}::\iota_k(d) = - 1} \Gamma_d(X_i)\right)^2 \right ]~.
\end{align*}
The conclusion therefore follows.
\qed

\section{Auxiliary Lemmas}
\begin{lemma} \label{lem:mean}
Suppose Assumptions \ref{as:Q}--\ref{as:close} hold. Then, for $r = 1, 2$,
\[ \frac{1}{n} \sum_{1 \leq i \leq |\mathcal D| n} Y_i^r(d) I \{D_i = d\} \stackrel{P}{\to} E[Y_i^r(d)]~. \]
\end{lemma}

\begin{proof}[\sc Proof of Lemma \ref{lem:mean}]
We prove the conclusion for $r = 1$ only and the proof for $r = 2$ follows similarly. To this end, write
\begin{multline*}
\frac{1}{n} \sum_{1 \leq i \leq |\mathcal D| n} Y_i(d) I \{D_i = d\} = \frac{1}{n} \sum_{1 \leq i \leq |\mathcal D| n} (Y_i(d) I \{D_i = d\} - E[Y_i(d) I \{D_i = d\} | X^{(n)}, D^{(n)}]) \\
+ \frac{1}{n} \sum_{1 \leq i \leq |\mathcal D| n} E[Y_i(d) I \{D_i = d\} | X^{(n)}, D^{(n)}]~.
\end{multline*}
Note
\[ \frac{1}{n} \sum_{1 \leq i \leq |\mathcal D| n} E[Y_i(d) I \{D_i = d\} | X^{(n)}, D^{(n)}] = \frac{1}{n} \sum_{1 \leq i \leq |\mathcal D| n} I \{D_i = d\} E[Y_i(d) | X_i] \stackrel{P}{\to} E[E[Y_i(d) | X_i]] = E[Y_i(d)]~, \]
where the equality follows from Assumption \ref{as:D} and the convergence in probability follows from Assumption \ref{as:close} and similar arguments to those used in the proof of Theorem \ref{thm:main_delta}. To complete the proof, we argue
\[ \frac{1}{|\mathcal D| n} \sum_{1 \leq i \leq |\mathcal D| n} (Y_i(d) I \{D_i = d\} - E[Y_i(d) I \{D_i = d\} | X^{(n)}, D^{(n)}]) \stackrel{P}{\to} 0~. \]
For this purpose, we proceed by verifying the uniform integrability condition in Lemma S.1.3 of \cite{bai2021inference} conditional on $X^{(n)}$ and $D^{(n)}$. Note for any $m > 0$ that
\begin{align*}
& \frac{1}{|\mathcal D| n} \sum_{1 \leq i \leq |\mathcal D| n} E[|Y_i(d) I \{D_i = d\} - E[Y_i(d) I \{D_i = d\} | X^{(n)}, D^{(n)}]) \\
&  \hspace{10em} \times I \{|Y_i(d) I \{D_i = d\} - E[Y_i(d) I \{D_i = d\} | X^{(n)}, D^{(n)}]| > m\} | X^{(n)}, D^{(n)}] \\
& = \frac{1}{|\mathcal D| n} \sum_{1 \leq i \leq |\mathcal D| n} E[|Y_i(d) I \{D_i = d\} - E[Y_i(d) | X_i] I \{D_i = d\}| \\
& \hspace{10em} \times I \{|Y_i(d) I \{D_i = d\} - E[Y_i(d) | X_i] I \{D_i = d\}| > m\} | X^{(n)}, D^{(n)}] \\
& \leq \frac{1}{|\mathcal D| n} \sum_{1 \leq i \leq |\mathcal D| n} E[|Y_i(d) - E[Y_i(d) | X_i]| I \{|Y_i(d) - E[Y_i(d) | X_i]| > m\} | X^{(n)}, D^{(n)}] \\
& = \frac{1}{|\mathcal D| n} \sum_{1 \leq i \leq |\mathcal D| n} E[|Y_i(d) - E[Y_i(d) | X_i]| I \{|Y_i(d) - E[Y_i(d) | X_i]| > m\} | X_i] \\
& \stackrel{P}{\to} E[|Y_i(d) - E[Y_i(d) | X_i]| I \{|Y_i(d) - E[Y_i(d) | X_i]| > m\}]~,
\end{align*}
where the first equality holds because of Assumption \ref{as:D}, the inequality holds because $0 \leq I \{D_i = d\} \leq 1$, the second equality holds because of Assumption \ref{as:D} again, and the convergence in probability follows from the weak law of large numbers because
\begin{multline*}
E[|Y_i(d) - E[Y_i(d) | X_i]| I \{|Y_i(d) - E[Y_i(d) | X_i]| > m\}] \leq E[|Y_i(d) - E[Y_i(d) | X_i]|] \\
\leq E[|Y_i(d)|] + E[|E[Y_i(d) | X_i]|] \leq E[|Y_i(d)|] + E[E[|Y_i(d)| | X_i]] = 2 E[|Y_i(d)|]~.
\end{multline*}
The proof could then be completed using the subsequencing argument as in (S.29) of the proof of Lemma S.1.5 of \cite{bai2021inference}.
\end{proof}

\begin{lemma} \label{lem:rho-dd'}
Suppose Assumptions \ref{as:Q}--\ref{as:close} hold. Then, $\hat \rho_n(d, d') \stackrel{P}{\to} E[\Gamma_d(X_i) \Gamma_{d'}(X_i)]$ as $n \to \infty$.
\end{lemma}

\begin{proof}[\sc Proof of Lemma \ref{lem:rho-dd'}]
To begin with, note
\begin{align*}
& E[\hat \rho_n(d, d') | X^{(n)}] \\
& = \frac{1}{n} \sum_{1 \leq j \leq n} \frac{1}{|\mathcal D| (|\mathcal D| - 1)} \sum_{\{i, k\} \subset \lambda_j} (\Gamma_d(X_i) \Gamma_{d'}(X_k) + \Gamma_d(X_k) \Gamma_{d'}(X_i)) \\
& = \frac{1}{n} \sum_{1 \leq j \leq n} \frac{1}{|\mathcal D| (|\mathcal D| - 1)} \sum_{\{i, k\} \subset \lambda_j} (\Gamma_d(X_i) \Gamma_{d'}(X_i) + \Gamma_d(X_k) \Gamma_{d'}(X_k) - (\Gamma_d(X_i) - \Gamma_d(X_k)) (\Gamma_{d'}(X_i) - \Gamma_{d'}(X_k)) ) \\
& = \frac{1}{|\mathcal D| n} \sum_{1 \leq i \leq |\mathcal D| n} \Gamma_d(X_i) \Gamma_{d'}(X_i) - \frac{1}{n} \sum_{1 \leq j \leq n} \frac{1}{|\mathcal D| (|\mathcal D| - 1)} \sum_{\{i, k\} \subset \lambda_j} (\Gamma_d(X_i) - \Gamma_d(X_k)) (\Gamma_{d'}(X_i) - \Gamma_{d'}(X_k)) \\
& \stackrel{P}{\to} E[\Gamma_d(X_i) \Gamma_{d'}(X_i)]~,
\end{align*}
where the convergence in probability follows from Assumptions \ref{as:Q}(c) and \ref{as:close}. To conclude the proof, we show
\begin{equation} \label{eq:rho-dd'0}
\hat \rho_n(d, d') - E[\hat \rho_n(d, d') | X^{(n)}] \stackrel{P}{\to} 0~.
\end{equation}
In order for this, we proceed to verify the uniform integrability condition in Lemma S.1.3 of \cite{bai2021inference} conditional on $X^{(n)}$. Define
\[ \hat \rho_{n, j}(d, d') = \Big ( \sum_{i \in \lambda_j} Y_i I \{D_i = d\} \Big ) \Big ( \sum_{i \in \lambda_j} Y_i I \{D_i = d'\} \Big )~. \]
In what follows, we repeatedly use the following inequalities:
\begin{align*}
I \Big \{ \Big |\sum_{1 \leq j \leq k} a_j \Big | > \lambda \Big \} & \leq \sum_{1 \leq j \leq k} I \Big \{ |a_j| > \frac{\lambda}{k} \Big \} \\
\Big |\sum_{1 \leq j \leq k} a_j \Big | I \Big \{ \Big |\sum_{1 \leq j \leq k} a_j \Big | > \lambda \Big \} & \leq \sum_{1 \leq j \leq k} k |a_j| I \Big \{ |a_j| > \frac{\lambda}{k} \Big \} \\
|a b| I \{|ab| > \lambda\} & \leq a^2 I \{|a| > \sqrt \lambda\} + b^2 I \{|b| > \sqrt \lambda\}~.
\end{align*}
We will also repeatedly use the facts that $0 \leq I \{D_i = d\} \leq 1$ and $I \{D_i = d\} I \{D_k = d\} = 0$ for $i \neq k$ in the same stratum. Note
\begin{align*}
& E[|\hat \rho_{n, j}(d, d') - E[\hat \rho_{n, j}(d, d') | X^{(n)}]| I \{|\hat \rho_{n, j}(d, d') - E[\hat \rho_{n, j}(d, d') | X^{(n)}]| > \lambda\} | X^{(n)}] \\
& \leq E \Big [ |\hat \rho_{n, j}(d, d')| I \Big \{ |\hat \rho_{n, j}(d, d')| > \frac{\lambda}{2} \Big \} \Big | X^{(n)} \Big ] + E \Big [ |E[\hat \rho_{n, j}(d, d') | X^{(n)}]| I \Big \{ |E[\hat \rho_{n, j}(d, d') | X^{(n)}]| > \frac{\lambda}{2} \Big \} \Big | X^{(n)} \Big ] \\
& = E \Big [ |\hat \rho_{n, j}(d, d')| I \Big \{ |\hat \rho_{n, j}(d, d')| > \frac{\lambda}{2} \Big \} \Big | X^{(n)} \Big ] + |E[\hat \rho_{n, j}(d, d') | X^{(n)}]| I \Big \{ |E[\hat \rho_{n, j}(d, d') | X^{(n)}]| > \frac{\lambda}{2} \Big \} \\
& \leq E \Big [ \Big | \sum_{i \in \lambda_j} Y_i(d) I \{D_i = d\} \sum_{i \in \lambda_j} Y_i(d') I \{D_i = d'\} \Big | I \Big \{ \Big | \sum_{i \in \lambda_j} Y_i(d) I \{D_i = d\} \sum_{i \in \lambda_j} Y_i(d') I \{D_i = d'\} \Big | > \frac{\lambda}{2} \Big \} \Big | X^{(n)} \Big ] \\
& \hspace{1.5em} + \Big | \frac{1}{|\mathcal D| (|\mathcal D| - 1)} \sum_{\{i, k\} \subset \lambda_j} (\Gamma_d(X_i) \Gamma_{d'}(X_k) + \Gamma_d(X_k) \Gamma_{d'}(X_i)) \\
& \hspace{1.5em} \times \Big | I \Big \{ \Big | \frac{1}{|\mathcal D| (|\mathcal D| - 1)} \sum_{\{i, k\} \subset \lambda_j} (\Gamma_d(X_i) \Gamma_{d'}(X_k) + \Gamma_d(X_k) \Gamma_{d'}(X_i)) \Big | > \frac{\lambda}{2} \Big \} \\
& \lesssim E \Big [ \sum_{i \in \lambda_j} Y_i^2(d) I \{D_i = d\} I \Big \{ \Big | \sum_{i \in \lambda_j} Y_i(d) I \{D_i = d\} \Big | > \sqrt{\frac{\lambda}{2}} \Big \} \Big | X^{(n)} \Big ] \\
& \hspace{1.5em} + E \Big [ \sum_{i \in \lambda_j} Y_i^2(d') I \{D_i = d'\} I \Big \{ \Big | \sum_{i \in \lambda_j} Y_i(d') I \{D_i = d'\} \Big | > \sqrt{\frac{\lambda}{2}} \Big \} \Big | X^{(n)} \Big ] \\
& \hspace{1.5em} + \sum_{\{i, k\} \subset \lambda_j} \Big ( |\Gamma_d(X_i) \Gamma_{d'}(X_k)| I \Big \{ |\Gamma_d(X_i) \Gamma_{d'}(X_k)| > \frac{\lambda}{2} \Big \} + |\Gamma_d(X_k) \Gamma_{d'}(X_i)| I \Big \{ |\Gamma_d(X_k) \Gamma_{d'}(X_i)| > \frac{\lambda}{2} \Big \} \Big ) \\
& \leq \sum_{i \in \lambda_j}  E \Big [ Y_i^2(d) I \Big \{ |Y_i(d)| > \sqrt{\frac{\lambda}{4}} \Big \} \Big | X_i \Big ] + \sum_{i \in \lambda_j} E \Big [ Y_i^2(d') I \Big \{ |Y_i(d')| > \sqrt{\frac{\lambda}{4}} \Big \} \Big | X_i \Big ] \\
& \hspace{1.5em} + \sum_{i \in \lambda_j} \Gamma_d^2(X_i) I \Big \{ |\Gamma_d(X_i)| > \sqrt{\frac{\lambda}{4}} \Big \} + \sum_{i \in \lambda_j} \Gamma_{d'}^2(X_i) I \Big \{ |\Gamma_{d'}(X_i)| > \sqrt{\frac{\lambda}{4}} \Big \}~.
\end{align*}
Therefore,
\begin{align*}
& \frac{1}{|\mathcal D| n} \sum_{1 \leq j \leq n} E[|\hat \rho_{n, j}(d, d') - E[\hat \rho_{n, j}(d, d') | X^{(n)}]| I \{|\hat \rho_{n, j}(d, d') - E[\hat \rho_{n, j}(d, d') | X^{(n)}]| > \lambda\} | X^{(n)}] \\
& \lesssim \frac{1}{|\mathcal D| n} \sum_{1 \leq i \leq |\mathcal D| n} E \Big [ Y_i^2(d) I \Big \{ |Y_i(d)| > \sqrt{\frac{\lambda}{4}} \Big \} \Big | X_i \Big ] + \frac{1}{|\mathcal D| n} \sum_{1 \leq i \leq |\mathcal D| n} E \Big [ Y_i^2(d') I \Big \{ |Y_i(d')| > \sqrt{\frac{\lambda}{4}} \Big \} \Big | X_i \Big ] \\
& \hspace{3em} + \frac{1}{|\mathcal D| n} \sum_{1 \leq i \leq |\mathcal D| n} \Gamma_d^2(X_i) I \Big \{ |\Gamma_d(X_i)| > \sqrt{\frac{\lambda}{4}} \Big \} + \frac{1}{|\mathcal D| n} \sum_{1 \leq i \leq |\mathcal D| n} \Gamma_{d'}^2(X_i) I \Big \{ |\Gamma_d(X_i)| > \sqrt{\frac{\lambda}{4}} \Big \} \\
& \stackrel{P}{\to} E \Big [ Y_i^2(d) I \Big \{ |Y_i(d)| > \sqrt{\frac{\lambda}{4}} \Big \} \Big ] + E \Big [ Y_i^2(d') I \Big \{ |Y_i(d')| > \sqrt{\frac{\lambda}{4}} \Big \} \Big ] \\
& \hspace{3em} + E \Big [ \Gamma_d^2(X_i) I \Big \{ |\Gamma_d(X_i)| > \sqrt{\frac{\lambda}{4}} \Big \} \Big ] + E \Big [ \Gamma_{d'}^2(X_i) I \Big \{ |\Gamma_d(X_i)| > \sqrt{\frac{\lambda}{4}} \Big \} \Big ]~,
\end{align*}
where the convergence in probability follows from the weak law or large numbers. Because $E[Y_i^2(d)] < \infty$, $E[Y_i^2(d')] < \infty$, $E[\Gamma_d^2(X_i)] \leq E[Y_i^2(d)] < \infty$, and $E[\Gamma_{d'}^2(X_i)] \leq E[Y_i^2(d')] < \infty$, we have
\begin{align*}
\lim_{\lambda \to \infty} E \Big [ Y_i^2(d) I \Big \{ |Y_i(d)| > \sqrt{\frac{\lambda}{4}} \Big \} \Big ] & = 0 \\
\lim_{\lambda \to \infty} E \Big [ Y_i^2(d') I \Big \{ |Y_i(d')| > \sqrt{\frac{\lambda}{4}} \Big \} \Big ] & = 0 \\
\lim_{\lambda \to \infty} E \Big [ \Gamma_d^2(X_i) I \Big \{ |\Gamma_d(X_i)| > \sqrt{\frac{\lambda}{4}} \Big \} \Big ] & = 0 \\
\lim_{\lambda \to \infty} E \Big [ \Gamma_{d'}^2(X_i) I \Big \{ |\Gamma_{d'}(X_i)| > \sqrt{\frac{\lambda}{4}} \Big \} \Big ] & = 0~.
\end{align*}
It follows from a subsequencing argument as in (S.29) of the proof of Lemma S.1.5 of \cite{bai2021inference} that \eqref{eq:rho-dd'0} holds. The conclusion therefore follows.
\end{proof}

\begin{lemma} \label{lem:rho-dd}
Suppose Assumptions \ref{as:Q}--\ref{as:close-4} hold. Then, $\hat \rho_n(d, d) \stackrel{P}{\to} E[\Gamma_d^2(X_i)]$ as $n \to \infty$.
\end{lemma}

\begin{proof}[\sc Proof of Lemma \ref{lem:rho-dd}]
For $1 \leq j \leq \frac{n}{2}$, define
\[ \hat \rho_{n, j}(d, d) = \sum_{i \in \lambda_{2j - 1}} Y_i I \{D_i = d\} \sum_{i \in \lambda_{2j}} Y_i I \{D_i = d\}~. \]
By defintition, $\hat \rho_n(d, d) = \frac{2}{n}\sum_{1 \leq j \leq \frac{n}{2}} \hat \rho_{n, j}$. Note by Assumption \ref{as:D},
\[ E[ \hat \rho_{n, j}(d, d) | X^{(n)}] = \frac{1}{|\mathcal D|^2} \sum_{i \in \lambda_{2j - 1}, k \in \lambda_{2j}} \Gamma_d(X_i) \Gamma_d(X_k)~. \]
Further note
\[ \Gamma_d(X_i) \Gamma_d(X_k) = \frac{1}{2} \Gamma_d^2(X_i) + \frac{1}{2} \Gamma_d^2(X_k) - \frac{1}{2} (\Gamma_d(X_i) - \Gamma_d(X_k))^2~. \]
Therefore,
\begin{align*}
& E[\hat \rho_n(d, d) | X^{(n)}] = \frac{2}{n} \sum_{1 \leq j \leq \frac{n}{2}} E[\hat \rho_{n, j}(d, d) | X^{(n)}] \\
& = \frac{2}{n} \sum_{1 \leq j \leq \frac{n}{2}} \frac{1}{|\mathcal D|^2} \sum_{i \in \lambda_{2j - 1}, k \in \lambda_{2j}} \Big ( \frac{1}{2} \Gamma_d^2(X_i) + \frac{1}{2} \Gamma_d^2(X_k) - \frac{1}{2} (\Gamma_d(X_i) - \Gamma_d(X_k))^2 \Big ) \\
& = \frac{1}{|\mathcal D| n} \sum_{1 \leq i \leq |\mathcal D| n} \Gamma_d^2(X_i) - \frac{1}{n |\mathcal D|^2} \sum_{1 \leq j \leq \frac{n}{2}} \sum_{i \in \lambda_{2j - 1}, k \in \lambda_{2j}} (\Gamma_d(X_i) - \Gamma_d(X_k))^2 \\
& \stackrel{P}{\to} E[\Gamma_d^2(X_i)]~,
\end{align*}
where the convergence in probability follows from Assumptions \ref{as:Q}(c) and \ref{as:close-4} as well as the weak law of large numbers. To conclude the proof, we show
\begin{equation} \label{eq:rho-dd0}
\hat \rho_n(d, d) - E[\hat \rho_n(d, d) | X^{(n)}] \stackrel{P}{\to} 0~.
\end{equation}
In order for this, we proceed to verify the uniform integrability condition in Lemma S.1.3 of \cite{bai2021inference} conditional on $X^{(n)}$. In what follows, we repeatedly use the following inequalities:
\begin{align*}
I \Big \{ \Big |\sum_{1 \leq j \leq k} a_j \Big | > \lambda \Big \} & \leq \sum_{1 \leq j \leq k} I \Big \{ |a_j| > \frac{\lambda}{k} \Big \} \\
\Big |\sum_{1 \leq j \leq k} a_j \Big | I \Big \{ \Big |\sum_{1 \leq j \leq k} a_j \Big | > \lambda \Big \} & \leq \sum_{1 \leq j \leq k} k |a_j| I \Big \{ |a_j| > \frac{\lambda}{k} \Big \} \\
|a b| I \{|ab| > \lambda\} & \leq a^2 I \{|a| > \sqrt \lambda\} + b^2 I \{|b| > \sqrt \lambda\}~.
\end{align*}
We will also repeatedly use the facts that $0 \leq I \{D_i = d\} \leq 1$ and $I \{D_i = d\} I \{D_k = d\} = 0$ for $i \neq k$ in the same stratum. Note
\begin{align*}
& E[|\hat \rho_{n, j}(d, d) - E[\hat \rho_{n, j}(d, d) | X^{(n)}]| I \{|\hat \rho_{n, j}(d, d) - E[\hat \rho_{n, j}(d, d) | X^{(n)}]| > \lambda\} | X^{(n)}] \\
& \leq E \Big [ |\hat \rho_{n, j}(d, d)| I \Big \{ |\hat \rho_{n, j}(d, d)| > \frac{\lambda}{2} \Big \} \Big | X^{(n)} \Big ] + E \Big [ |E[\hat \rho_{n, j}(d, d) | X^{(n)}]| I \Big \{ |E[\hat \rho_{n, j}(d, d) | X^{(n)}]| > \frac{\lambda}{2} \Big \} \Big | X^{(n)} \Big ] \\
& = E \Big [ |\hat \rho_{n, j}(d, d)| I \Big \{ |\hat \rho_{n, j}(d, d)| > \frac{\lambda}{2} \Big \} \Big | X^{(n)} \Big ] + |E[\hat \rho_{n, j}(d, d) | X^{(n)}]| I \Big \{ |E[\hat \rho_{n, j}(d, d) | X^{(n)}]| > \frac{\lambda}{2} \Big \} \\
& \leq E \Big [ \Big | \sum_{i \in \lambda_{2j - 1}} Y_i(d) I \{D_i = d\} \sum_{i \in \lambda_{2j}} Y_i(d) I \{D_i = d\} \Big | I \Big \{ \Big | \sum_{i \in \lambda_{2j - 1}} Y_i(d) I \{D_i = d\} \sum_{i \in \lambda_{2j}} Y_i(d) I \{D_i = d\} \Big | > \frac{\lambda}{2} \Big \} \Big | X^{(n)} \Big ] \\
& \hspace{1.5em} + \Big | \frac{1}{|\mathcal D|^2} \sum_{i \in \lambda_{2j - 1}, k \in \lambda_{2j}} \Gamma_d(X_i) \Gamma_d(X_k) \Big | I \Big \{ \Big | \frac{1}{|\mathcal D|^2} \sum_{i \in \lambda_{2j - 1}, k \in \lambda_{2j}} \Gamma_d(X_i) \Gamma_d(X_k) \Big | > \frac{\lambda}{2} \Big \} \\
& \lesssim E \Big [ \sum_{i \in \lambda_{2j - 1}} Y_i^2(d) I \{D_i = d\} I \Big \{ \Big | \sum_{i \in \lambda_{2j - 1}} Y_i(d) I \{D_i = d\} \Big | > \sqrt{\frac{\lambda}{2}} \Big \} \Big | X^{(n)} \Big ] \\
& \hspace{1.5em} + E \Big [ \sum_{i \in \lambda_{2j}} Y_i^2(d) I \{D_i = d\} I \Big \{ \Big | \sum_{i \in \lambda_{2j}} Y_i(d) I \{D_i = d\} \Big | > \sqrt{\frac{\lambda}{2}} \Big \} \Big | X^{(n)} \Big ] \\
& \hspace{1.5em} + \sum_{i \in \lambda_{2j - 1}, k \in \lambda_{2j}} |\Gamma_d(X_i) \Gamma_d(X_k)| I \Big \{ |\Gamma_d(X_i) \Gamma_d(X_k)| > \frac{\lambda}{2} \Big \} \\
& \leq \sum_{i \in \lambda_{2j - 1}}  E \Big [ Y_i^2(d) I \Big \{ |Y_i(d)| > \sqrt{\frac{\lambda}{2}} \Big \} \Big | X_i \Big ] + \sum_{i \in \lambda_{2j}} E \Big [ Y_i^2(d) I \Big \{ |Y_i(d)| > \sqrt{\frac{\lambda}{2}} \Big \} \Big | X_i \Big ] \\
& \hspace{1.5em} + \sum_{i \in \lambda_{2j - 1}} \Gamma_d^2(X_i) I \Big \{ |\Gamma_d(X_i)| > \sqrt{\frac{\lambda}{2}} \Big \} + \sum_{i \in \lambda_{2j}} \Gamma_d^2(X_i) I \Big \{ |\Gamma_d(X_i)| > \sqrt{\frac{\lambda}{2}} \Big \}~.
\end{align*}
Therefore
\begin{align*}
& \frac{1}{|\mathcal D| n} \sum_{1 \leq j \leq \frac{n}{2}} E[|\hat \rho_{n, j}(d, d) - E[\hat \rho_{n, j}(d, d) | X^{(n)}]| I \{|\hat \rho_{n, j}(d, d) - E[\hat \rho_{n, j}(d, d) | X^{(n)}]| > \lambda\} | X^{(n)}] \\
& \lesssim \frac{1}{|\mathcal D| n} \sum_{1 \leq i \leq |\mathcal D| n} E \Big [ Y_i^2(d) I \Big \{ |Y_i(d)| > \sqrt{\frac{\lambda}{2}} \Big \} \Big | X_i \Big ] + \frac{1}{|\mathcal D| n} \sum_{1 \leq i \leq |\mathcal D| n} \Gamma_d^2(X_i) I \Big \{ |\Gamma_d(X_i)| > \sqrt{\frac{\lambda}{2}} \Big \} \\
& \stackrel{P}{\to} E \Big [ Y_i^2(d) I \Big \{ |Y_i(d)| > \sqrt{\frac{\lambda}{2}} \Big \} \Big ] + E \Big [ \Gamma_d^2(X_i) I \Big \{ |\Gamma_d(X_i)| > \sqrt{\frac{\lambda}{2}} \Big \} \Big ]~,
\end{align*}
where the convergence in probability follows from the weak law or large numbers. Because $E[Y_i^2(d)] < \infty$ and $E[\Gamma_d^2(X_i)] \leq E[Y_i^2(d)] < \infty$, we have
\begin{align*}
\lim_{\lambda \to \infty} E \Big [ Y_i^2(d) I \Big \{ |Y_i(d)| > \sqrt{\frac{\lambda}{4}} \Big \} \Big ] & = 0 \\
\lim_{\lambda \to \infty} E \Big [ \Gamma_d^2(X_i) I \Big \{ |\Gamma_d(X_i)| > \sqrt{\frac{\lambda}{4}} \Big \} \Big ] & = 0~.
\end{align*}
It follows from a subsequencing argument as in the proof of Lemma S.1.5 of \cite{bai2021inference} that \eqref{eq:rho-dd0} holds. The conclusion therefore follows.
\end{proof}

\begin{lemma} \label{lem:fwl}
Suppose $(Y_i, X_{1, i}', X_{2, i}')'$, $1 \leq i \leq n$ is an i.i.d.\ sequence of random vectors, where $Y_i$ takes values in $\mathbf R$, $X_{1, i}$ takes values in $\mathbf R^{k_1}$, and $X_{2, i}$ takes values in $\mathbf R^{k_2}$. Consider the linear regression
\[ Y_i = X_{1, i}' \beta_1 + X_{2, i}' \beta_2 + \epsilon_i~. \]
Define $\mathbb X = (X_1, \ldots, X_n)'$, $\mathbb X_1 = (X_{1, 1}, \ldots, X_{1, n})'$, and $\mathbb X_2 = (X_{2, 1}, \ldots, X_{2, n})'$. Define $\mathbb P_2 = \mathbb X_2 (\mathbb X_2' \mathbb X_2)^{-1} \mathbb X_2'$ and $\mathbb M_2 = \mathbb I - \mathbb P_2$. Let $\hat \beta_{1, n}$ and $\hat \beta_{2, n}$ denote the OLS estimator of $\beta_1$ and $\beta_2$. Define $\hat \epsilon_i = Y_i - X_{1, i}' \hat \beta_{1, n} - X_{2, i}' \hat \beta_{2, n}$. Define
\[ \tilde {\mathbb X}_1 = \mathbb M_2 \mathbb X_1~. \]
Let
\[ \hat \Omega_n = (\mathbb X' \mathbb X)^{-1} (\mathbb X' \mathrm{diag}(\hat \epsilon_i^2: 1 \leq i \leq n) \mathbb X) (\mathbb X' \mathbb X)^{-1} \]
denote the heteroskedasticity-robust variance estimator of $(\hat \beta_{1, n}, \hat \beta_{2, n})$. Then, the upper-left $k_1 \times k_1$ block of $\hat \Omega_n$ equals
\[ (\tilde {\mathbb X}_1' \tilde {\mathbb X}_1)^{-1} (\tilde {\mathbb X}_1' \mathrm{diag}(\hat \epsilon_i^2: 1 \leq i \leq n) \tilde {\mathbb X}_1) (\tilde {\mathbb X}_1' \tilde {\mathbb X}_1)^{-1}~. \]
\end{lemma}

\begin{proof}[\sc Proof of Lemma \ref{lem:fwl}]
By the formula for the inverse of a partitioned matrix, the first $k_1$ rows of $(\mathbb X' \mathbb X)^{-1}$ equal
\[\begin{pmatrix}
(\mathbb X_1' \mathbb M_2 \mathbb X_1)^{-1} & - (\mathbb X_1' \mathbb M_2 \mathbb X_1)^{-1} \mathbb X_1' \mathbb X_2 (\mathbb X_2' \mathbb X_2)^{-1}
\end{pmatrix}~. \]
Furthermore,
\[ \mathbb X' \mathrm{diag}(\hat \epsilon_i^2: 1 \leq i \leq n) \mathbb X = \begin{pmatrix}
\mathbb X_1' \mathrm{diag}(\hat \epsilon_i^2: 1 \leq i \leq n) \mathbb X_1 & \mathbb X_1' \mathrm{diag}(\hat \epsilon_i^2: 1 \leq i \leq n) \mathbb X_2 \\
\mathbb X_2' \mathrm{diag}(\hat \epsilon_i^2: 1 \leq i \leq n) \mathbb X_1 & \mathbb X_2' \mathrm{diag}(\hat \epsilon_i^2: 1 \leq i \leq n) \mathbb X_2
\end{pmatrix}~. \]
The conclusion then follows from elementary calculations.
\end{proof}

\section{Additional Tables and Figures}\label{sec:additional-simulations}
\subsection{Power Plots}
In Section \ref{sec:sims-multcovs}, we presented truncated power plots for the first and third configurations in order to make the horizontal axes the same as that of the second power plot. Here we present plots showing the entire ``S'' shape of the power curves for \textbf{MT} and \textbf{MT2} under all three configurations. 

\begin{figure}[ht!]
\centering
\includegraphics[width=\textwidth]{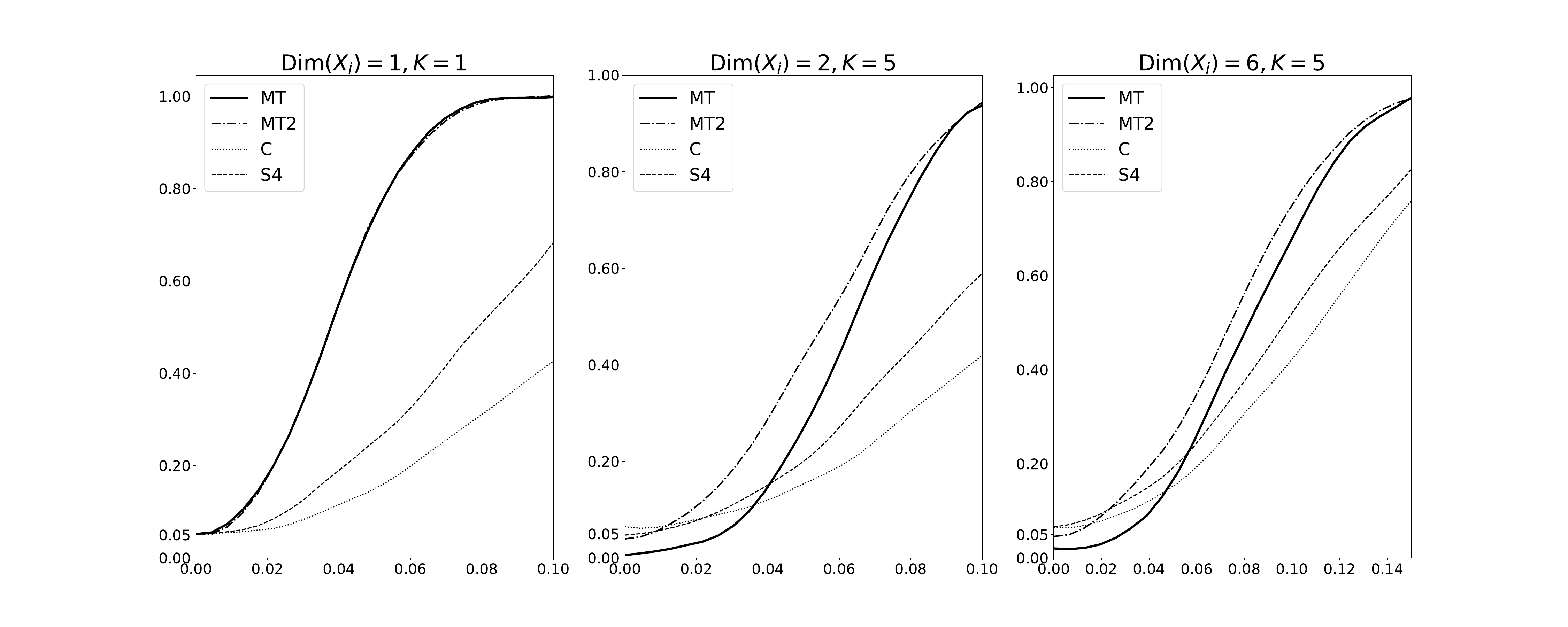}
\caption{Reject probability under various $\tau$s for the alternative hypothesis}
\label{fig:power_plots2}
\end{figure}

\subsection{Comparing Super Population and Finite Population Inference}\label{sec:fin_pop}
In this section, we compare the coverage properties of confidence intervals constructed using our proposed variance estimator versus two other well-known estimators, under both the super and finite population approaches to inference. First, we revisit the setting introduced in Section \ref{sec:sims-inference}, but now we consider only the matched-tuples design {\bf (MT)}, and construct confidence intervals for the parameter $\Delta_{\nu^1_{-1}}$ using one of three variance estimators:
\begin{enumerate}
    \item the variance estimator $\hat{\mathbb{V}}_{\nu,n}$ introduced in Section \ref{sec:main_tuple},
    \item a standard heteroskedasticity-robust variance estimator obtained from the regression in \eqref{eq:sfe}, and
    \item the block-cluster variance estimator considered in Theorem \ref{thm:bcve}.
\end{enumerate}
For the super population simulations, we generate the data as in Section \ref{sec:sims-inference}. For the finite population simulations, we simply use each DGP to generate the covariates and outcomes \emph{once}, and then fix these in repeated samples.

Table \ref{table:finite_population} presents coverage probabilities and average confidence interval lengths (in parentheses) with varying sample sizes, based on $2,000$ Monte Carlo replications. As expected given our theoretical results, $\hat{\mathbb{V}}_{\nu,n}$ delivers exact coverage in large samples under the super-population framework in all cases, whereas the robust variance estimator and BCVE are both generally conservative. In the finite population framework, we find that both $\hat{\mathbb{V}}_{\nu,n}$ and BCVE deliver exact coverage for some model specifications in large populations, but all three methods are generally conservative. $\hat{\mathbb{V}}_{\nu,n}$ displays some under-coverage in small populations relative to BCVE, but as the population size increases, $\hat{\mathbb{V}}_{\nu,n}$ generally produces narrower confidence intervals. 

Next, we repeat the above exercise using a calibrated simulation design analogous to that used in Section \ref{sec:sims-multcovs}, but utilizing the wave 6 data from \cite{McKenzie2014}. To construct our data generating process, we run an OLS regression of $Y_i$ on a constant and the seven covariates $X_i$ employed for matching, obtaining $\hat\beta$ and residuals $\hat\epsilon$. Subsequently, for $d \in \{0, 1, 2\}$ we compute $Y_i(d)$ based on the following model:
\begin{align*}
Y_i(d) = X_i^\prime \hat\beta + (X_i - \bar X_i)^\prime\hat{\beta}\cdot \gamma \cdot d + \epsilon_i~,
\end{align*}
with $X_i$ drawn from the empirical distribution of the data and $\epsilon_i \sim N(0, \text{var}(\hat\epsilon))$. Note that when $\gamma = 0$ we obtain a model with a constant treatment effect of zero, but that as $\gamma$ increases so does the amount of treatment effect heterogeneity. For the super-population simulations, the data is re-generated for each of the Monte Carlo replications. For the finite population simulations, the data is generated only \emph{once} and then fixed in repeated samples. In each experimental assignment we match the units into triplets and assign one unit to each of $d \in \{0, 1, 2\}$.

Table \ref{table:finite_population_empirical} presents coverage probabilities and average confidence interval lengths (in parentheses) for the parameter $\Delta_\nu = E[Y_i(1) - Y_i(0)]$, based on 2,000 Monte Carlo replications. Our first observation is that given the results for $\gamma = 0$, it is clear that the covariates $X_i$ explain little of the variation in experimental outcomes in our simulation design since all three variance estimators obtain exact coverage. However, as we artificially increase the amount of treatment effect heterogeneity by increasing the parameter $\gamma$, we find that, in line with our theoretical results, both the robust variance estimator and BCVE become slightly conservative. Moreover, in the finite population framework, $\hat{\mathbb{V}}_{\nu ,n}$ starts to become conservative as well.



\begin{table}[ht!]
\begin{adjustbox}{max width=\linewidth,center}
\begin{tabular}{cllllllllllll}
\toprule
\multicolumn{1}{l}{} & &  \multicolumn{5}{c}{Super Population}                     &  & \multicolumn{5}{c}{Finite Population}       \\ \cmidrule{3-7} \cmidrule{9-13}
Model                &     Method       & 4$n$=40  & 4$n$=80 & 4$n$=160 & 4$n$=480 & 4$n$=1000 &  & 4$n$=40  & 4$n$=80 & 4$n$=160 & 4$n$=480 & 4$n$=1000 \\
\midrule
\multirow{6}{*}{1}   & $\hat{\mathbb{V}}_{\nu,n}$ & 0.9340 & 0.9445 & 0.9435 & 0.9460 & 0.9470 && 0.9620 & 0.9550 & 0.9335 & 0.9445 & 0.9535\\
                                          &                            & (1.810) & (1.253) & (0.881) & (0.508) & (0.351) && (2.002) & (1.547) & (0.923) & (0.480) & (0.354)\\
                                          & Robust                     & 0.9855 & 0.9910 & 0.9930 & 0.9890 & 0.9920 && 0.9905 & 0.9895 & 0.9860 & 0.9950 & 0.9970\\
                                          &                            & (2.375) & (1.727) & (1.226) & (0.714) & (0.495) && (2.373) & (1.891) & (1.208) & (0.702) & (0.506)\\
                                          & BCVE                       & 0.9350 & 0.9470 & 0.9400 & 0.9455 & 0.9455 && 0.9185 & 0.9390 & 0.9405 & 0.9470 & 0.9525\\
                                          &                            & (1.821) & (1.262) & (0.885) & (0.509) & (0.351) && (1.822) & (1.475) & (0.938) & (0.483) & (0.354)\\\\
\multirow{6}{*}{2}   & $\hat{\mathbb{V}}_{\nu,n}$ & 0.9295 & 0.9395 & 0.9400 & 0.9525 & 0.9505 && 0.9495 & 0.9375 & 0.9405 & 0.9370 & 0.9520\\
                                          &                            & (1.897) & (1.299) & (0.896) & (0.509) & (0.352) && (1.829) & (1.309) & (0.848) & (0.505) & (0.354)\\
                                          & Robust                     & 0.9850 & 0.9905 & 0.9955 & 0.9965 & 0.9955 && 0.9870 & 0.9820 & 0.9970 & 0.9945 & 0.9980\\
                                          &                            & (2.489) & (1.809) & (1.290) & (0.751) & (0.522) && (2.337) & (1.560) & (1.354) & (0.749) & (0.540)\\
                                          & BCVE                       & 0.9185 & 0.9395 & 0.9415 & 0.9545 & 0.9515 && 0.9340 & 0.9395 & 0.9425 & 0.9415 & 0.9530\\
                                          &                            & (1.858) & (1.282) & (0.893) & (0.508) & (0.352) && (1.789) & (1.311) & (0.852) & (0.518) & (0.356)\\\\
\multirow{6}{*}{3}   & $\hat{\mathbb{V}}_{\nu,n}$ & 0.9445 & 0.9545 & 0.9600 & 0.9435 & 0.9450 && 0.9970 & 0.9790 & 0.9975 & 0.9890 & 0.9945\\
                                          &                            & (2.499) & (1.702) & (1.193) & (0.679) & (0.469) && (2.439) & (1.710) & (1.144) & (0.686) & (0.468)\\
                                          & Robust                     & 0.9800 & 0.9915 & 0.9920 & 0.9905 & 0.9910 && 1.0000 & 0.9985 & 1.0000 & 0.9995 & 1.0000\\
                                          &                            & (3.080) & (2.222) & (1.593) & (0.922) & (0.640) && (3.112) & (2.228) & (1.485) & (0.916) & (0.654)\\
                                          & BCVE                       & 0.9915 & 0.9940 & 0.9980 & 0.9960 & 0.9965 && 0.9995 & 0.9995 & 1.0000 & 1.0000 & 1.0000\\
                                          &                            & (3.748) & (2.578) & (1.811) & (1.032) & (0.714) && (3.766) & (2.628) & (1.729) & (1.015) & (0.709)\\\\
\multirow{6}{*}{4}   & $\hat{\mathbb{V}}_{\nu,n}$ & 0.9355 & 0.9480 & 0.9375 & 0.9445 & 0.9470 && 0.9310 & 0.9345 & 0.9540 & 0.9535 & 0.9640\\
                                          &                            & (1.889) & (1.319) & (0.927) & (0.534) & (0.371) && (1.674) & (1.292) & (1.015) & (0.562) & (0.373)\\
                                          & Robust                     & 0.9470 & 0.9680 & 0.9580 & 0.9635 & 0.9655 && 0.9435 & 0.9560 & 0.9695 & 0.9685 & 0.9770\\
                                          &                            & (1.931) & (1.406) & (1.005) & (0.584) & (0.406) && (1.751) & (1.410) & (1.085) & (0.599) & (0.407)\\
                                          & BCVE                       & 0.9550 & 0.9740 & 0.9700 & 0.9710 & 0.9750 && 0.9730 & 0.9760 & 0.9750 & 0.9760 & 0.9815\\
                                          &                            & (2.208) & (1.543) & (1.077) & (0.617) & (0.428) && (2.190) & (1.572) & (1.149) & (0.655) & (0.432)\\\\
\multirow{6}{*}{5}   & $\hat{\mathbb{V}}_{\nu,n}$ & 0.9315 & 0.9435 & 0.9495 & 0.9465 & 0.9530 && 0.9620 & 0.9615 & 0.9735 & 0.9625 & 0.9680\\
                                          &                            & (2.012) & (1.386) & (0.962) & (0.550) & (0.381) && (2.244) & (1.153) & (0.975) & (0.554) & (0.377)\\
                                          & Robust                     & 0.9530 & 0.9660 & 0.9790 & 0.9770 & 0.9850 && 0.9805 & 0.9870 & 0.9950 & 0.9870 & 0.9875\\
                                          &                            & (2.152) & (1.570) & (1.117) & (0.650) & (0.452) && (2.472) & (1.415) & (1.162) & (0.655) & (0.448)\\
                                          & BCVE                       & 0.9615 & 0.9730 & 0.9790 & 0.9785 & 0.9845 && 0.9610 & 0.9915 & 0.9930 & 0.9880 & 0.9870\\
                                          &                            & (2.419) & (1.667) & (1.155) & (0.662) & (0.458) && (2.506) & (1.530) & (1.151) & (0.656) & (0.453)\\\\
\multirow{6}{*}{6}   & $\hat{\mathbb{V}}_{\nu,n}$ & 0.9065 & 0.9290 & 0.9305 & 0.9425 & 0.9505 && 0.9105 & 0.9675 & 0.9655 & 0.9715 & 0.9665\\
                                          &                            & (4.730) & (3.361) & (2.388) & (1.388) & (0.961) && (4.846) & (3.244) & (2.233) & (1.425) & (1.025)\\
                                          & Robust                     & 0.9425 & 0.9600 & 0.9615 & 0.9660 & 0.9670 && 0.9625 & 0.9835 & 0.9855 & 0.9835 & 0.9765\\
                                          &                            & (5.001) & (3.624) & (2.606) & (1.521) & (1.055) && (5.392) & (3.449) & (2.437) & (1.549) & (1.090)\\
                                          & BCVE                       & 0.9560 & 0.9675 & 0.9660 & 0.9725 & 0.9735 && 0.9670 & 0.9875 & 0.9865 & 0.9865 & 0.9860\\
                                          &                            & (5.623) & (3.930) & (2.767) & (1.595) & (1.101) && (5.886) & (3.812) & (2.537) & (1.611) & (1.166)\\
\toprule
\end{tabular}
\end{adjustbox}
\caption{Coverage rate and average CI length (parentheses) under the super and finite population approaches to inference}
\label{table:finite_population}
\end{table}

\begin{table}[ht!]
\begin{adjustbox}{max width=\linewidth,center}
\begin{tabular}{cllllllllllll}
\toprule
\multicolumn{1}{l}{} & &  \multicolumn{5}{c}{Super Population}                     &  & \multicolumn{5}{c}{Finite Population}       \\ \cmidrule{3-7} \cmidrule{9-13}
Model                &     Method       & 3$n$=60 & 3$n$=120 & 3$n$=360 & 3$n$=750 & 3$n$=1200 &  & 3$n$=60 & 3$n$=120 & 3$n$=360 & 3$n$=750 & 3$n$=1200 \\
\midrule
\multirow{6}{*}{$\gamma=0$}   & $\hat{\mathbb{V}}_{\nu,n}$ & 0.949 & 0.943 & 0.946 & 0.946 & 0.952 & & 0.950 & 0.940 & 0.955 & 0.946 & 0.953 \\
&&(225.457) & (160.525) & (92.715) & (64.226) & (50.706) & & (225.896) & (159.946) & (92.607) & (64.235) & (50.771) \\
&Robust&0.950 & 0.943 & 0.950 & 0.947 & 0.952 & & 0.947 & 0.943 & 0.955 & 0.951 & 0.955 \\
&&(223.224) & (160.560) & (93.791) & (65.160) & (51.503) & & (224.081) & (160.511) & (93.731) & (65.128) & (51.553) \\
&BCVE&0.948 & 0.938 & 0.943 & 0.940 & 0.946 & & 0.953 & 0.944 & 0.954 & 0.943 & 0.950 \\
&&(229.461) & (162.261) & (92.762) & (64.198) & (50.674) & & (230.041) & (161.019) & (92.765) & (64.089) & (50.685) \\\\
\multirow{6}{*}{$\gamma=1$}   & $\hat{\mathbb{V}}_{\nu,n}$ & 0.940 & 0.946 & 0.953 & 0.960 & 0.959 & & 0.946 & 0.941 & 0.947 & 0.948 & 0.953 \\
&&(229.287) & (164.518) & (94.925) & (65.239) & (51.591) & & (233.870) & (165.423) & (94.580) & (65.390) & (51.554) \\
&Robust&0.936 & 0.955 & 0.961 & 0.970 & 0.963 & & 0.945 & 0.950 & 0.954 & 0.958 & 0.960 \\
&&(230.262) & (166.659) & (97.449) & (67.499) & (53.449) & & (232.131) & (167.113) & (97.281) & (67.482) & (53.420) \\
&BCVE&0.936 & 0.945 & 0.957 & 0.961 & 0.959 & & 0.949 & 0.946 & 0.950 & 0.950 & 0.956 \\
&&(232.063) & (165.622) & (95.388) & (65.468) & (51.662) & & (237.561) & (166.805) & (94.836) & (65.553) & (51.658) \\ \\              
\multirow{6}{*}{$\gamma=3$}   & $\hat{\mathbb{V}}_{\nu,n}$ & 0.947 & 0.949 & 0.963 & 0.966 & 0.957 & & 0.948 & 0.952 & 0.953 & 0.947 & 0.952 \\
&&(251.942) & (180.451) & (101.057) & (70.280) & (55.300) & & (253.653) & (177.162) & (102.184) & (70.042) & (55.324) \\
&Robust&0.961 & 0.962 & 0.978 & 0.977 & 0.975 & & 0.951 & 0.961 & 0.962 & 0.968 & 0.968 \\
&&(255.377) & (188.130) & (108.362) & (76.242) & (60.466) & & (257.964) & (185.413) & (109.376) & (75.993) & (60.422) \\
&BCVE&0.947 & 0.955 & 0.969 & 0.971 & 0.963 & & 0.958 & 0.957 & 0.954 & 0.959 & 0.961 \\
&&(256.837) & (185.391) & (103.913) & (72.470) & (57.259) & & (260.735) & (181.843) & (105.186) & (72.325) & (57.091) \\\\
\multirow{6}{*}{$\gamma=5$}   & $\hat{\mathbb{V}}_{\nu,n}$ & 0.945 & 0.947 & 0.966 & 0.964 & 0.957 & & 0.940 & 0.959 & 0.978 & 0.968 & 0.966 \\
&&(285.897) & (199.748) & (111.957) & (78.191) & (60.960) & & (284.327) & (200.163) & (113.900) & (77.267) & (60.890) \\
&Robust&0.959 & 0.965 & 0.986 & 0.981 & 0.977 & & 0.955 & 0.970 & 0.986 & 0.983 & 0.982 \\
&&(295.771) & (215.171) & (125.135) & (88.824) & (70.149) & & (293.489) & (215.318) & (127.164) & (88.177) & (70.040) \\
&BCVE&0.949 & 0.958 & 0.975 & 0.976 & 0.970 & & 0.949 & 0.962 & 0.981 & 0.975 & 0.975 \\
&&(296.164) & (209.731) & (119.286) & (83.916) & (65.873) & & (293.557) & (209.593) & (121.447) & (83.287) & (65.842) \\
\toprule
\end{tabular}
\end{adjustbox}
\caption{Coverage rate and average CI length (parentheses) under the super and finite population approaches to inference}
\label{table:finite_population_empirical}
\end{table}

\subsection{Calibrated Simulation Design Details}\label{sec:calibrated_details}
In this section we provide details for the calibrated simulation study considered in Section \ref{sec:sims-multcovs}. Following \cite{rubin2016}, we consider data obtained from the New York Department of Education, who were considering implementing a $2^5$ factorial experiment to study five new intervention programs: a quality review, a periodic assessment, inquiry teams, a school-wide performance bonus program and an online resource program; details about each of these programs can be found in \cite{dasgupta2015causal}. The data-set contains covariate information for $1,376$ schools. As in \cite{rubin2016}, we consider experimental designs constructed using nine covariates which were deemed likely to be correlated with schools' performance scores: total number of students, proportion of male students, enrollment rate, poverty rate, and five additional variables recording the proportion of students of various races. 

Since the NYDE has yet to run such an experiment, and given the limitations of the available dataset, we select one covariate (``number of teachers") from the original dataset to use as the potential outcome under control, and then construct the potential outcomes under the various treatment combinations using the model described in Section \ref{sec:sims-multcovs}. Specifically, we first demean and standardize all 9 covariates (denoted $\tilde{X}_i$), and then estimate a parameter vector $\beta$ by ordinary least squares in the following linear model specification for $Y_i(-1, -1, \ldots, -1)$:
\begin{equation}\label{eq:regression}
    Y_i(-1,-1, \ldots, -1) = \gamma_{(-1, -1, \ldots , -1)} \tilde{X}_i' \beta + \epsilon_{i}~,
\end{equation}
where $\gamma_{(-1, -1, \ldots, -1)} = - 1$ as defined in Section \ref{sec:sims-multcovs}. Table \ref{table:regression} presents the regression results. For each treatment combination $d$, we then compute $Y_i(d)$ using the model from Section \ref{sec:sims-multcovs} given by
\[Y_i(d) = \tau\cdot\left(d^{(1)} + \frac{\sum_{k=2}^K d^{(k)} }{K-1}\right) + \gamma_{d}\tilde{X}_i' \beta + \epsilon_{i}~,\]
where $\tilde{X}_i$ is drawn from the empirical distribution of the data and  $\epsilon_i \sim N(0, 0.1)$, where we note that $0.1$ is approximately equal to the sample variance of the residuals of the regression in \eqref{eq:regression}. 
\begin{table}
\centering
\begin{tabular}{lcccccc}
\toprule
               & \textbf{coef} & \textbf{std err} & \textbf{z} & \textbf{P$> |$z$|$} & \textbf{[0.025} & \textbf{0.975]}  \\
\midrule
\textbf{constant}       &    2.824e-06  &        0.007     &     0.000  &         1.000        &       -0.014    &        0.014     \\
\textbf{Total}          &      -0.9808  &        0.016     &   -60.609  &         0.000        &       -1.012    &       -0.949     \\
\textbf{nativeAmerican} &       0.0374  &        0.054     &     0.699  &         0.485        &       -0.068    &        0.143     \\
\textbf{black}          &       2.9378  &        3.175     &     0.925  &         0.355        &       -3.285    &        9.160     \\
\textbf{latino}         &       2.6158  &        2.836     &     0.922  &         0.356        &       -2.942    &        8.174     \\
\textbf{asian}          &       1.6866  &        1.822     &     0.926  &         0.355        &       -1.884    &        5.258     \\
\textbf{white}          &       1.9064  &        2.150     &     0.887  &         0.375        &       -2.308    &        6.121     \\
\textbf{male}           &      -0.0379  &        0.007     &    -5.355  &         0.000        &       -0.052    &       -0.024     \\
\textbf{stability}      &       0.0045  &        0.007     &     0.636  &         0.525        &       -0.009    &        0.018     \\
\textbf{povertyRate}    &      -0.1818  &        0.011     &   -16.350  &         0.000        &       -0.204    &       -0.160     \\
\bottomrule
\end{tabular}
\caption{Model \eqref{eq:regression} OLS Regression Results}
\label{table:regression}
\end{table}

\subsection{More Results for the Empirical Application}
In this section we repeat our analysis for the data on long-term effects obtained through the final round (wave 7) of surveys from the original paper. For the analysis of long-term effects, we follow the same procedure as in the original paper, except we additionally drop the four groups with sizes ranging from 5 to 8. Note that the estimated effects are different for the fixed-effect regression. This is because, as in the analysis in the original paper, we do \emph{not} drop entire quadruplets from our dataset whenever one member of the quadruplet was missing due to non-response in the final survey round. 
\begin{table}[ht!]
\centering
\setlength{\tabcolsep}{4pt}
\label{table:application-wave7}
\begin{adjustbox}{max width=\linewidth,center}
\begin{tabular}{cccccccccccc}
\toprule
& & &   \multicolumn{5}{c}{All} & & High initial & & Low initial \\ \cmidrule{4-8}  \cmidrule{10-10} \cmidrule{12-12} 
&  &  &firms & & Males &  & Females & & Profit women & & Profit women \\ \cmidrule{4-4} \cmidrule{6-6} \cmidrule{8-8} \cmidrule{10-10} \cmidrule{12-12}
&  &  & (1) & & (2) &  & (3) & & (4) & & (5) \\
\midrule
& &Cash treatment & 18.02 & & 56.17 &  & -8.43 & & -15.32 & & -3.84 \\
OLS & &  & (29.66) & & (67.95) &  & (18.25) & & (38.99) & & (17.14) \\
without group& & In-kind treatment & 31.59 & & 62.02 &  & 4.63 & & 42.10 & & -13.40 \\
fixed effects& &  & (21.63) & & (40.60) &  & (20.97) & & (48.82) & & (16.08) \\
&& Cash=in-kind ($p$-val) & 0.680 & & 0.938 &  & 0.484 & & 0.171 & & 0.554 \\
\\
\multirow{5}{*}{Matched-Tuples} & & Cash treatment & 18.02 & & 56.17 &  & -8.43 & & -15.32 & & -3.84 \\
 & &  & (26.07) & & (60.09) &  & (17.25) & & (42.10) & & (16.60) \\
& & In-kind treatment & 31.59 & & 62.02 &  & 4.63 & & 42.10 & & -13.40 \\
& &  & (19.47) & & (39.02) &  & (18.57) & & (45.30) & & (14.32) \\
&& Cash=in-kind ($p$-val) & 0.641 & & 0.931 &  & 0.456 & & 0.147 & & 0.556 \\
\toprule
\end{tabular}
\end{adjustbox}
\caption{Point estimates and standard errors for testing the treatment effects of cash and in-kind grants
using different methods (wave 7)}
\end{table}

\clearpage

\bibliography{factorial}

\end{document}